\documentclass[11pt]{amsart}

\usepackage{amssymb,amsmath,amsthm} 
\usepackage[mathscr]{eucal}
\usepackage{tikz}
\usepackage{amscd}
\usepackage{mathtools}
\usepackage{verbatim}
\usepackage{bbm}
\usepackage[margin=1in]{geometry}

\usepackage{graphics}

\newtheorem{theorem}{Theorem}[section]
\newtheorem{lemma}[theorem]{Lemma}
\newtheorem{corollary}[theorem]{Corollary}
\newtheorem{remark}[theorem]{Remark}
\newtheorem{proposition}[theorem]{Proposition}
\newtheorem{exercise}[theorem]{Exercise}
\newtheorem{asm}{Assumption}
\theoremstyle{definition}

\DeclareMathOperator*{\argmin}{arg\,min}

\newcommand{\X}{\mathcal{X}}
\newcommand{\F}{\mathcal{F}}
\newcommand{\R}{\mathbb{R}}

\newcommand{\Q}{\mathbb{Q}}

\newcommand{\N}{\mathbb{N}}

\newcommand{\C}{\mathbb{C}}

\newcommand{\ve}{\varepsilon}

\renewcommand{\d}{\mathrm{d}}

\newcommand{\parr}[2]{\frac{\partial #1}{\partial #2}}

\newcommand{\pars}[2]{\frac{\partial^2 #1}{\partial #2 ^2}}
\newcommand{\floor}[1]{\left\lfloor #1 \right\rfloor}

\newcommand{\hP}{\widehat{P}}

\newcommand{\hg}{\hat{g}}

\DeclarePairedDelimiterX{\infdivx}[2]{(}{)}{%
  #1\;\delimsize\|\;#2%
}

\DeclarePairedDelimiter{\norm}{\lVert}{\rVert}
\newcommand{\inner}[1]{\left\langle #1 \right\rangle}
\renewcommand{\norm}[1]{\left\lVert #1\right\rVert}

\newcommand{\E}[1]{\mathrm{E}\left[ #1 \right]}

\renewcommand{\exp}[1]{\text{exp}\left( #1 \right)}
\renewcommand{\P}[1]{\mathrm{P}\left( #1 \right)}

\newcommand{\one}{\mathbf{1}}

\newcommand{\cd}{\overset{\mathrm{d}}{\rightarrow}}

\newcommand{\ca}{\overset{\mathrm{a.s.}}{\rightarrow}}
\newcommand{\cov}[1]{\text{Cov}\left( #1 \right)}
\newcommand{\var}[1]{\text{Var}\left( #1 \right)}

\newcommand{\EE}[2]{\mathrm{E}_{#1}\left[ #2 \right]}

\renewcommand{\H}{\mathcal{H}}

\newcommand{\supp}[1]{\text{supp}\left( #1 \right)}

\newcommand{\hf}{\hat{f}}

\newcommand{\ra}{\rightarrow}

\renewcommand{\phi}{\varphi} 

\newcommand\indep{\protect\mathpalette{\protect\independenT}{\perp}}
\def\independenT#1#2{\mathrel{\rlap{$#1#2$}\mkern2mu{#1#2}}}
\newcommand{\h}{\widehat}
\newcommand{\hGamma}{\widehat{\Gamma}}
\newcommand{\tg}{\tilde{g}}
\newcommand{\rs}{\rightsquigarrow}
\renewcommand{\tilde}{\widetilde}
\renewcommand{\hat}{\widehat}
\newcommand{\G}{\mathbb{G}}
\newcommand{\eqas}{\overset{\mathrm{a.s.}}{=}}
\begin{document}
\begin{abstract}
In a nonparametric instrumental regression model, we strengthen the conventional moment independence assumption towards full statistical independence between instrument and error term. This allows us to prove identification results and develop estimators for a structural function of interest when the instrument is discrete, and in particular binary. When the regressor of interest is also discrete with more mass points than the instrument, we state straightforward conditions under which the structural function is partially identified, and give modified assumptions which imply point identification. These stronger assumptions are shown to hold outside of a small set of conditional moments of the error term. Estimators for the identified set are given when the structural function is either partially or point identified. When the regressor is continuously distributed, we prove that if the instrument induces a sufficiently rich variation in the joint distribution of the regressor and error term then point identification of the structural function is still possible. This approach is relatively tractable, and under some standard conditions we demonstrate that our point identifying assumption holds on a topologically generic set of density functions for the joint distribution of regressor, error, and instrument. Our method also applies to a well-known nonparametric quantile regression framework, and we are able to state analogous point identification results in that context. 
\end{abstract}

\title{Nonparametric Identification and Estimation with Independent, Discrete Instruments} 
\author{Isaac Loh }
\address{Department of Economics, Northwestern University}
\email{isaacloh2015@u.northwestern.edu}

\maketitle

\tableofcontents

\thanks{\footnotesize The author is greatly indebted to Joel Horowitz for his helpful critique and invaluable input on this work.  }

\section{Introduction} 

In this paper we consider the identification and estimation of a structural function $g$, which satisfies the relation
\begin{align*}
Y = g(X) + U
\end{align*}
in the case where $Y$ and $X$ are observable random variables and $U$ is unobserved. We are concerned with the case where the regressor $X$ is possibly endogenous so that $\E{U|X} \neq 0$. This complicates estimation of $g$, as one does not have the relation $\E{Y|X} = g(X)$, whereby $g(X)$ is identified and may be estimated by a broad range of kernel estimators. The typical approach to this problem is to introduce an instrumental variable $W$ which satisfies: 
\begin{align*}
\E{U|W} = 0.
\end{align*}
When $W$ and $X$ are both continuous random variables this problem has been studied extensively and $g$ can be estimated as the solution to an ill-posed inverse problem, cf.\ \cite{HH2005} and \cite{NP2003} among others. However, instruments in the applied literature are often discrete with few mass points. For instance, \cite{AK1991} use season of birth to instrument in a linear model for the number of years of education. See \cite{FH2015} for several more instances in which researchers have used instruments $W$ with fewer mass points than $X$. Generally $g$ is not point identified if this is the case: see Proposition 1 of \cite{FH2015}. A common solution is to assume that $g$ is a linear function of $X$ which allows it to be identified and estimated with e.g.\ a binary instrument. However, this assumption is rarely justified in practice.

We deal first with the case in which $X$ has $K$ mass points and $W$ has $2$ mass points, with $2 < K$. Most of the results that we display for such binary instruments may be readily generalized for instruments which take on more than two mass points, i.e.\ when $W$ has $L$ mass points with $L < K$. For instance, one can restrict attention to a subpopulation on which $W$ has two points of support. To more conveniently treat this binary instrument we write, without loss of generality, $W \in \{0,1\}$. Our approach to this problem is motivated by an observation from \cite{D2014} that ``in specific econometric applications, the conditional mean assumption is typically established by arguing that the stronger independence assumption holds". In mathematical terms, one typically argues that $\E{U|W} = 0$ by making the stronger claim that in fact $U \indep W$. If this is held to be true, then for any sequence of integrable functions $\{f_m\}$, one has $\E{f_m(U) |W = 0 } = \E{f_m(U) | W = 1}$. This is similar to the observation made by \cite{P2017} that independence actually provides infinitely many moment restrictions which can be used to identify $g$. We show that when $X$ is discrete, only finitely many such moment conditions can be used to partially identify $g$. Later, when we consider a continuously distributed $X$, we use the full power of the independence assumption $U \indep W$. 

When $X$ has $K$ mass points we take $f_m: u \mapsto u^m$ to be the function raising $U$ to the $m^\text{th}$ integer power. In Section \ref{SS:ident} we show that considering the moment restriction $\E{U^m|W = 0} = \E{U^m|W = 1}$ for $m = 1, \ldots , K$ partially identifies the function $g$, which can be considered as a vector in $\R^K$ over the support of $X$, under a light relevance condition for $X$. Partial identification in this case means that $g$ is identified up to a set of size at most $K!$. The first results of this section, Theorem \ref{T:identified} and Corollary \ref{C:Bezout}, only require that the first $K$ moments of $U$ exist and be independent of $W$, which is a consequence of full independence $U \indep W$, provided that the moments of $U$ exist. In other words, we do not require full independence to prove these results. We also show in Theorem \ref{T:2} that, even under our relevance condition for $X$, it is possible that $g$ is not point identified even by the full independence assumption $U \indep W$. Indeed, we exhibit random variables $Y$ and $X$ and functions $g_1$ and $g_2$ such that $Y - g_1(X) \indep W$ and $Y - g_2(X) \indep W$. These counterexamples can be found under very stringent assumptions on the conditional distributions $X|W = 0$ and $X|W = 1$, suggesting that point identification is not possible unless restrictions are also placed upon the distribution of $U$. Section \ref{SS:pointid} takes the additional step making assumptions on the joint distribution of the vector $(X,W,U)$. It describes a condition on the conditional moments of $U$ which implies point identification. Proposition \ref{P:generic} demonstrates that this additional condition is fulfilled by the majority of joint distributions for $(X,W,U)$, in a sense to be described further on. Hence, imposing $\E{U^{m}|W = 0} = \E{U^m | W = 1}$ for $m = 1, \ldots, K+1$ is often enough to ensure that the $K$-vector $g$ is point identified. 

In Section \ref{S:estimate} we use our moment independence relations to present an estimator $\hg$ of the function $g$ when $g$ is point identified. The estimator is shown to be almost surely consistent under light conditions, including an identification condition for $g$. We then show that a modified version of $\hg$, which we denote $\tg$, is $\sqrt{n}$-consistent for $g$ under an invertibility condition for a $K \times K$ matrix, $V$, whose coefficients are polynomials in the conditional moments of $Y$ given $X$ and $W$. \cite{P2017} provided an efficient estimator for $g$ in the case that $X$ takes on finitely many values (as $g$ can be treated as a vector); however, the proposed estimator required the user to provide basis functions satisfying certain regularity and approximation properties for the distributions of $U$ and $W$. In contrast, our estimator $\tg$ requires the user to solve a multivariate polynomial system of fixed dimension with coefficients that are directly calculated from the data (solving multivariate polynomial systems is a well-studied problem and most mathematics packages include toolkits for this application, which are reviewed in\ \cite{C2005}) and then minimize an objective function over the solutions obtained, which usually constitute only a finite set in $\R^K$. Thus, it might be expected that our estimator comes with a reduced computational cost. Theorem \ref{T:clt} establishes asymptotic normality for the quantity $\sqrt{n} (\tg - g)$ with a limiting variance which can be estimated consistently from the data.  

In Section \ref{S:partial}, we address the case where the structural function $g$ is perhaps not point identified but only partially identified, and $X$ is either discrete or continuously distributed. Estimation of the identified set requires some uniform convergence results from empirical process theory, and our assumptions in this section are mostly standard therein (see \cite{VW1996}). Our main result in this section is Proposition \ref{P:donsker2}, which states some convergence properties of an estimator for the identified set of $g$ under some standard integrability assumptions on the covering numbers for the class to which $g$ belongs. It is shown that the estimator enjoys the property of (asymptotically) containing the identified set for $g$ and excluding any element not in the identified set. 

Section \ref{S:continuous} extends our analysis of identification to the case where $X$ is continuously distributed. Existing results on identification in our setting, such as the work contained in \cite{D2014} and more recently \cite{C2019}, have typically focused on local identification (that is, identification in some neighborhood of the true structural function $g$) because the operator which arises out of our independence assumption is nonlinear and difficult to characterize. The problem is similar to that faced when considering identification in quantile regression models like the one discussed in \cite{CH2005}, and is typically addressed by making a nonlinear completeness assumption which is highly intractable. We address this issue by linearizing the operator offered to us by the independence assumption with one higher dimension (see Theorem \ref{T:contident}), which allows for the application of the more traditional theory of linear maps. Our method allows us to give a sufficient condition for global identification which is comparable to a standard instrument completeness condition (Lemma \ref{L:compeleteinstrument}), and to show that this condition holds on dense and topologically generic sets (Corollaries \ref{C:densityapprox} and \ref{C:densityapprox2}) under certain commonplace assumptions. Our method also applies to identification in quantile regression models, and we prove a new result in \S \ref{SS:quantile} which augments the work done in \cite{C2005}. A discussion of our work and a comparison to existing results are included in \S \ref{SS:comparison}.  

All proofs are located in our Appendix, Section \ref{S:proofs}, which concludes.  

\section{Model, Discrete Case}

Suppose that we have the additively separable model 
\begin{align}
Y = g(X) + U \label{E:model}
\end{align}
and an instrument $W$ which is strongly exogenous in that either $\E{U^m|W = 0} = \E{U^m | W = 1}$ for certain values of $m$, to be specified our assumptions, or there is full statistical independence and $U \indep W$. Here we shall consider a subcase where $X$ and $W$ take on finitely many values; $X \in [K]$, where throughout we define $[K] \equiv \{1, \ldots, K\}$, and $W \in \{0,1\}$ so that we have a \textit{binary instrument}. When the support of $W$ contains more than two points, one can always partition $\supp{W}$ into two sets and define a new binary instrument to be an indicator of either partition element, adapting our assumptions to the constructed instrument. One helpful observation is full statistical independence provides us with a number of moment conditions which we can use to identify $g$: 
\begin{align}
\E{ (Y - g(X))^m | W = 0} = \E{ (Y - g(X))^m | W = 1} \label{E:nmoment}
\end{align}
for all $m \in \N$, provided that the moments exist. When the moments of $U$ determine its distribution, \eqref{E:nmoment} is in fact equivalent to independence of $U$ and $W$. As $X$ takes on finitely many values we may regard $g$ as a vector in $\R^K$. Thus, for the sake of brevity let $p_k(\ell) \equiv \P{X = k | W = \ell}$ and $g_k \equiv g(k)$. In this section, we will interchangeably use vector and function notation for functions $h$ over $\supp{X} = [K]$ in this manner. 
Then, with the law of iterated expectations we may rewrite \eqref{E:nmoment} as 
\begin{align}
\sum_{k = 1}^K \Big[p_k(0)\E{ (Y - g_k)^m | W = 0, X = k} - p_k(1) \E{ (Y - g_k)^m | W = 1, X =k } \Big] = 0 \label{E:nmoment2}
\end{align}
for all $m \in \N$ such that $\E{U^m}$ exists. The system \eqref{E:nmoment2} may be viewed as a degree $n$ multivariate polynomial in $g_1, \ldots , g_K$. Note for instance that 
\[\E{(Y - g_k)^m|W = \ell, X = k} = \sum_{j = 0}^m \binom{m}{j} \E{Y^j | W = \ell, X = k} (-g_k)^{m-j}
\]
Importantly, the coefficients of these polynomials (in particular, the conditional moments of $Y$ given the $W$ and $X$) may be estimated directly from the data. 

\section{Identification, Discrete Case} 
\subsection{Partial Identification}\label{SS:ident}

In this section we demonstrate that one can use the system \eqref{E:nmoment2} to obtain some conditions for the partial identification of $\{g_k\}$. The relatively straightforward assumptions we require as as follows: 
\begin{asm} \label{A:1}
\normalfont For any strict subset $J \subsetneq [K]$, $\P{X \in J | W = 0} \neq \P{X \in J| W = 1}$. In other words, $\sum_{k \in J} (p_k(0) - p_k(1))$ is nonvanishing in $J \subsetneq [K]$. 
\end{asm}
It is straightforward to see that identification of $g$ can fail when Assumption \ref{A:1} is not fulfilled, in the absence of any restrictions on the distribution of $U$, as we demonstrate in the following lemma (whose proof allows for a large degree of flexibility in the conditional distributions of $U$ given $X$ and $W$). 

\begin{lemma}\label{L:a1}
Suppose that $K \ge 2$ and for some subset $J \subsetneq K$ one has $\P{X \in J| W = 0} = \P{X \in J | W = 1}$. Then there exist conditional distributions for $U|_{X, W}$ such that the model \eqref{E:model} and restrictions $U \indep W$, $\E{U} = 0$ do not point identify $g$, and in fact the identified set for $g$ contains a continuum of elements. 
\end{lemma}

\bigskip

Lemma \ref{L:a1} illustrates that Assumption \ref{A:1} is fundamental to identification and partial identification of $g$. Note that the next assumption that we make is weaker than specifying full independence $U \indep W$ in the case that $|\E{U^m}| < \infty$ for $m \in [K]$. It provides $K$ (nonlinear, polynomial) equations that aid in partially identifying the $K$-vector $g$:  
\begin{asm}  \label{A:2}
\normalfont $\E{U} = 0$ and $\E{U^m| W} = \E{U^m} < \infty$ for $m = 1, \ldots , K$. 
\end{asm}

With these we may prove the following: 
\begin{theorem}\label{T:identified}
If Assumptions \ref{A:1} and \ref{A:2} hold then the set of possible solutions to the system of equations formed by \eqref{E:nmoment2} for $m = 1, \ldots , K$ and $\E{U} = 0$ in $\R^K$ is finite. 

In particular, the identified set of vectors $\{g_k\}_{k=1}^K$ is finite. 
\end{theorem}

A loose upper bound on the size of the identified set is then provided by the B\'{e}zout Bound of Algebraic Geometry.  
\begin{corollary}[B\'{e}zout's Theorem] \label{C:Bezout}
If Assumptions \ref{A:1} and \ref{A:2} hold, then $\{g_k\}_{k=1}^K$ is identified up to a set of size at most $K!$. In particular, the number of possible solutions to \eqref{E:nmoment2} for $n = 1, \ldots , K$ is bounded above by $K!$. 
\end{corollary}

One immediately asks whether it is possible to extend Assumption \ref{A:1} on the conditional distributions of the regressor $X$ conditional on $W$ to obtain point identification of the vector $\{g_k\}_{k=1}^K$. It turns out that this is not possible if the probability vectors $\{p_k(0)\}$ and $\{p_k(1)\}$ are presumed to be strictly positive, as we show next. 

\begin{theorem} \label{T:2}
Let $K \ge 3$ and $p_k(0), p_k(1) > 0$ for all $k \in [K]$. Then for every pair of probability vectors $\{p_k(0)\}_{k=1}^K, \{p_k(1)\}_{k=1}^K$, there are distributions for $Y$ and functions $g_1 \neq g_2: [K] \ra \R$ such that 
\begin{align*}
&Y - g_1(X) \indep W \\
&Y - g_2(X) \indep W.
\end{align*}
\end{theorem}

\subsection{Conditions for Point Identification when $X$ is Discrete} \label{SS:pointid}

Theorem \ref{T:2} implies that sometimes it is not possible to point identify the function $g: [K] \ra \R$ with only assumptions on the joint distribution of $X$ and $W$, such as Assumptions \ref{A:1} and \ref{A:2}. In this section we establish conditions which allow $g$ to be point identified and attempt to show that these conditions are fulfilled in a wide variety of settings. The general strategy is as follows: first, use the multivariate polynomial equations (with variables $\{g_k\}_{k \in [K]}$) stated in \eqref{E:nmoment2}, for $m \in [K]$ as well as the normalizing relation $\E{U } = 0$ to characterize the identified set down to a finite number of points in $\R^K$. Then, use the polynomial equation supplied by \eqref{E:nmoment2} with $m = K+1$ to show that generally only one of the points in the aforementioned finite set satisfies $\E{U^{K+1}| W = 0} = \E{U^{K+1}|W = 1}$. The motivation of this section is: \textit{independence of the first $K$ moments of $U$ from $W$ typically is enough to obtain partial identification of the vector $g$ (Theorem \ref{T:identified}), whereas independence of the first $K+1$ moments of $U$ is typically enough to point identify $g$. }

To begin, note that for a particular vector $h \in \R^K$ which satisfies \eqref{E:nmoment2} for some $m \in \N$ we may combine equations \eqref{E:model} and \eqref{E:nmoment2} to write equivalently: 
\begin{align}
\label{E:i1}
\sum_{k = 1}^K \Big[p_k(0)\E{ (U +\delta_k)^m | W = 0, X = k} - p_k(1) \E{ (U + \delta_k)^m | W = 1, X =k } \Big] = 0
\end{align}
where we have denoted $\delta_k \equiv h_k - g_k$. Fixing a joint distribution for $X, W$ which satisfies Assumption \ref{A:1}, Theorem \ref{T:identified} and Corollary \ref{C:Bezout} imply that the size of the set of possible solutions to \eqref{E:i1} for $m = 1, \ldots , K$ is bounded above by $K!$ for every fixed set of $U$-moment vectors of the form: 
\begin{align}\label{E:lowermoments}
S_{K-1} \equiv \{ \E{U^m |W = \ell, X = k}: \ell \in \{0,1\},k \in [K], m \in [K-1]\}.
\end{align}
Note that dependence on $\E{U^K| W = \ell, X = K}$ is suppressed because when \eqref{E:i1} is expanded with $m = K$, the terms involving $K^\text{th}$ moments of $U$ vanish by moment independence. A useful observation is that $S_{K-1}$, which may in this context be called a vector of lower moments of $U$, only depends on the condition moments of $U^m$ for $m \le K-1$ (not $K$). 

Denote the finite set of possible $\delta$ vectors permitted under Theorem \ref{T:identified} as a function of $S_{K-1}$ by $A(S_{K-1})$, i.e.\ define 
\begin{align*}
A(S_{K-1}) \equiv \left\{ h - g : \, h \text{ satisfies } \eqref{E:nmoment2} \text{ for }m \in [K] \text{ and } \E{h - g} = 0 \right\},
\end{align*}
where for notational ease we use the convention $\E{h - g} = \E{h(X) - g(X)}$. 
Fixing $S_{K-1}$ and thus $A(S_{K-1})$, we note that in order to satisfy \eqref{E:i1} for $m = K+1$ we must have the relation: 
\begin{align}
&\sum_{k=1}^K \Big[ p_k(0) \E{U^{K}|W = 0, X = k}  - p_k(1) \E{U^K|W = 1 , X = k} \Big] \delta_k \label{E:i2} + P \left( S_{K-1}, \delta_k\right) = 0,
\end{align}
where 
\small
\[
P \left( S_{K-1}, \delta_k\right) \equiv K^{-1}\sum_{k=1}^K \left[ \sum_{j = 0}^{K-1}  \binom{K+1}{j} \delta_k^{K+1 - j} \left(p_k(0)\E{U^j|W = 0, X = k } - p_k(1) \E{U^j|W = 1, X = k }\right) \right]
\]
\normalsize
is the term determined by $S_{K-1}$ and $\delta \in A(S_{k-1})$ in the binomial expansion of \eqref{E:i1}. Suppose that the system formed by \eqref{E:i1} for $m = 1 , \ldots , K$ and the normalization relation $\E{U} = 0$ does not identify the function $g$, so that $A(S_{k-1}) \supsetneq \{0\}$. If in addition the equation \eqref{E:i1} for $m = K+1$ does not identify $g$ then there exists $\delta \neq 0 \in A(S_{K-1}) \subset \R^K$ such that \eqref{E:i2} holds: that is, the existence of nontrivial $\delta$ implies that a specific linear relation must hold among the $K^\text{th}$ moments $\E{U^K|W = 1, X = k }$ whose coefficients are determined by the lower moments of $U$ which are contained in $S_{K-1}$. This leads to the following observation, that point identification is generically fulfilled when one considers the $K^\text{th}$ conditional moments of $U$ in a sense that we make precise below.

The main result of this section is Proposition \ref{P:generic}, which distills our remarks so far into a genericity result on conditions for point identification. Below we give a brief introduction to its result. 

A set $T_0 \subset T$ is commonly said to be meagre and its complement $T\setminus T_0$ generic if (topologically) $T_0$ is the countable union of relatively closed nowhere dense sets, or (measure-theoretically) there exists some well-behaved measure $\mu$ on $T$ such that $\mu(T_0) = 0$ and $\mu(T) > 0$. Fix a joint distribution for $X,W$ satisfying Assumptions \ref{A:1} and \ref{A:2} and let $S_{K-1}$ denote the vector of lower moments in \eqref{E:lowermoments}, assuming that the moments $\E{U^m}$ exist for $m \in [K+1]$. Consider the set $T$ which consists of all possible vectors of $K^\text{th}$ moments of $U$ that extend $S_{K-1}$ and the moment independence assumption $\E{U^K|W = 0} = \E{U^K| W = 1}$. In imposing the first requirement, we mean that there exist actual probability distributions with lower moments in $S_{K-1}$ and $K^\text{th}$ moments in $T$. Formally, $T$ is a set of $2K$-vectors that verifies 
\begin{align}
T \equiv \Big\{& v \in \R^{2K}: \,\sum_{k = 1}^K p_k(0) v_k = \sum_{k=1}^K p_k(1) v_{K + k}  \label{E:once}\\
&\text{and there exists random variables }V_k \text{ satisfying }: \nonumber \\
& \E{V_k^m} = \E{ U^m|W = \floor{(k-1)/K}, X = (k-1)\mod{K} + 1}, \nonumber\\
& \E{V_k^K} = v_k \text{ for all } k \in [2K] \text{ and } m \in [K-1]\Big\}. \nonumber
\end{align}
$T$ is a subset of a hyperplane of $\R^{2K}$ corresponding with the linear restriction in the first line of \eqref{E:once}. Furthermore, define $T_0$ to be the subset of moment vectors in $T$ under which the restrictions $\E{U^m|W = 0} = \E{U^m| W = 1}$ for $m \in [K+1]$ and $\E{U} = 0$ do not point identify $g$, i.e.\ those for which there exist a nonzero $\delta \in A(S_{K-1})$ such that $\E{(U + \delta(X))^m|W = 0} = \E{(U + \delta(X))^m|W = 1}$ for $m \in [K+1]$ and $\E{\delta(X) } = 0$. Note that such a $\delta$ satisfies \eqref{E:i1} for $m \in [K+1]$ and also $\E{\delta(X)} = 0$. The first finding of Proposition \ref{P:generic} is that $T$ is ``large" in the $(2K-1)$-dimensional hyperplane given by the linear restriction in \eqref{E:once} in that it has non-empty interior in this hyperplane and is in fact assigned infinite measure by the unique translation-invariant Haar measure (completely analogous to Lebesgue measure) on that plane. The second is that $T_0$ is ``small" in $T$: it consists of at most the intersection of finitely many $(2K-2)$-dimensional hyperplanes with $T$, which are each assigned Haar measure $0$. Hence, in light of both topological and measure theoretical conditions, point identification of $g$ can be regarded as generic in terms of the conditional $K^\text{th}$ moments of $U$. 

\begin{proposition}\label{P:generic}
Suppose $\E{|U|^{m}} < \infty$ and $\E{U^m|W = \ell} = \E{U^m}$ for $m \le K+1$ and that for all $\ell, k$, $U|_{W = \ell, X = k}$ has at least $\floor{(K-1)/2}+1$ points in its support.
Suppose Assumptions \ref{A:1} and \ref{A:2} hold and let $S_{K-1}$ be defined as in \eqref{E:lowermoments}. Let $T \subset \R^{2K}$ denote the set of possible vectors $\{ \E{U^K|W = \ell, X = k }: \ell \in \{0,1\}, X \in [K]\}$ which satisfy Assumption \ref{A:2} and extend $S_{K-1}$. Let $T_0\subset T$ denote the subset of possible $K^\text{th}$ moment vectors for which $\E{U^m| W = 0} = \E{U^m|W = 1}$ for $m \in [K+1]$ and $\E{U} = 0$ \textit{does not} point identify $g$ (i.e.\ there exists $0 \neq \delta \in A(S_{K-1})$ which satisfies \eqref{E:i1} for $m \in [K+1]$, and $\E{\delta(X)} = 0$). 

Then $T_0$ is contained in a finite intersection of translated subspaces of strictly lower dimension than $T$, and $T$ has nonempty interior in the hyperplane implied by the linear restriction \eqref{E:once}. In particular, the standard $2K$-dimensional Haar measure $\mu$ supported on the hyperplane corresponding with the linear restriction $\E{U^K| W = 0 } = \E{U^K|W = 1}$ satisfies $\mu(T) =\infty $ but $\mu(T_0) = 0$. 
\end{proposition}

\noindent\textbf{Remark}: The requirement that the conditional distributions of $U$ have a lower bounded number of points of support is implied if $U$ has continuous conditional distributions. This requirement is made to ensure that we can flexibly supply $U$ with higher order moments. 

\bigskip

As the elements $\delta \in A(S_{K-1})$ are solutions to nonlinear polynomial equations and $T_0$ is furthermore a function of those elements, there is no closed form expression for $T_0$ in terms of the $S_{K-1}$, which complicates the question of whether the structural function $g$ is point identified in any particular instance. However, Proposition \ref{P:generic} assures us that there are conditions which guarantee that the identified set is a singleton, and that these conditions are fulfilled quite often when, conditional on the vector of lower moments $S_{K-1}$, the $K^\text{th}$ conditional moment vectors $\{\E{U^K|W = \ell, X = k}: \ell \in \{0,1\}, k \in [K] \}$ arise from a prior distribution which is continuous with respect to Lebesgue measure on the hyperplane in $\R^{2K}$ that is consistent with the fundamental linear restriction in \eqref{E:once}. In fact, Fubini's theorem and Proposition \ref{P:generic} clearly imply that, if this is the case, $g$ is point identified with prior probability $1$. Further on, we provide complementary conditions for point identification of the function $g$ which involve the joint distributions of $X$ and $U$ given $W$, and show that these conditions are topologically generic (Proposition \ref{P:resid}). These conditions are specialized to the case where $X$ is continuously distributed, but they may readily be adapted to the discrete case heretofore considered.

\section{Estimation when $g$ is Point Identified} \label{S:estimate}

In this section we adapt our identification result to propose estimators of the function $g$ in \eqref{E:model}. 
We begin by making the observation that the polynomial system expressed in \eqref{E:nmoment2}, associated with the moment independence condition $\E{U^m|W} = \E{U^m}$ for $m = 1, \ldots , K+1$ and the mean independence assumption $\E{U} = 0$, point identify the function $g$ under conditions which are enumerated in Proposition \ref{P:generic} and shown therein to be commonplace. Moreover, these same polynomials can be estimated directly from the data. 

We make the following standard and straightforward assumption to ensure that the necessary laws of large numbers may be invoked. All asymptotic statements (e.g.\ almost surely, eventually) are made as the sample size $n$ goes to infinity. 

\begin{asm}\label{A:3}
$(X_i, W_i, Y_i)_{i=1}^n$ is an iid sample of $X, W, Y$. For all $\ell \in \{0,1\}$ and $k \in [K]$, $\E{|Y|^{K+1} |W = \ell, X =k } < \infty$. $\P{W = 0}, \P{W = 1} > 0$. 
\end{asm}

Recall that our identification result Theorem \ref{T:identified} centered on a system of polynomial equations \eqref{E:nmoment2}. The central insight of this result was that the solution to the polynomial system could be characterized completely by finitely many of the equations, rather than the infinitely many equations utilized by \cite{P2017}. Note that GMM is typically asymptotically efficient only if the true, efficient score function belongs to the closure of linear manifold spanned by the moment conditions, which typically requires the use of an asymptotically diverging number of moments (see \cite{CF2014}). Therefore, our approach, which only uses a bounded number of moment conditions, is not expected to be asymptotically efficient. The advantage of our proposed methods is that they are straightforward to implement, and involve mainly solutions of systems of polynomial equations, around which a significant theory has been developed. 

Our estimation strategy is based on estimating an empirical analogue to the system \eqref{E:nmoment2} and solving for the function $g$ on $[K] = \mathrm{supp}( X)$. As shorthand, we continue to variably denote $g$ as a function $g: [K] \ra \R$ and as a vector in $\R^K$. 
Hence, we may now define $\Gamma : \R^K \ra \R^{K+2}$ to be the following vector valued function: 
\begin{align*}
\Gamma(h) & \equiv \begin{pmatrix}
P_0(h) \\
P_1 (h) \\
\vdots \\
P_{K+1}(h)
\end{pmatrix}
\end{align*}
where $P_0(h) \equiv \sum_{k=1}^K p_k(0) \E{Y - h_k|W = 0, X = k }$ and for $m \ge 1$, $P_m(h)$ is given as the left side of \eqref{E:nmoment2}, which is to say 
\[
P_m(h) \equiv \E{ (Y - h(X))^m | W = 0} - \E{(Y - h(X))^m| W = 1}.
\]
Note that the Binomial theorem implies that for $m \ge 1$ we have 
\begin{align}
P_m(h) &\equiv \sum_{k=1}^K  p_k(0)\E{(Y - g_k)^m | W = 0, X = k } - p_k(1)\E{ (Y - h_k)^m | W = 1, X = k } \label{E:pm} \\
& = \sum_{k=1}^K \sum_{j=0}^m \binom{m}{j} \left(p_k(0) \E{ Y^j | W = 0, X = k } - p_k(1) \E{Y^j | W = 1, X = k } \right) (-h_k)^{m-j} \nonumber \\
& = \sum_{k=1}^K \sum_{j=0}^m \binom{m}{j} (Q_{j,k,0} - Q_{j,k,1})(-h_k)^{m-j}, \nonumber
\end{align}
where we make the denotations
\begin{align}
C_{j,\ell,k}& \equiv \E{ Y^j \one_{W = \ell, X = k } }\nonumber \\
Q_{j, \ell, k}& \equiv \E{Y^j| W = \ell, X = k } p_k(\ell) = \E{Y^j|W = \ell , X = k} \P{X = k| W = \ell} \nonumber \\
& = \E{ Y^j \one_{W = \ell, X = k } } \P{W = \ell}^{-1} \nonumber \\
& = C_{j,\ell,k} \P{W = \ell}^{-1} ; \label{E:Qform}
\end{align}
that is, $P_m(h)$ is indeed an $m^\text{th}$ degree multivariate polynomial in the coordinates of $h \in \R^K$. 

Recall that under Assumptions \ref{A:1} and \ref{A:2}, Corollary \ref{C:Bezout} implies that $\Gamma$ attains at most $K!$ zeros in $\R^K$, one of which corresponds to the true function $g$ represented in \eqref{E:model}. Consider estimation of $\Gamma$ with the empirical function $\h{\Gamma}$: 
\begin{align*}
\hGamma(h) & \equiv \begin{pmatrix}
\hP_0(h) \\
\hP_1 (h) \\
\vdots \\
\hP_{K+1}(h)
\end{pmatrix},
\end{align*}
where for $j = 0, \ldots ,K+1$, $\hP_j$ is the plug-in estimator of $P_j$ with coefficients $Q_{j,\ell,k}$ for $j \le K+1$ replaced with the plug-in estimator: 
\begin{align*}
\h{Q}_{j,\ell,k} & \equiv  n^{-1} \sum_{i =1}^n Y_i^j \one_{W_i = \ell, X_i = k }\left(n^{-1} \sum_{i = 1}^n \one_{W_i = \ell} \right)^{-1}\\
& \equiv \hat{C}_{j,\ell,k}  \left( n^{-1} \sum_{i = 1}^n \one_{W_i = \ell} \right)^{-1} 
\end{align*}
(where $\hat{C}_{j,\ell,k}$ is defined as indicated). Letting $\EE{n}{\cdot}$ denote expectation with respect to the empirical measure $\frac{1}{n}\sum_{i=1}^n \delta_{X_i}$, it is straightforward to verify that 
\[
\hP_m(h) = \EE{n}{(Y - h(X))^m | W = 0} - \EE{n}{(Y - h(X))^m|W = 1}.
\]
The resulting estimator $\hGamma$ enjoys almost sure uniform convergence to $\Gamma$ as a result of the type of functions (polynomial) under consideration. 
\begin{lemma}\label{L:unifconv}
Under Assumption \ref{A:3}, $\hGamma(h)$ converges uniformly almost surely to $\Gamma(h)$ over any bounded subset of $\R^K$. 
\end{lemma}

We now consider the problem of estimating $g$. We supplant Assumptions \ref{A:1} and \ref{A:2} with an identification assumption on $g$. Namely, we suppose that the system of polynomials $P_0, \ldots, P_{K+1}$ uniquely identify $g$. Conditions for this to be the case are discussed in Proposition \ref{P:generic}, and are shown to hold outside of a very small (Haar measure $0$) subset of cases. Some alternate conditions which are sufficient for point identification may be drawn from the results of Section \ref{S:continuous}, and in particular Theorem \ref{T:contident} and its corollaries. 

\begin{asm}[Identification] \label{A:4}\normalfont
The vector $g \in \R^K$ uniquely satisfies $\Gamma(g) = 0$. Moreover, $\norm{g} < R$ for some known constant $R$ (where $\norm{\cdot}$ denotes the Euclidean norm). 
\end{asm}

We now define a preliminary estimator of $g$ as (unsurprisingly)
\begin{align*}
\hg \equiv \argmin_{\substack{h \in \R^K:\\ \norm{h}  < R}} \norm{ \hGamma(h)}. 
\end{align*}
By Lemma \ref{L:unifconv} and a standard extremum estimation argument, the following is then true:
\begin{lemma}\label{L:clt}
Under Assumptions \ref{A:3} and \ref{A:4}, one has $\hg \ca g$. 
\end{lemma}
\noindent Standard arguments which establish the asymptotic normality of extremum estimators cannot be applied \textit{in situ} to the estimator $\hat{g}$ because the components of $\Gamma$ are not population moments but linear combinations of conditional moments.

We now turn our attention to a polynomial-based estimator of $g$ which has an asymptotic normality property. We denote this estimator by $\tilde{g}$. The estimator $\tilde{g}$ is estimated in two stages: first, solve a multivariate polynomial system of equations $\hat{\Lambda}(h)$ for $h$ within the parameter set, and then minimizing the objective function $\hat{\Gamma}(h)$ over the obtained solution set. The idea of the proof of asymptotic normality is to apply the implicit function theorem and make a Delta-method argument on the estimator. To this end, define the vector valued function $\Lambda: \R^K \ra \R^K$ by 
\begin{align*}
\Lambda(h) \equiv \begin{pmatrix}
P_0(h) \\ \vdots \\ P_{K-1}(h)
\end{pmatrix} \quad 
\h{\Lambda}(h) \equiv \begin{pmatrix}
\hP_0(h) \\ \vdots \\ \hP_{K-1}(h) 
\end{pmatrix}
\end{align*}
and for some fixed $R < \infty$, define the zero set $\mathcal{Z}_R$ of any function $\gamma: \R^K \ra \R^K$ to be: 
\begin{align*}
\mathcal{Z}_R(\gamma)\equiv \left\{ \begin{array}{l l}
\left\{ h \in \R^K: \, \norm{h}  \le R \text{ and }\gamma(h) = 0 \right\} & \text{ if such a solution } h \text{ exists} \\
\left\{ h \in \R^K : \, \norm{h}  \le R \right\} & \text{ otherwise }
\end{array}\right.
\end{align*}
For the $R$ used in Assumption \ref{A:4}, 
we now define our estimator $\tg$ to be 
\begin{align*}
\tg & = \argmin_{ \substack{h \in \mathcal{Z}_R(\h{\Lambda})}} \norm{ \hGamma(h)}. 
\end{align*}
\begin{remark}\normalfont
By Bernstein's theorem (see\ \cite{C2005} \S5, Theorem 5.4), a polynomial system of the form $\Gamma(h): \R^K \ra \R^K$ \textit{generically} admits $K! $ solutions in $(\C\setminus\{0\})^K$.\footnote{This may be calculated using standard formulae for mixed volume in terms of normal $K$-dimensional volume, and the formula for the volume of a typical simplex in $\R^K$ formed by the $K$ elementary vectors} Generic in this case means that there is a nonzero polynomial in the coefficients of the polynomials $P_0, \ldots , P_{K-1}$ such that the property holds whenever the nonzero polynomial is nonvanishing for $P_0, \ldots , P_{K-1}$. For any fixed $d \in \N$, it can be shown via induction on $d$ that the set of solutions (variety) for a multivariate polynomial on $\R^d$ has zero Lebesgue measure on $\R^d$. Hence, under the assumption that the coefficients of $P_1, \ldots , P_K$ avoid this zero-measure subset of $\R^d$, $d$ indicating the number of coefficients in those polynomials, then $\mathcal{Z}_R(\h{\Lambda})$ is nonempty for $R$ large enough. In our case this suggests that \textit{typically}, upon solving for $\Lambda$, the user should obtain at most $K!$ solutions.  
\end{remark}

\noindent Application of the implicit function theorem requires that we assume an invertibility condition on the Jacobian matrix of $\Lambda$. 

\begin{asm}\label{A:5} \normalfont
The $K \times K$ matrix $V$ whose $(m,k)^\text{th}$ coordinate is given by 
\begin{align*}
V_{m,k} &= \parr{P_{m-1}}{h_k} \Big|_{h = g} \\
& = \left\{ \begin{array}{ll}
\sum_{j=0}^{m-2}  \binom{m}{j} Q_{j,k,0} (m-j ) (-g_k)^{m-j - 1} & \text{ for } m = 1\\
\sum_{j=0}^{m-2}  \binom{m}{j}\left( Q_{j,k,1} - Q_{j,k,0} \right) (m-j ) (-g_k)^{m-j - 1}
 & \text{ for } m \ge 2
 \end{array}
 \right.
\end{align*}
for $m, k \in [K]$ is invertible. 
\end{asm}

If we had instead put $V_{m,k} = \parr{P_{m}}{h_k}\Big|_{h = g}$ in Assumption \ref{A:C3}, and thus omitted $P_0$, then $V$ would not be invertible; this stems from the fact that if $P_m(g) = 0$ for $m \ge 1$, then also $P_m(g + c \one_K) = 0$ for all $c \in \R$ (where $\one_K$ is the $K$-dimensional vector $(1,\ldots , 1)'$).

Letting $V$ denote the matrix of partial derivatives indicated in Assumption \ref{A:5}, define the functions $\Psi_m: \R^{K} \times \R^{2K^2 + 2}, \, m = 0, \ldots , K-1$ by 
\begin{align*}
\Psi_m(v,w) \equiv\left\{ \begin{array}{ll}
 \sum_{k=1}^K \sum_{j=0}^{1} \binom{m}{j} \left( w_{\iota(j, 0,k)}/w_{2K^2 + 1} \right) (-v_k)^{m-j} & \text{ if } m = 0\\
 \sum_{k=1}^K \sum_{j=0}^{m} \binom{m}{j} \left( w_{\iota(j,0,k)}/w_{2K^2 + 1} - w_{\iota(j, 1,k )}/ w_{2K^2 + 2} \right) (-v_k)^{m-j}& \text{ otherwise }
\end{array}
\right.
\end{align*}
where we define the indexing bijection $\iota: \{0 , \ldots , K - 1\} \times \{0,1\} \times \{1, \ldots , K\}  \ra \{1, \ldots , 2K^2 \}$ by $\iota(j,\ell,k) = 2Kj + 2k + \ell - 1$. Although the definition of $\Psi_m$ introduces substantial notational difficulty, comparison with \eqref{E:pm} reveals that $\Psi_m(v,w)$ just becomes $P_m$ with the proper choice of coefficient vector $w$. The difference from \eqref{E:pm} is that the definition of $\Psi_m$ allows the coefficients of the latter to differ from $Q_{j,k,\ell}$, as should be expected when those coefficients are estimated from data.

To continue, let $w^* \in \R^{2K^2 + 2}$ denote the vector satisfying $w_{\iota(j,\ell,k)}^* = \E{Y^j \one_{W = \ell, X = k }}$ for $j \in \{0, \ldots, K-1\}$, $k \in [K]$, and $\ell \in \{0,1\}$, and also $w^*_{2K^2 + 1} = \P{W = 0}, w^*_{2K^2 + 2} = \P{W = 1}$. By definition, $w^*$ is precisely the ``correct" choice of $w$ which equates $\Psi_m(\cdot, w)$ with $P_m(\cdot)$. Define $\Psi: \R^{K} \times \R^{2K^2 + 2}$ to be the vector valued function 
\begin{align*}
\Psi(v,w) = \begin{pmatrix}
\Psi_0(v,w) \\
\vdots \\
\Psi_{K-1}(v,w)
\end{pmatrix}
\end{align*}
and let $\Delta = D_w \Psi(h, w^*)|_{h = g}$ denote its Jacobian matrix (with respect to $w$) evaluated at $(h,w) = (g, w^*)$. Finally, let the $(2K^2 + 2) \times (2K^2 + 2)$ matrix $\Omega$ by 
\footnotesize
\begin{align*}
\Omega & = \begin{pmatrix}
\var{Y^0 \one_{W = 0, X = 1} } & \cov{ Y^0 \one_{W = 0, X = 1}, Y^0 \one_{W = 1 , X = 1}} & \cdots & \cov{Y^0 \one_{W = 0,X = 1}, \one_{W = 1}} \\
\cov{ Y^0 \one_{W = 0, X = 1}, Y^0 \one_{W = 0 , X = 0}} & \var{Y^0 \one_{W = 1, X = 1} } & \cdots &  \cov{Y^0 \one_{W = 0,X = 1}, \one_{W = 1}}\\
\vdots & \vdots & \ddots & \vdots  \\
\cov{Y^0 \one_{W = 0,X = 1}, \one_{W = 1}} & \cov{Y^0 \one_{W = 0,X = 1}, \one_{W = 1}} & \cdots & \var{\one_{W = 1}} 
\end{pmatrix}
\end{align*}
\normalsize
to be the covariance matrix of the random vector 
\small
\begin{align*}
\big( &Y^0 \one_{W = 0, X = 1}, Y^0 \one_{W = 0, X = 1}, Y^0 \one_{W = 0, X = 2}, Y^0 \one_{W = 1, X = 2}, \ldots,\\
 &\quad Y^{K-1} \one_{W = 0, X = K}, Y^{K-1} \one_{W = 1, X = K}, \one_{W = 0}, \one_{W = 1} \big)'.
\end{align*}
\normalsize

With our invertibility assumption, we may now state the following central limit theorem for the estimator $\tg$:

\begin{theorem}\label{T:clt}
Under Assumptions \ref{A:3}---\ref{A:5}, the convergence
\[
\sqrt{n}  (\tg - g) \cd N(0, \left((\mathrm{D}_w \omega(w^*))\,  \Omega \, (\mathrm{D}_w \omega(w^*))  '\right)^{-1/2} )  
\]
holds with $D_w \omega(w^*) = -V^{-1} \Delta$. 
Moreover, if $\E{Y^{2K - 2} \one_{ W = \ell, X = k}} < \infty$ for all $\ell \in \{0,1\}, X \in [K]$, then there exist consistent estimators
$(\h{\mathrm{D}_w \omega(w^*)})$ and $\h{\Omega}$, for $\mathrm{D}_w(\omega(w^*))$ and $\Omega$, respectively. 
\end{theorem}  

\section{Estimation when $g$ is Partially Identified} \label{S:partial}

In this section we suppose that $g$ is possibly only partially identified. In other words, we consider the case where $g \in \mathcal{H}$ for some space of continuous functions $\mathcal{H}$ over the support of $X$, $\mathcal{X} = \supp X$. We let $X$ be either discrete or continuously distributed in this section. Let $\norm{\cdot}_\infty$ indicate the supremum norm over $\mathcal{H}$, i.e.\ $\norm{h}_\infty = \sup_{x \in \X} |h(x)|$. Assume that
\begin{asm} \normalfont \label{A:nine}
The standing model $Y = g(X) + U$ \eqref{E:model} holds with $\supp W = \{0,1\}$, and 
\begin{enumerate}
\item 
$U \indep W$ and $\E{U} = 0$
\item 
$\P{W = 0} \P{W = 1} \neq 0$ 
\item 
$g \in \mathcal{H}$ and $\mathcal{H}$ is bounded in $\norm{\cdot}_\infty$
\item 
$\limsup_{m \ra \infty} \left(   \frac{\E{U^m}}{m!}\right)^{1/m} < \infty$. 
\end{enumerate}
\end{asm}
Parts (1) and (2) of the assumption is familiar, (2) restricts $g$ to lie in our class $\mathcal{H}$ (to be defined shortly) and (3) is a regularity condition on the moments of $U$ which is satisfied by probability distributions whose Fourier transform exists in a complex neighborhood of the origin, or alternately whose Laplace transform exists in a neighborhood of the origin. Under Assumption \ref{A:nine}, the partially identified set for $g$ is: 
\begin{align*}
\mathcal{H}_0 \equiv \{h \in \mathcal{H}: \, \E{Y - h(X)} = 0 \text{ and } (Y - h(X)) \indep W \}.
\end{align*}
A common object of interest is not the entire function $g$ but the value of $T g$, where $T$ is some complex-valued functional defined on $\mathcal{H}$. Then the partially identified set for $Lg$ is $L \mathcal{H}_0 = \{Lh: h \in \mathcal{H}_0\}$. Any set-valued estimator for $\mathcal{H}_0$ may readily extended to an estimator for $T\mathcal{H}_0$ by considering its image under $L$. 

Our first objective is to more tractably characterize the identified set $\mathcal{H}_0$, which we do in the following lemma: 

\begin{lemma}\label{L:unitinterval}
Under Assumption \ref{A:nine}, the characteristic function $\E{e^{it(Y - h(X))}}$ is holomorphic in a neighborhood of the origin in $\C$ for any $h \in \H$, and
\small
\begin{align*}
\mathcal{H}_0 &= \left\{ h \in \mathcal{H}: \, \E{Y - h(X)} = 0 \text{ and } \sup_{t \in [0,1]} \left| \E{ e^{it(Y - h(X))}| W = 0} - \E{ e^{it (Y - h(X)}|W = 1}\right|= 0 \right\}. 
\end{align*}
\normalsize
\end{lemma}

\noindent Lemma \ref{L:unitinterval} provides a convenient characterization of $\mathcal{H}_0$ which we exploit. Notably, we may restrict attention to only $t$ within a compact subset of $\R$. To employ some of the tools of empirical process theory, we now make the following assumption on the class $\mathcal{H}$. 
\begin{asm} \normalfont\label{A:ten}
The covering numbers $N(\ve, \mathcal{H}, \norm{\cdot}_\infty)$ satisfy the integrability condition 
\begin{align*}
\int_0^\infty \sqrt{ \log N(\ve, \mathcal{H}, \norm{\cdot}_\infty )} \, \mathrm{d} \ve < \infty.
\end{align*}
\end{asm}
Assumption \ref{A:ten} is a uniform bound on the entropy of $\mathcal{H}$. It is satisfied if, for example, $\mathcal{X}$ is a  bounded and convex subset of $\R^d$ with nonempty interior, and 
\begin{align} \label{E:holderclass}
\mathcal{H} = \left\{ h: 
\max_{|k| \le \alpha} \sup_{x \in \X}| D^k h(x)| + \max_{|k| = \alpha} \sup_{x,y \in \X} \frac{|D^k h(x) - D^k h(y)|}{\norm{x - y}}\le M
\right\}
\end{align}
for some constants $M$ and $\alpha > d/2$, where $k = (k_1, \ldots , k_d)$ is a multi-index of $d$ integers, $|k| = \sum_{j = 1}^d k_j$, and $D^k \equiv \frac{\partial^k}{\partial x_1^{k_1} \cdots \partial x_d^{k_d}}$ (see \cite{VW1996}, Theorem 2.7.1, which is more general and applies to H\"{o}lder-continuous derivatives as well as Lipschitz in the case where $\alpha$ is not an integer). In fact, \eqref{E:holderclass} implies that one has $\log N(\ve, \H, \norm{\cdot}_\infty) \le K \left( \frac{1}{\ve} \right)^{d/\alpha}$ for some constant $K$. \cite{S2012} gives conditions, in particular restrictions on the weighted Sobolev norms of functions in $\mathcal{H}$, which guarantee that this is the case. It is clear that if $\X$ is finite then Assumption \ref{A:ten} is immediately satisfied as long as $\mathcal{H}$ is bounded. 

Define the classes of functions from the sample space $\R \times \mathcal{X} \times \{0,1\}$ to $\C$ which we will subsequently consider by
\begin{align*}
&\mathcal{E} \equiv \left\{ f(y,x) = y - h(x): \, h \in \H\right\}\\
&\F \equiv \left\{ f(y,x,w) = \exp{ i t (y - h(x))} \one_{w = \ell} : \, t \in [0,1], h \in \mathcal{H}, \ell \in \{0,1\} \right\} 
\end{align*}
When the uniform entropy condition holds for $\mathcal{H}$ under $\norm{\cdot}_\infty$, we can infer that the class $\mathcal{F}$ obeys a Donsker theorem, in the sense that
\begin{lemma}\label{L:covering}
Let Assumptions \ref{A:nine} and \ref{A:ten} hold and let $P$ denote the distribution of $(Y,X,U,W)$. Then the classes $\mathcal{E}$ and $\mathcal{F}$ are Glivenko-Cantelli and $P$-Donsker. 
\end{lemma}

As we are dealing with classes of potentially complex valued random functions, we designate an empirical process $\mathbb{G}_n = \sqrt{n} ( \mathbb{P}_n - \mathbb{P})$ indexed by a class, say $\F$, $P$-Donsker if it converges in the weak sense to a complex valued Gaussian process $\mathbb{G}$ (that is, a process whose marginals are all joint complex Gaussian distributions) which takes on values in $L^\infty(\F)$. All of the standard results from empirical process theory apply to complex-valued functions, which can be readily seen by limiting consideration to the real and complex parts of functions in $\F$ individually, and then noting that the union of two Glivenko-Cantelli classes is clearly Glivenko-Cantelli, whilst the union of two $P$-Donsker classes is also $P$-Donsker (\cite{K2008}, Corollary 9.31). 

From Lemma \ref{L:covering}, conclude (continuing to use the convention that $0/0 = 0$ when calculating conditional expectations with respect to empirical measure $\EE{n}{\cdot}$): 
\begin{proposition}\label{P:donsker} Let $\F_0 \equiv [0,1] \times \H$. Then for all $t \in [0,1]$ and $h \in \H$, 
one has the convergence 
\begin{align} \label{E:y0}
&\sqrt{n} \Bigg[\EE{n}{e^{it(Y - h(X))}|W = 0}  - \EE{n}{e^{it(Y - h(X))}|W = 1} \\
& \qquad \qquad \qquad - \E{e^{it(Y - h(X))}|W = 0}  + \E{e^{it(Y - h(X))}|W = 1}  \Bigg]  \rs \mathbb{D}(t,h), \nonumber
\end{align}
where $\mathbb{D}$ is a tight mean-zero Gaussian process in $\ell^\infty(\F_0)$. 
\end{proposition}

As a consequence of Lemma \ref{L:covering} and Proposition \ref{P:donsker}, we are able to state an estimator for the identified set $\H_0$ by substituting for moment equalities with their finite sample analogues. Let $\eta_n \ra 0$ be a positive sequence, and define 
\begin{align*}
\hat{\H}_n = \Bigg\{ &h \in \mathcal{H}: \,| \EE{n}{Y - h(X)}| \le \eta_n \\
&\text{ and } \sup_{t \in [0,1]} \left| \EE{n}{ e^{it(Y - h(X))}| W = 0} - \EE{n}{ e^{it (Y - h(X)}|W = 1}\right|\le \eta_n \Bigg\}, 
\end{align*}
where $\EE{n}{\cdot}$ denotes the sample mean. Then we have the following convergence result: 
\begin{proposition} \label{P:donsker2} Under Assumptions \ref{A:nine} and \ref{A:ten},
if $\eta_n \ra 0$ and $\eta_n \sqrt{n} \ra \infty$ then for all sequences $(\alpha_n)_{n \in \N}$ satisfying $\liminf_{n \ra \infty} \frac{\alpha_n}{2 \eta_n } > 1$,
\begin{align*}
&\P{\H_0 \subset \hat{\H}_n} \ra 1 \text{ and } 
\P{\H_{\alpha_n} \cap \hat{\H}_n = \emptyset} \ra 1,
\end{align*}
where  \small
\[\H_\alpha = \left\{ h \in \mathcal{H}: \, \E{Y - h(X)} \ge \alpha \text{ or } \sup_{t \in [0,1]} \left| \E{ e^{it(Y - h(X))}| W = 0} - \E{ e^{it (Y - h(X))}|W = 1}\right| \ge \alpha \right\}. 
\]
\normalsize
Moreover, if $\eta_n \propto n^{-\gamma}$ and $\log N(\ve, \F, \norm{\cdot}_\infty) \le K \left( \frac{1}{\ve} \right)^\omega$ for some $\gamma \in (0,1/2)$ and $\omega \in (0,2)$ then 
\begin{align*}
\one_{\H_0 \subset \hat{\H}_n} \ca 1 \text{ and }  \one_{\H_{\alpha_n} \cap \hat{\H}_n = \emptyset} \ca 1.
\end{align*}
\end{proposition}
\noindent The set $\hat{\mathcal{H}}_n$ is thus a consistent estimator for $\mathcal{H}_0$ under the stated assumptions. Given that the bound $K!$ stated on the identified set in Theorem \ref{T:identified} is large, the reader may find it practical even in the discrete case to reduce consideration of $\mathcal{H}_0$ down to its image under a linear map $L$ as previously discussed, and then estimate it by $L \hat{\mathcal{H}}_n$ or a perhaps an interval containing that set.

\section{Identification Results when $X$ is a Continuous Random variable}\label{S:continuous}

We now turn to the case where $X$ is a continuous random variable, considering first a scalar $X$ and binary instrument $W$ satisfying our standing model \eqref{E:model} with $U \indep W$. We state conditions which allow the function $g$ to be point identified by the joint distribution of observables $(Y,X, W)$. Our main strategy here is to linearize the nonlinear operator imposed by the independence restriction, which allows us to more tractably state some sufficient conditions for identification. This avoids some of the difficulties encountered when dealing with strictly nonlinear operators (see for instance \cite{C2019}, \S 2, and \cite{CH2005}). Our approach, which is encapsulated in Theorem \ref{T:contident}, allows us to construct examples in which $g$ is point identified and then to prove that the conditions which enable point identification hold on a topologically generic set (Proposition \ref{P:resid}) of density functions. This result reinforces Proposition \ref{P:generic} in suggesting that our standing model with independence assumption might typically be enough to point identify $g$. Interestingly, our method generalizes to quantile regression models, and we are able to give a new identification result in that setting which extends the result of \cite{CH2005} (\S \ref{SS:quantile}). A discussion of our results is given in \S \ref{SS:comparison}.

The main assumption that we will make in this section (Assumption \ref{A:C3}) is in the familiar form of a restriction on the kernel of a linear operator, which is determined by the joint distribution of observables. To begin, consider the following base assumptions, which are regularity conditions on the joint distribution of $X$ and $U$:

\begin{asm}\normalfont \label{A:C1}
$X$ is supported on a compact set $\mathcal{X} \subset \R^{d_X}$, $d_X \in \N$. Moreover, $g \in C(\mathcal{X})$ and $\norm{g}_\infty =\sup_{x \in \mathcal{X}}|g(x)| \le B$ for some known constant $B \le \infty$. 
\end{asm}

\begin{asm} \normalfont \label{A:C2}
The joint distribution of $(X,U)|W = w$ is continuous with respect to (product) Lebesgue measure $\lambda$ on $\mathcal{X} \times \R$ with continuous densities $f_\ell(x,u)\equiv f_{W = \ell}(x,u)$, for $\ell \in \{0,1\}$.
\end{asm}

Assumption \ref{A:C1} is used mainly to provide notational simplicity and to ensure the existence and convergence of certain integrals. One could most expediently address the case of an $X$ with unbounded support by considering the pushforward of $X$ by a continuous and invertible mapping $\psi$, say the probability integral transform. If $\psi$ maps $\supp{X}$ into a bounded subset of $\X\subset \R^{d_X}$ then $\psi(X)$ may be considered instead of $X$ and assumptions on $g$ may be presumed to hold for $g \circ \psi$. Our notation and assumptions are for $X$ continuously distributed with respect to Lebesgue measure, but this regularity condition could readily be dropped in favor of another dominating measure (in the discrete $X$ case, consider counting measure). Note that Assumption \ref{A:C1} allows the choice $B = \infty$, which places no restrictions on $g$ beyond continuity. Assumption \ref{A:C2} is standard. 

The constant $B$, which is the provided bound on the sup-norm of $g$, is fundamental in what follows. For any $t \in \R$ denote by $f_w^t(x,u)$ the transformed function $f_w(x, t + u)$. Define now the operator $T$ on $L^2(\mathcal{X} \times \R) $ by
\begin{align*}
(Th)(t)& \equiv \int_\mathcal{X} \int_0^{2B} h(x,u) (f_0( x, t + u) - f_1(x, t + u)) \, \mathrm{d} u \, \mathrm{d}x \\
& = \inner{h, f_0^t - f_1^t}_{L^2(\mathcal{X} \times [0,2B]))} \\
& = \E{ h(X, U-t) \one_{U \in[ t, 2B + t]}|W = 0 } - \E{ h(X, U-t) \one_{U \in [t, 2B + t ]} | W = 1},
\end{align*} 
where the last line follows by substituting $v = t + u$ in the definition of $(Th)(t)$. The most important aspect of $T$ is that it takes the familiar form of a linear operator between Banach spaces. 

To develop some insight on $T$, we characterize the range of our operator $T$ in the following lemma. All $L^p$ spaces are taken with respect to Lebesgue measure unless otherwise noted, and spaces of continuous functions $C(\cdot)$ are equipped naturally with the sup-norm, which makes them complete metric spaces. 
\begin{lemma} \label{L:operator} Suppose Assumptions \ref{A:C1} and \ref{A:C2} hold. 
Then the linear mapping $T:L^\infty(\mathcal{X} \times \R)\ra  C(\R)$ is bounded. If $B < \infty$, then additionally $T: L^2(\mathcal{X} \times \R)\ra C(\R)$ is bounded. Moreover, if $f_0 - f_1 \in L^2(\mathcal{X} \times \R)$ and $B< \infty$ then $T$ is a bounded linear map from $L^2(\mathcal{X} \times \R)$ to $L^2(\R)$. 
\end{lemma}

The main assumption that we make is that the joint distribution of $X,U$ is sufficiently rich, in that the linear operator $T$ has small enough kernel. It will turn out that the stipulation $U \indep W$ implies that $\ker T$ is always nontrivial. Now, to proceed, we define the subspace of $x$-invariant functions $\mathcal{V}$ on $\X \times \R$ by 
\begin{align*}
\mathcal{V} \equiv \Big\{ h \in L^2(\mathcal{X} \times \R) :& \, h(x,u) = H(u) \text{ on } \R \text{ for some }H:\R\ra \R \\
& \text{ and } \supp {h} \subset \X \times [0,2B]
\Big\}.
\end{align*}
Define also the set of functions $\mathcal{W}$ by 
\begin{align*}
\mathcal{W} \equiv \left\{h \in \mathcal{V}: h(x,u) = \one_{u \in [0,\delta(x)]} \text{ for some nonconstant }\delta(x) \in C(\mathcal{X})_+,\, \sup_{x \in \X} |\delta(x)| \le 2B
\right\}
\end{align*}
where we have used $C(\mathcal{X})_+$ to denote the set of positive continuous real-valued functions over $\mathcal{X}$. Note that $\mathcal{V}$ forms a closed linear subspace of $L^2(\X \times \R)$, and the restriction of its elements to $\X \times [0,2B]$ is a closed linear subspace of $L^2(\X \times [0,2B])$. Therefore, projection onto $\mathcal{V}$ is a well-defined and bounded operator.

Our key assumption, which is given in two (nonequivalent) forms, is as follows. 
\begin{asm} \normalfont\label{A:C3}
When $T$ is viewed as an operator mapping $L^2(\X \times [0,2B]) \ra C(\R)$, either:
\begin{enumerate}
\item[(i)]
$\ker{T} \cap \mathcal{W} = \emptyset$,
\item[(ii)] or more specifically
$\ker{T}\subset  \mathcal{V}$.
\end{enumerate}
\end{asm}
Note that we impose the constraint that $\delta(x)$ is nonconstant in the definition of $\mathcal{W}$. Hence, it may be readily be seen that $\mathcal{V} \cap \mathcal{W} = \emptyset$ so that Assumption \ref{A:C3}(ii) is stronger than Assumption \ref{A:C3}(i). It will turn out that Assumption \ref{A:C3}(i) is necessary and sufficient for our purposes of identification, but Assumption \ref{A:C3}(ii) has a more meaningful interpretation that we now turn to.

It is a fact that our standing assumption that $U \indep W$ implies that $\mathcal{V} \subset \ker{T}$; indeed, for any $h \in \mathcal{V}$ there exists by definition some function $H$ such that: 
\begin{align*}
Th(t) = \E{ H(U - t) \one_{U \in [t, 2B +  t]} | W = 0 } - \E{ H(U - t) \one_{U \in [t, 2B + t]} | W = 1} = 0.
\end{align*}
So under independence Assumption \ref{A:C3}(ii) amounts to the condition that $\ker{T}$ is \textit{precisely} $\mathcal{V}$. Indeed, have the following result clarifying the relationship between Assumption \ref{A:C3}(ii) and independence $U \indep W$.

\begin{lemma} \label{L:indepequiv}
Under Assumptions \ref{A:C1} and \ref{A:C2}, independence $U \indep W$ is equivalent to the inclusion $\mathcal{V} \subset \ker{T}$. 
\end{lemma}

Now, we are able to clarify the conditions necessary and sufficient to obtain point identification of $g$ under our stated regularity assumptions. The following theorem is an application of Green's theorem whose proof makes clear why we have introduced the operator $T$:
\begin{theorem}\label{T:contident}
Suppose Assumptions \ref{A:C1}, \ref{A:C2}, and the restrictions $U \indep W$ and $\E{U} = 0$ hold. Then Assumption \ref{A:C3}(i) holds if and only if $g$ is point identified in the set $\mathcal{G} \equiv \left\{ h \in C(\mathcal{X}): \norm{h}_\infty \le B \right\}$. 
In particular, Assumption \ref{A:C3}(ii) implies point identification of $g$. 
\end{theorem}

Assumption \ref{A:C3} is clearly central and merits further investigation. 
To further understand it we can rephrase our requirement in terms more closely resembling typical completeness assumptions, which typically appear as:
\begin{align} \label{E:C4}
\E{f(X)|W} \overset{\mathrm{a.s.}}{=} 0  \implies f(X) \overset{\mathrm{a.s.}}{ = } 0, 
\end{align}
where $W$ is some instrument for an endogenous regressor $X$. This particular assumption has been addressed in varying forms; see \cite{A2017} for a recent treatment. 

To place our Assumption \ref{A:C3}(ii) in terms of the more familiar condition \eqref{E:C4}, let $V$ be a random variable distributed as uniform $\mathcal{U}[0,2B]$, independently of $(U,X)$, which we assume has a distribution with density $f(x,u)$ on $\mathcal{X} \times \R$ in accordance with Assumption \ref{A:C2}.

\begin{lemma} \label{L:compeleteinstrument}
Let $\tilde{U} \equiv U + V$ where $V\sim \mathcal{U}[0,2B]$ is independent of $(U,X,W)$. Then  under Assumptions \ref{A:C1} and \ref{A:C2}, Assumption \ref{A:C3}(ii) is equivalent to the following assertion: 
\[
\E{ h(X,V) | \tilde{U}, W = 0} \overset{\mathrm{a.s.}}{ = } \E{h(X,V )| \tilde{U}, W = 1} \implies h \in \mathcal{V}
\]
whenever $h \in L^2(\mathcal{X}\times [0,2B])$. 
\end{lemma}

\subsection{More on Assumption \ref{A:C3}}

Given its utility, we wish to explore conditions under which the stronger Assumption \ref{A:C3} holds, in particular with respect to the conditional density functions $f_0(x,u)$ and $f_1(x,u)$. To this end, define $\Gamma$ as the set of functions $\gamma(x,u)$ over $\mathcal{X} \times \R$ which are of the form $f_0(x,u)-f_1(x,u)$, where $f_0$ and $f_1$ are proper density functions, i.e.\
\[
\Gamma \equiv \left\{ \gamma : \, \gamma = f_0 - f_1, \text{ } f_0, f_1 \text{ are density functions over }\mathcal{X} \times \R
\text{ satisfying } U \indep W \right\}.
\]
We use $U \indep W$ to indicate that one should have the relation $\int_\X f_0(x,u) \, \mathrm{d}x = \int_\X f_1(x,u) \, \mathrm{d} x$, for a.e.\ $u$, i.e.\ equality of the conditional distributions of $U$ given $W$, almost everywhere.
It can be seen that 
\begin{align} \label{E:closedset}
\Gamma = \left\{ \gamma: \norm{\gamma}_{L^1(\X \times \R)} \le 2, \, \int_\X \gamma(x,u) \, \mathrm{d} x  = 0 \text{ for a.e. }u,  \right\}.
\end{align}
An alternate formulation is to impose the moment condition that $f_0$ satisfies $\int_\R \int_\X u f_0(x,u) \, \mathrm{d}x \, \mathrm{d} u = 0$ which has been employed throughout the paper. It may readily be seen that: 
\small
\begin{align*}
&\left\{ \gamma : \, \gamma = f_0 - f_1, \text{ } f_0, f_1 \text{ are density functions over }\mathcal{X} \times \R
\text{ satisfying } U \indep W \text{ and } \E{U} = 0\right\}\\
& = \left\{ \gamma: \norm{\gamma}_{L^1(\X \times \R)} \le 2, \, \int_\X \gamma(x,u) \, \mathrm{d} x  = 0 \text{ for a.e. }u, \, \int_\R \int_\X  u\, \gamma^+(x,u) \, \mathrm{d}x \, \mathrm{d}u = 0 \text{ if } \norm{\gamma}_{L^1(\X \times \R)} = 2 \right\},
\end{align*}
\normalsize
where inclusion in one direction is clear and inclusion in the other direction follows from the fact that, given $\gamma$ satisfying $\norm{\gamma}_{L^1(\X \times \R)} \le 2$ and $\int_\X \gamma(x,u) \, \mathrm{d} x \, \mathrm{d} u = 0$ for \text{a.e.} $x$, one may set
\begin{align}
f_0 &= \gamma^+ + \left( 1 - \int_\X \int_\R \gamma(x,u)^+ \, \mathrm{d} u \, \mathrm{d} x \right) \rho \label{E:reformulation}\\ 
f_1 &= \gamma^-  + \left( 1 - \int_\X \int_\R \gamma(x,u)^+ \, \mathrm{d} u \, \mathrm{d} x \right) \rho, \nonumber
\end{align}
where $\gamma = \gamma^+- \gamma^-$, $\gamma^+, \gamma^- \ge 0$, and $\rho$ is an arbitrary probability density function defined on $\X \times \R$ chosen to satisfy $\int_\R \int_\X uf_0(x,u)\, \mathrm{d} x \, \mathrm{d} u = 0$ (and if necessary to ensure integrability of the functions). Then $\gamma = f_0 - f_1$. Note that \eqref{E:reformulation} only slightly differs from \eqref{E:closedset} and only in those elements $\gamma$ for which $\norm{\gamma}_{L^1(\X \times \R)} = 2$. As the results established in this section concern density and genericity in the $L^1(\X \times \R)$ norm and are not affected by the immaterial change from \eqref{E:closedset} to \eqref{E:reformulation}, we work with the more convenient definition \eqref{E:closedset}. 
Note that \eqref{E:closedset} implies that $\Gamma$ is a closed set in $L^1(\X\times \R)$; if for instance $\gamma_n$ is a sequence occurring in $\Gamma$ such that $\gamma_n \ra \gamma$ in $L^1(\X \times \R)$ then the Lebesgue differentiation theorem implies
\begin{align*}
\int_\X \gamma(x,u) \, \mathrm{d} x & \eqas \lim_{b \ra 0^+} (2b)^{-1}|\X|^{-1} \int_\R \one_{-b \le u \le b } \int_\X \gamma(x,u) \, \mathrm{d} x \, \mathrm{d} u \\
& = \lim_{b \ra 0^+} (2b)^{-1} |\X|^{-1}\lim_{n \ra \infty} \int_\R \int_\X \one_{-b \le u \le b } \gamma_n(x,u) \, \mathrm{d} x \, \mathrm{d} u = 0. 
\end{align*}

For any $\gamma \in \Gamma$, define the linear operator $T_\gamma: L^2(\mathcal{X} \times \R ) \ra C(\R)$ (see Lemma  \ref{L:operator}) by 
\begin{align*}
T_\gamma h(t) \equiv \int_\X \int_0^{2B} h(x,u) \gamma(x,t + u) \, \mathrm{d} u \, \mathrm{d} x. 
\end{align*}
Then let $\Gamma_0 = \left\{ \gamma \in \Gamma: \, \gamma \text{ is continuous and } \ker{T_\gamma} = \mathcal{V} \right\}$. Then we have the following density result: 
\begin{proposition}
In the preceding notation, $\Gamma_0$ is dense in $\Gamma$ in the $L^1(\X \times \R )$-norm.  \label{P:densityapprox}
\end{proposition}

As an immediate corollary, we obtain: 
\begin{corollary}\label{C:densityapprox}
For any continuous probability distributions $f_0$, $f_1$ on $\X \times \R$ satisfying $U \indep W$, and $\ve > 0$, there exist probability densities $f_0^\ve$ and $f_1^\ve$ such that $\norm{f_\ell^\ve - f_\ell}_{L^1(\X \times  \R )} < \ve$ for $\ell =0,1$, $f_0^\ve - f_1^\ve \in \Gamma_0$, and $\int_{\R} \int_\X u f_0^\ve(x,u) \, \mathrm{d}x \, \mathrm{d} u = 0$. 
\end{corollary}

\subsubsection{Topological genericity of Assumption \ref{A:C3}(i) under a Lipschitz restriction} \label{SSS:generic}

By making an additional assumption on the smoothness of the function $g$, one can comment on the topologically genericity of Assumption \ref{A:C3}(i). Recall that in a given topological space $(\mathfrak{X}, \mathcal{T})$ a set is called \textbf{residual} or \textbf{comeagre} if it contains a countable intersection of open dense sets (a dense $G_\delta$ set). When $(\mathfrak{X}, \mathcal{T})$ is a complete metric space, the Baire Category theorem implies that any residual set contains a dense set (additionally, any residual set is not countable), and residuality is used to define a generic property on a given topological space. For instance, the irrational numbers comprise a residual set in $\R$ whereas the rationals do not. Recall that $\Gamma$ is a closed set in $L^1(\X \times \R )$, and thus it is a complete metric space with the topology induced by $L^1(\X \times \R )$.

In this section, we will confine $g$ to belong to the class of Lipschitz-continuous functions on $\X$. Then if $h$ is also in the identified set and Lipschitz continuous, the difference $\delta = g - h$ is Lipschitz continuous. By contrapositive, if there are no Lipschitz continuous functions $\delta$ such satisfy $\one_{u \in [0, \delta(x)]} \in \ker{T}$, then it follows straightforwardly by the method of Theorem \ref{T:contident} that $g$ is point identified under the Lipschitz restriction. Thus, analogously to $\mathcal{W}$, define 
\[
\mathcal{W}_\mathrm{Lip} \equiv \left\{ \one_{u \in [0, \delta(x)]} \text{ for some Lipschitz-continuous, nonconstant } \delta \in C(\X)_+\right\} 
\]
$\mathcal{W}_\mathrm{Lip}$ is the restriction of $\mathcal{W}$ to the Lipschitz case. 
Equipped with these definitions, we have the following result. 

\begin{proposition} \label{P:resid}
Let $\Gamma_1 \equiv \left\{ \gamma \in \Gamma: \, \ker{T_\gamma } \cap \mathcal{W}_{\mathrm{Lip}} = \emptyset \right\}$; then $\Gamma_1$ is a residual set in $\Gamma$ in the topology induced from $L^1(\X \times \R  )$. 
\end{proposition}

The genericity result is also relevant when we consider densities, and not functions which are the difference of densities. For let $\mathfrak{X}$ denote the set of probability density functions over $\X \times \R$ equipped with the $L^1(\X \times \R )$ norm, and $\mathfrak{F} \subset \mathfrak{X} \times \mathfrak{X}$ the set of pairs of densities $(f_0, f_1)$ such that $f_0 - f_1 \in \Gamma$, with the induced product topology. Note that $\mathfrak{F}$ is manifestly closed therein. Let $\mathfrak{F}_1$ denote the set of pairs $(f_0, f_1)$ such that $f_0 - f_1 \in \Gamma_1$. Then:
\begin{corollary}\label{C:densityapprox2}
When $\mathfrak{F}$ is equipped with its induced product topology, $\mathfrak{F}_1$ is a residual set in $\mathfrak{F}$. 
\end{corollary}
In Corollary \ref{C:densityapprox2} and the definition of $\mathfrak{F}$ we do not impose the additional moment requirement that $\int_\R \int_\X uf_0(x,u) \, \mathrm{d} x \, \mathrm{d} u = 0$ considered in \eqref{E:reformulation} because the set of densities which satisfy this condition, and indeed the weaker condition of mere integrability of $uf_0(x,u)$, is not a closed subset of $L^1(\X\times \R)$. Imposing this requirement would require us to consider a stronger topology on $\mathfrak{F}$; results in this direction could certainly be made along the lines of Proposition \ref{P:densityapprox}, but the $L^1$ topology is arguably the most natural when discussing the $L^1$-closed set of probability density functions. 

\subsubsection{Identification when $U$ has Compact Support}

Given that the examples produced in Proposition \ref{P:densityapprox} and Corollary \ref{C:densityapprox} of conditional density functions which point identified $g$ had unbounded support, it may come as a surprise to the reader that there exist examples of density functions which point identify $g$ and have bounded support. For suppose that $\supp{U} \subset [-C_1, C_2]$ for fixed constants $C_1, C_2 > 0$. We show the following density result which is a corollary of Proposition \ref{P:densityapprox} and Corollary \ref{C:densityapprox}: 
\begin{corollary} \label{C:boundedsppt}
If $C_1 + C_2 > 2B$ then for any continuous probability densities $f_0, f_1$ on $\X \times [-C_1 , C_2]$ satisfying $U \indep W$ and $\E{U} = 0$ and any $\ve > 0$ there exist probability densities $f_0^\ve, f_1^\ve \in L^1(\X \times [-C_1, C_2])$ such that $\norm{f_\ell^\ve - f_\ell}_{L^1(\X \times \R)} < \ve$ for $\ell = 0 , 1$, $f_0^\ve - f_1^\ve \in \Gamma_0$, and $\int_\R \int_\X uf_0^\ve (x,u) \, \mathrm{d} x \, \mathrm{d} u = 0$. 
\end{corollary}
Retaining the notation of \S \ref{SSS:generic},
let $\mathfrak{F}^{C_1,C_2}$ denote the set of elements in $\mathfrak{F}$ with support in $\X \times [-C_1, C_2]$. Similarly, let $\mathfrak{F}_1^{C_1, C_2}$ denote the elements $(f_0,f_1) \in \mathfrak{F}^{C_1,C_2}$ such that $\ker{T_{f_0 - f_1}} \cap \mathcal{W}_{\mathrm{Lip}} \cap \mathcal{W} = \emptyset$ (note now the dependence on the bound $B$ via $\mathcal{W}$). Then exactly the same arguments which led to Proposition \ref{P:resid} and Corollary \ref{C:densityapprox2} imply in light of Corollary \ref{C:boundedsppt} that
\begin{corollary}
If $C_1 + C_2 > 2B$ then $\mathfrak{F}_1^{C_1, C_2}$ is a residual set in $\mathfrak{F}^{C_1, C_2}$. 
\end{corollary}

\subsection{Connection with Identification in Nonparametric Instrumental Variables Quantile Regression} \label{SS:quantile}

Interestingly, it is possible to extend the methods used in this section to identification in a quantile regression model as considered by \cite{CH2005} and later by \cite{HL2007}. Consider the framework 
\begin{align}
\label{E:quantilemodel} &Y = g(X) + U \\
&\P{U \le 0 | W } \eqas \gamma(W) \nonumber
\end{align}
adopted in \cite{HL2007}, where $\gamma$ is some known function mapping $\supp{W}$ into $[0,1]$ and variables retain their interpretation from our standing model \eqref{E:model}. The model displayed in \eqref{E:quantilemodel} nests the model considered in \cite{HL2007} (consider the constant function $\gamma(w) = q$, $q$ fixed), who show that their model subsumes the setup considered by \cite{CH2005}. For a random variable $Z$, say that $W$ is \textit{boundedly complete } for $(X,U)$ if for all bounded functions $h: \supp{(X,U)} \ra \R$ one has $\E{h(X,U)|W} \overset{\mathrm{a.s.}}{=} 0$ if and only if $h \eqas 0$. In the spirit of Assumption \ref{A:C3}, say that $W$ is \textit{boundedly} $X$\textit{-complete} for $(X,U)$ if for all bounded functions $h: \supp{(X,U)} \ra \R$ one has $\E{h(X,U)|W} \overset{\mathrm{a.s.}}{=} 0$ only if $h(X,U) \overset{\mathrm{a.s.}}{=} H(U)$ for some function $H$, i.e.\ $h$ does not depend on $X$. It should be clear that if $Z$ is boundedly complete for $(X,U)$, then it is boundedly $X$-complete for $(X,U)$.

Equipped with these definitions, we derive the following identification result for the structural function $g$ completely along the lines of Theorem \ref{T:contident}: 

\begin{proposition}\label{P:quantile}
Suppose model \eqref{E:quantilemodel} holds. If $W$ is boundedly complete for $(X,U)$, then $g$ is point identified. If $0$ is in the interior of $\supp{U}$ and $W$ is boundedly $X$-complete for $(X,U)$, then $g$ is also point identified. 
\end{proposition}

One helpful aspect of Proposition \ref{P:quantile} is that it sidesteps some issues faced when considering identification of nonlinear operators, which is faced by \cite{CH2005}. A number of sufficient conditions for bounded completeness have been developed by e.g.\ \cite{D2011}, to which we refer the interested reader. Roughly speaking, if 
\begin{align*}
(X,U) = \mu(\nu(W) + \ve)
\end{align*}
for some random disturbance $\ve$ which is independent of $W$, where $\mu$ and $\nu$ are possibly vector-valued functions, then there are light conditions which can be made (see Assumptions 1-3 and 4 of \cite{D2011}) to ensure that $W$ is boundedly complete for $(X,U)$.

\subsection{Discussion, Comparison with Local Identification}
\label{SS:comparison}
The condition obtained by \cite{D2014} (see their equation (16)) and more recently considered by \cite{C2019} (see their Assumption 2.1) for local identification in our model is the relation: for $h$ satisfying $\E{h(X)} = 0$ and $\E{|h(X)|^2} < \infty$,
\begin{align} \label{E:localid}
\E{h(X)|U , W = 0} - \E{h(X)|U, W = 1} \eqas 0 \implies h \eqas 0 
\end{align}
Comparison with Assumption \ref{A:C3}(ii) shows that the stronger assumption we make in order to obtain point identification of $g$ is stronger than \eqref{E:localid}, as should be expected. For suppose that \eqref{E:localid} does not hold for some square integrable $h$: then for all $t \in \R$
\begin{align*}
Th(t) &\equiv \E{ h(X) \one_{U \in [t,2B + t]} | W = 0} - \E{h(X) \one_{U \in [t,2B + t]} | W = 1}  \\
& = \E{ \E{h(X) |U,W = 0} \one_{U \in [t, 2B + t]} | W = 0   } - \E{ \E{h(X) |U,W = 1} \one_{U \in [t, 2B + t]} | W = 1   } \\
& = \E{ (\E{h(X)|U, W = 0} - \E{h(X)|U, W = 1}) \one_{U \in [t,2B + t]}} = 0.
\end{align*}
Hence, $h \in \ker{T} \setminus \mathcal{V}$ and Assumption \ref{A:C3}(ii) is violated. The relation of \eqref{E:localid} with the necessary and sufficient condition Assumption \ref{A:C3}(i) is more difficult to ascertain, which may suggest that \eqref{E:localid} is not a necessary condition for local identification. 

A typical completeness condition puts $\dim X = \dim W$ and asks that, conditional on some restrictions on the function $h$, $\E{h(X)|W} \eqas 0$ if and only if $h(X) \eqas 0$. One of the most studied examples where the completeness condition is fulfilled puts $X = \mu(\nu(W) + \ve)$, as in \cite{D2011}; in this case, it is somewhat essential that $\dim W \ge \dim X$ and that $W$ satisfies a large support condition. Interestingly, in both of the settings we have discussed, identification has been shown to arise when a form of completeness condition holds for an instrument which has possibly lower dimension than its regressor. For instance, Lemma \ref{L:compeleteinstrument} shows that our Assumption \ref{A:C3} is tantamount to the assertion that a random variable $\tilde{U} = U + V$ is complete for the vector $(X,V)$ in a sense defined there, and within the class of functions $\mathcal{V}$. Of course, $\dim(X,V) > \dim \tilde{U}$, which imposes some difficulties when attempting to view our identification assumptions through the typical lens of instrument completeness. Moreover, in Proposition \ref{P:quantile} we require $W$ to act as a complete instrument for $(X,U)$, so that in order to apply conventional examples of completeness one would have to have $\dim W \ge \dim X + \dim U$. Hence, while the conditions enumerated in Lemma \ref{L:compeleteinstrument} and Proposition \ref{P:quantile} are not necessary for identification, they suggest that to state examples of identified models in our framework is also to make progress on finding sufficient conditions for the completeness condition when the instrument has strictly lower degree than the regressor (and vice-versa).



\pagebreak

\section{Appendix: Proofs}\label{S:proofs}

\subsection{Proof of Theorem \ref{T:identified}}

 We begin our proof by stating a result from algebraic geometry which characterizes the solution set (variety) of a system of multivariate polynomials when it is finite. 
\begin{theorem}
[Finiteness Theorem, \cite{C2005}] Let $I \subset K[x_1, \ldots , x_K]$ be a polynomial ideal over a field $K \subset \C$. Then the following are equivalent: 
\begin{enumerate}
\item The variety $V(I)$ is a finite set 
\item For each $k$, $1\le k \le K$, there is some $m_k$ such that $x_k^{m_k} \in \mathrm{LT}(I)$.
\end{enumerate}
\end{theorem}
Here, $\mathrm{LT}(I)$ is the monomial ideal generated by the leading terms of polynomials in $I$. The determination of a leading term requires a fixed monomial order. We use the \textit{graded lexicographic order} (or graded reverse lexicographic order), which satisfies $x^\alpha >_{\mathrm{grlex}} x^\beta$ if $\sum_{i=1}^n \alpha_i > \sum_{i=1}^n \beta_i$ or if equality in the total degrees of the monomials holds and $x^{\alpha} >_{\mathrm{lex}} x^\beta$. Now consider the following system of polynomials: 
\begin{align}
0 = P_1(x_1, \ldots , x_K) & = P_1^{(1)}(x_1, \ldots , x_K) + P_1^{(0)}(x_1, \ldots , x_K) \label{E:system2} \\
& \vdots \nonumber \\
0 = P_n(x_1, \ldots , x_K) & = P_n^{(1)} (x_1, \ldots , x_K) + P_n^{(0)} (x_1, \ldots, x_K) \nonumber
\end{align}
where for all $i$ we have let $P_i^{(1)}$ denote the polynomial consisting of all monomials of $P_i$ with highest total degree, and $P_i^{(0)}$ the polynomial consisting of all of the remaining monomials (with strictly smaller total degree). As a corollary to the finiteness theorem, we obtain the following: 
\begin{lemma} \label{L:boundsln}
Suppose that the number of solutions to the reduced polynomial system: 
\begin{align*}
P_1^{(1)}(x_1, \ldots, x_k) & = 0 \\
& \vdots \\
P_n^{(1)}(x_1, \ldots , x_k) & = 0
\end{align*}
is finite. Then the variety of the full system \eqref{E:system2} is finite. 
\end{lemma}
\begin{proof}
Assume that our reduced system has a finite number of solutions. Then the finiteness theorem implies that for every $k \in [K]$ there are polynomials $g_1, \ldots , g_n \in K[x_1, \ldots , x_K]$ satisfying: 
\begin{align} \label{E:leadingterm}
\mathrm{LT}(g_1 P_1^{(1)} + \cdots + g_n P_n^{(1)} ) = x_k^{m_k}
\end{align}
for some $k$. Let $\ell_1, \ldots , \ell_K$ denote the total degrees of the monomials in $P_1^{(1)}, \ldots, P_K^{(1)}$ respectively. For each $k$, let $g_k^{(2)}$ denote the terms in $g_k$ of total degree exceeding $m_k - \ell_k$, $g_k^{(1)}$ the terms of total degree exactly $m_k - \ell_k$, and $g_k^{(0)}$ the terms of total degree less than $m_k - \ell_k$. Note that $\sum_{k=1}^K g_k^{(2)} P_k^{(1)}$ is necessarily a polynomial consisting of monomials of total degree greater than $m_k$; this must be the zero polynomial or else $x_k^{m_k}$ is not the leading term of \eqref{E:leadingterm}. Moreover $\sum_{k=1}^K g_k^{(0)}P_k^{(1)}$ is a polynomial consisting of monomials of total degree strictly less than $m_k$. Hence we may assume without loss of generality that $g_k^{(2)} = g_k^{(0)} = 0$ for all $k$ without affecting the equality of \eqref{E:leadingterm}. This essentially concludes, because 
\begin{align*}
\sum_{k = 1}^K g_k P_k & = \sum_{k=1}^K g_k^{(1)} P_k^{(1)} + \sum_{k=1}^K g_k^{(1)} P_k^{(0)};
\end{align*}
notice that the second summand is a polynomial with monomials having total degree strictly less than $m_k$, and so our choice of monomial order one has 
\[
\mathrm{LT} \left( \sum_{k=1}^K g_k P_k \right) = \mathrm{LT}\left( \sum_{k=1}^K g_k^{(1)} P_k^{(1)} \right) = x_k^{m_k}
\]
which follows by \eqref{E:leadingterm}. So $x_k^{m_k} \in \mathrm{LT}(\langle P_1, \ldots , P_n \rangle)$ and the finiteness theorem concludes. 
\end{proof}

\noindent\textbf{Remark}: Let $I$ be the ideal generated by the system \eqref{E:system2} and $I^{(1)}$ the ideal generated by the polynomials $P_i^{(1)}$. Our argument has shown that any monomial appearing in $\mathrm{LT}(I^{(1)})$ must also appear in $\mathrm{LT}(I)$. Combined with the fact that the size of the solution set, counting multiplicities, of \eqref{E:system2} is given by the number of standard monomials not appearing in  $\mathrm{LT}(I)$ (see \cite{C2005}) our argument has also established an upper bound on the solution set of \eqref{E:system2} in terms of the solution set of the reduced system, counting multiplicities. 

\bigskip

Now consider our particular system of polynomial equations \eqref{E:nmoment2}. It is easy to see that the polynomial system formed of terms of highest total degree in this system are: 
\begin{align*}
\sum_{k=1}^K (p_k(0) - p_k(1)) (-1)^n g_k^n = 0
\end{align*}
for all $n$. By rescaling this is equivalent to 
\begin{align} \label{E:leadingterms}
\sum_{k=1}^K (p_k(0) - p_k(1)) g_k^n = 0,
\end{align}
for all $n$. 

Note that Assumption \ref{A:2} and independence imply that $\{g_k\}$ must also satisfy 
\begin{align} \label{E:leadingterms2}
\sum_{k=1}^K p_k(0)\E{ Y - g_k| X = k, W = 0}  = \E{ U| W = 0} = \E{U} = 0.
\end{align}
Hence the vector $\{g_k\}_{k=1}^K$ must satisfy \eqref{E:leadingterms2} in addition to \eqref{E:nmoment2}. Extract the terms of highest total degree from \eqref{E:leadingterms2} and combine with \eqref{E:leadingterms} to determine by Lemma \ref{L:boundsln} that the number of solutions to the polynomial system of equations formed by \eqref{E:nmoment2} and \eqref{E:leadingterms2} is finite if the variety of the system: 
\begin{align} \label{E:lt3}
&\sum_{k=1}^K p_k(0) g_k = 0 \\
&\sum_{k=1}^K (p_k(0) - p_k(1)) g_k^n = 0 \nonumber
\end{align}
is finite. We will show that the number of such solutions in $\C^k$ is finite, using only $n = 1, \ldots , K$. In particular we will show that only the trivial solution $\{g_k\}_{k=1}^K = 0_k$ satisfies \eqref{E:lt3} under Assumption \ref{A:1}. We proceed by contradiction, supposing that $g_k$ is a  nonzero solution of \eqref{E:lt3}. Because $\sum_{k=1}^K p_k(0) = 1$ we may exclude solutions of the form $g_k = c$, $c \neq 0$, from consideration. Hence the $g_k$ take on at least $2$ distinct values. Now suppose that $\{g_k: k = 1, \ldots, K\} \setminus \{0\} = \{z_1, \ldots , z_M\}$ where $0<M \le K$, and let $K_m = \{k: g_k = z_m\}$ denote the set of indices on which $g_k$ equals $z_m \neq 0$, for $1 \le m \le M$. By assumption at least $K_1$ must be nonempty. Then for all $n = 1, \ldots , K$ \eqref{E:lt3} implies that 
\begin{align}
\sum_{m = 1}^M \sum_{k \in K_m} (p_k(0) - p_k(1)) z_m^n = \sum_{k=1}^K (p_k(0) - p_k(1)) g_k^n = 0. \label{E:one}
\end{align}
We have already shown that $K_m \subsetneq [K]$ for all $m$ and that $M \ge 2$. Moreover, Assumption \ref{A:1} implies that $\sum_{k \in K_m} (p_k(0) - p_k(1))$ is nonvanishing for all $m$. 
Now using the fact that $M \le K$ we have the following linear relation: 
\begin{align*}
&\begin{pmatrix}
1 & \cdots & 1 \\
z_1 & \cdots & z_M \\
& \vdots & \\
z_1^{M-1} & \cdots & z_M^{M-1}
\end{pmatrix}
\begin{pmatrix}
z_1 & 0 & \cdots & 0 \\
0 & z_2 & \cdots & 0 \\
 &  & \ddots  & \\
0 & 0 & \cdots & z_M
\end{pmatrix}
\begin{pmatrix}
\sum_{k \in K_1} (p_k(0) - p_k(1)) \\
\vdots \\
\sum_{k \in K_M} (p_k(0) - p_k(1))
\end{pmatrix}\\
&\quad  = 
\begin{pmatrix}
z_1 & \cdots & z_M \\
z_1^2 & \cdots & z_M^2 \\
& \vdots & \\
z_1^M & \cdots & z_M^M
\end{pmatrix}
\begin{pmatrix}
\sum_{k \in K_1} (p_k(0) - p_k(1)) \\
\vdots \\
\sum_{k \in K_M} (p_k(0) - p_k(1))
\end{pmatrix}
 = \begin{pmatrix}
0 \\
\vdots \\
0
\end{pmatrix}.
\end{align*}
One recognizes the matrix on the left as the transpose of a Vandermonde matrix whose determinant can be calculated explicitly as $\prod_{1 \le m < \ell \le M} (z_m - z_\ell)$, which is nonzero as the $z_m$ are distinct. Moreover, the diagonal matrix multiplying it is clearly invertible as the $z_m$ were specified to be nonzero. Hence we have our contradiction, and the trivial solution uniquely satisfies \eqref{E:lt3}. Lemma \ref{L:boundsln} concludes.  \qed  

\subsection{Proof of Theorem \ref{T:2}}

Suppose that Assumptions \ref{A:1} and \ref{A:2} hold along with the independence assumption $U \indep W$ in our discrete framework \eqref{E:model}. Fix vectors $\{p_k(0)\}$ and $\{p_k(1)\}$ for the remainder of the proof such that $p_k(0), p_k(1) > 0$ for all $k$. We claim that it is sufficient to show the existence of probability density functions $f_k^0$ and $f_k^1$ for $k = 1, \ldots , K$ as well as a nonzero vector $\{h_k\}_{k=1}^K$ such that $(h_k) \perp (p_k(0))$, $(h_k) \perp (p_k(1))$ satisfying: 
\begin{align}\label{E:ft1}
&\sum_{k=1}^K f_k^0(y) p_k(0) = \sum_{k=1}^K f_k^1(y) p_k(1) \\
&\sum_{k=1}^K f_k^0(y + h_k) p_k(0) = \sum_{k=1}^K f_k^1(y + h_k) p_k(1). \nonumber 
\end{align}
To see that this is the case, let $[Y|X = k, W = \ell] \sim f_k^\ell$. Note that in \eqref{E:ft1} we may assume without loss of generality that: 
\begin{align*}
\int_{-\infty}^\infty y\sum_{k=1}^K f_k^0(y) p_k(0)\, \mathrm{d}y = 0,
\end{align*}
by translating the density functions $f_k^\ell$ simultaneously by a constant $c$ if necessary. Hence, the first line of \eqref{E:ft1} implies $\E{Y|W = 0} = \E{Y|W = 1} = 0$ and also $Y \indep W$ so one may freely take $g(X) = 0$, $Y = U$ to satisfy Assumptions \ref{A:1} and \ref{A:2}. Moreover, by orthogonality of $h_k$ and $p_k(0)$ we have 
\begin{align*}
\int_{-\infty}^\infty \sum_{k=1}^K y f_k^0 (y + h_k) p_k(0) \, \mathrm{d}y  = \int_{-\infty}^\infty \sum_{k=1}^K (y - h_k) f_k^0( y) p_k(0) \, \mathrm{d} y = 0,
\end{align*}
so also $\E{ Y - h_X | W = 0 } = \E{ Y - h_X | W = 1} = 0$ and moreover $Y - h_X \indep W$. So one may also set $g(k) = h_k$ for $1 \le k \le K$ and still satisfy Assumptions \ref{A:1} and \ref{A:2} along with independence, whence $g$ is not point identified. With \eqref{E:ft1} in hand note that it is sufficient to consider the case $p_k(0) \neq p_k(1)$ for all $k$; else, set $f_k^0(y) = f_k^1(y)$ and drop the index $k$ from consideration. 

Proceed by fixing a nonconstant vector $h_k$ with the necessary orthogonality properties, assuming that $K \ge 3$ so that such an $h$ vector exists. By taking Fourier transforms (assuming certain regularity conditions, which we will prove) \eqref{E:ft1} is equivalent to 
\begin{align}
\sum_{k=1}^K \left(\hf_k^0 (t) p_k(0) - \hf_k^1(t) p_k(1)\right) & = \F\left[ \sum_{k=1}^K \left( f_k^0(y) p_k(0) -f_k^1(y) p_k(1) \right) \right](t)  \label{E:ft2}\\
& = 0 \nonumber \\
\sum_{k=1}^K  e^{2\pi i h_k t}\left[ \hf_k^0(t) p_k(0) - \hf_k^1(t) p_k(1) \right] & = \F\left[ \sum_{k=1}^K \left( f_k^0(y + h_k) p_k(0) -f_k^1(y + h_k) p_k(1) \right) \right](t) \nonumber \\
& = 0.  \nonumber
\end{align}
We proceed by exhibiting Schwartz functions $\gamma_k$ in the frequency domain for which: 
\begin{align} \label{E:ft3} 
&\gamma_k(0) = p_k(0) - p_k(1)\\ 
&\sum_{k=1}^K \gamma_k(t) = 0 \nonumber \\
&\sum_{k=1}^K  \gamma_k(t) e^{2\pi i h_k t}= 0 \nonumber
\end{align}
and then reconstructing the density functions $f_k^\ell$ by Fourier inversion. To save on summation notation, let $K \mu$ be counting measure on $X \equiv [K]$ and note that the second two lines of \eqref{E:ft3} are equivalent to $\int_X \gamma_x \, \d\mu(x) = 0$, $\int_X e^{2\pi i h_x t} \gamma_x \, \d\mu(x)$; proceeding with this notation will have the benefit of establishing our results for more general distributions of $X$ (e.g.\ continuous). Now, using the fact that $K \ge 3$ so that no two vectors span $\R^K$ the main component of $\gamma_x$ is derived from the Gram-Schmidt procedure as: 
\begin{align*}
\alpha_x(t) \equiv p_x(0) - p_x(1) - \frac{\int_X (p_y(0) - p_y(1)) e^{2\pi ih_y t} \, \d\mu(y) }{1 - \int_X e^{2\pi i h_y t} \, \d\mu(y) \int_X e^{-2\pi i h_y t} \, \d\mu(y)} \left( e^{-2\pi i h_x t} - \int_X e^{-2\pi i h_z t} \, \d \mu(z) \right).
\end{align*}
We verify a few properties of the function $\alpha$: 
\begin{lemma} \label{L:lt1}
$\alpha_x(t)$ satisfies $\int_X \alpha_x(t) \, \d\mu(x) = \int_X \alpha_x(t) e^{2\pi i h_x t} = 0$
\end{lemma}
\begin{proof}
This follows from the observation that $\int_X (p_x(0) - p_x(1))\, \d \mu(x) = 0$ and straightforward calculation.
\end{proof}
\begin{lemma}\label{L:lt2}
$\alpha_x(t) \in C^\infty(\R)$ up to removable singularities, $\mu$-almost surely.
\end{lemma}
\begin{proof}
It suffices to prove the claim for $\alpha^*_x(t) \equiv \alpha_x(t) - (p_x(0) - p_x(1))$. Note first that by Cauchy-Schwarz, for $t \in \R$, 
\begin{align*}
\norm{\alpha^*_x(t) }_2 &\le \left| \frac{\int_X (p_y(0) - p_y(1)) \left(e^{2\pi ih_y t} - \int_X e^{2\pi i h_z t} \, \d\mu(z)\right) \, d\mu(y) }{1 - \int_X e^{2\pi i h_y t} \, \d\mu(y) \int_X e^{-2\pi i h_y t} \, \d\mu(y)} \right| \norm{e^{-2\pi i h_x t} - \int_X e^{-2\pi i h_z t} \, \d \mu(z)  }_2 \\
& \le \frac{\norm{p_y(0) - p_y(1)}_2\norm{e^{2\pi ih_y t} - \int_X e^{2\pi i h_z t} \, \d\mu(z)  }_2 }{\norm{e^{2\pi ih_y t} - \int_X e^{2\pi i h_z t} \, \d\mu(z)  }^2_2}\norm{e^{2\pi ih_y t} - \int_X e^{2\pi i h_z t} \, \d\mu(z)  }_2 \\
& \le \norm{p_y(0) - p_y(1)}_2 < \infty.
\end{align*}
where all norms are taken in $L^2(X, \mu)$. Now write $\alpha^*_x(t) = \frac{P_x(t)}{Q(t)}$ where 
\[Q(t) = 1 - \int_X e^{2\pi i h_y t} \, \d\mu(y) \int_X e^{-2\pi i h_y t} \, \d\mu(y) = 1 - \int_X \int_X e^{2\pi i (h_y - h_z) t} \, \d\mu(y) \d\mu(z).
\]
By the dominated convergence theorem it is clear that both $P_x$ and $Q$ are entire functions in $t$, whence $\alpha_x^*(t)$ has either removable singularities or poles of finite order.\footnote{Let $t \in Z_Q$ be such that $Q(t) = 0$. At $t$ we may factor $Q$ as $Q(z ) = (z - t)^m G(z )$, where $m \ge 1$ and $G(t) \neq 0$. Hence $ \frac{P(z)}{Q(z)} = (z - t)^{-m} \frac{P(z)}{G(z)}$ where $\frac{P}{G}$ is holomorphic in a neighborhood of $t$.} Moreover $Q$ is nonconstant as 
\begin{align*}
\pars{}{t} Q(t) & = \int_X \int_X  (h_y - h_z)^2 e^{2\pi i (h_y - h_z) t} \, \d\mu(y) \d\mu(z) \neq 0
\end{align*}
(evaluate at $t = 0$ to see the nonequivalence). Hence, the set of zeroes $Z_Q$ of $Q$ is at most a discrete set, whence countable. Let $t_0 \in Z_Q\cap \R$. One can either have $\lim_{t \ra t_0} |\alpha_x^*(t) |= \infty$ (pole) or $\limsup_{t \ra t_0} |\alpha_x^*(t) | < \infty$ (removable singularity). However, the bound $\norm{\alpha_x^*(t)}_2 < \infty$ implies that the measure of $x$ on which $\lim_{t \ra t_0} | \alpha_x^*(t)| = \infty$ is zero. Deleting at most countably many such null sets (one for each point in $Z_Q \cap \R$) we may assume that for all $x$ we have $\limsup_{t \ra t_0} |\alpha^*_x(t)| < \infty$, whence for all $x$ the point $t_0$ is a removable singularity. Modify $\alpha_x^*(t)$ on these at most countably many points so that it is holomorphic on the real line. This is an immaterial change for the Fourier transform. Then for every $t \in \R$, there is some $\ve > 0$ such that $\alpha_x^*(t)$ is holomorphic on $B(t, \ve) \subset \C$. This implies that $\alpha_x^*(t)$ is infinitely (complex) differentiable within this ball, which implies the desired result. 
\end{proof}

One of the singularities of $\alpha_x(t)$ is at $t = 0$. In Lemma \ref{L:lt3} we derive explicitly the value of $\alpha_x(t)$ (with singularities removed) at $t = 0$.

\begin{lemma} \label{L:lt3}
One has $\lim_{t \ra 0} \alpha_x(t) = p_x(0) - p_x(1)$.
\end{lemma}
\begin{proof}
We use the notation of Lemma \ref{L:lt2}. It is sufficient to show that $\lim_{t \ra 0} \alpha_x^*(t) = 0$. By application of the dominated convergence theorem it is straightforward to see that $\lim_{t \ra 0} P_x(t) = \lim_{t \ra 0} Q(t) = 0$. We proceed via L'H\^{o}pital's Rule. One has
\begin{align*}
\parr{}{t} P_x(t) = & \int_X i h_y (p_y(0) - p_y(1)) e^{2\pi i h_y t}\, \d\mu(y) \left( e^{-2\pi i h_x t} - \int_X e^{-2\pi i h_z t} \, \d \mu(z) \right) \\
& + \int_X (p_y(0) - p_y(1)) e^{2\pi i h_y t}\, \d\mu(y) \left(- i h_x e^{-2\pi i h_x t} + \int_X i h_z e^{-2\pi i h_z t} \, \d \mu(z) \right)\\
\parr{}{t} Q(t) = & \int_X \int_X i (h_y - h_x) e^{2\pi i (h_y - h_x) t}\, \d \mu(y) \d \mu(z),
\end{align*}
whence once more $\lim_{t \ra 0} \parr{}{t} P_x(t) = \lim_{t \ra 0} \parr{}{t} Q(t) = 0$. However, one final calculation yields: 
\begin{align*}
\pars{}{t} P_x(t) = & -\int_X  h_y^2 (p_y(0) - p_y(1)) e^{2\pi i h_y t}\, \d\mu(y) \left( e^{-2\pi i h_x t} - \int_X e^{-2\pi i h_z t} \, \d \mu(z) \right) \\
& + \int_X (p_y(0) - p_y(1)) e^{2\pi i h_y t}\, \d\mu(y) \left(- h_x^2 e^{-2\pi i h_x t} + \int_X i h_z^2 e^{-2\pi i h_z t} \, \d \mu(z) \right) \\
& + 2 \int_X  i h_y (p_y(0) - p_y(1)) e^{2\pi i h_y t}\, \d\mu(y) \left(- i h_x e^{-2\pi i h_x t} + \int_X i h_z e^{-2\pi i h_z t} \, \d \mu(z) \right).
\end{align*}
The first two lines of the preceding display vanish at $t = 0$. As for the third line, note that by the dominated convergence theorem, 
\begin{align*}
\lim_{t \ra 0} \int_X i h_y (p_y(0) - p_y(1)) e^{2\pi i h_y t} \, \d\mu(y) = \int_X i h_y (p_y(0) - p_y(1)) \, \d\mu(y) = 0
\end{align*}
by orthogonality of $h_y$ and $p_y(0), p_y(1)$. Hence $\lim_{t \ra 0} \pars{}{t} P_x(t) = 0$ but we have already shown in Lemma \ref{L:lt2} that $\lim_{t \ra 0} \pars{}{t} Q(t) \neq 0$. So $\lim_{t \ra 0} \alpha_x^*(t) = 0$ as we require. 
\end{proof}

\begin{lemma}
\label{L:lt4} One has $\alpha_x(t) = \overline{\alpha_x(-t)}$ for all $t \in \R$
\end{lemma}
\begin{proof}
This is clear from the definition of $\alpha_x(t)$; note in particular that 
\begin{align*}
Q(t) = 1 - \int_X e^{2\pi i h_y t} \d\mu(y) \int_X e^{-2\pi i h_y t } \d\mu(y) = 1 - \int_X e^{2\pi i h_y t} \d\mu(y) \overline{\int_X e^{2\pi i h_y t} \d\mu(y) } \in \R,
\end{align*}
so one only has to check the property for the numerator $P_x(t)$. 
\end{proof}

We continue building our function $\gamma$ by introducing two new functions. Let $\phi(t)$ be a smooth, real valued, and compactly supported function such that $\phi(0) = 1$. For some fixed and positive $M \in \R$ let $\psi_M(t)$ be the Fourier transform of a uniform distribution on $[-M, M]$, i.e.\ 
\[ 
\psi_M(t) \equiv \frac{1}{2M} \F\left[ \one_{[-M,M]} \right](t)
\]
We then let $\gamma_x(t) =(\phi(t) \alpha_x(t))\psi_M(t)$. Lemma \ref{L:lt1} applies when $\alpha_x(t)$ is replaced with $\gamma_x(t)$. Note that $\psi_M(0) = \frac{1}{2M}\int_{-M}^M  \, \mathrm{d}w = 1$ so $\gamma_x(0) = p_x (0) - p_x(1)$. Moreover, 
\begin{align*}
\F^{-1} (\gamma_x(t) ) &= \F^{-1} ((\phi(t) \alpha_x(t)) \cdot \psi_M(t) ) \\
& = \F^{-1} (\phi(t) \alpha_x(t)) * \F^{-1} (\phi_M(t)) \\
& = \F^{-1} (\phi(t) \alpha_x(t)) * \left( \frac{1}{2M} \one_{[-M,M]} \right),
\end{align*}
where we have used the fact that the Fourier transform of a convolution is the the product of the Fourier transforms. Because we have the inclusion $\phi(t) \alpha_x(t) \in C_0^\infty(\R)$, $\phi(t) \alpha_x(t)$ is in the Schwartz space so that $\F^{-1} (\phi(t) \alpha_x(t))$ is as well; it follows that $\F^{-1} (\phi(t) \alpha_x(t)) \in L^p(\R)$ for all $p \ge 1$. In particular, $\F^{-1} (\gamma_x(t)) \in L^1(\R)$ because it is the convolution of two $L^1(\R)$ functions. Because $\psi_M(t)$ is the Fourier transform of a real valued function one has $\psi_M(-t) = \overline{\psi_M(t)}$ for every $t \in \R$; as $\phi(t)$ is real valued this implies with Lemma \ref{L:lt4} that $\gamma_x(-t) = \overline{\gamma_x(t)}$ for all real $t$. Hence
\begin{align*}
\overline{\F^{-1} (\gamma_x(t))} & = \overline{ \int_{-\infty}^\infty \gamma_x(t) e^{2\pi i t w} \, \d t} = \int_{-\infty}^\infty \gamma_x(-t) e^{-2 \pi i t w } \, \d t \\
& = \int_{-\infty}^\infty \gamma_x(t) e^{2\pi i w} \, \d t = \F^{-1}(\gamma_x(t)),
\end{align*}
whence $\F^{-1}(\gamma_x(t)) \in \R$. By taking Fourier transforms it is also the case that: 
\begin{align*}
\int_{-\infty}^\infty \F^{-1}[ \phi(t) \alpha_x(t)](w)\, \d w = \F\big[\F^{-1} [ \phi(t) \alpha x(t)]\big](0) = p_x(0) - p_x(1). 
\end{align*}
Now we prove the following result on convolutions:
\begin{lemma} \label{L:lt5}
Let $g \in L^1(\R)$ be a real valued function; then 
\[\lim_{M \ra \infty} \int_{-\infty}^\infty \left| g * \left( \frac{1}{2M} \one_{[-M, M]} \right) (w) \right| \, \d w \ra \left|\int_{-\infty}^\infty g(w) \, \d w \right|.
\]
\end{lemma}
\begin{proof}
Fix any $\ve > 0$. By the density of $C_0^\infty(\R)$ in $L^1(\R)$ we may find a smooth compactly supported function $g_0$ such that $\norm{g - g_0}_1 < \ve$. Hence, 
\begin{align*}
&\left| \int_{-\infty}^\infty \left| g * \frac{1}{2M} \one_{[-M,M]} \right| \, \d w - \int_{-\infty}^\infty \left|g_0 * \frac{1}{2M} \one_{[-M,M]}\right| \, \d w \right| \\
& \quad  = \left| \int_{-\infty}^\infty \left| \frac{1}{2M} \int_{-\infty}^\infty g(y) \one_{[-M,M]}(w-y) \, \d y \right|\, \d w - \int_{-\infty}^\infty \left| \frac{1}{2M} \int_{-\infty}^\infty g_0(y) \one_{[-M,M]}(w-y) \, \d y\right| \, \d w\right| \\
& \quad \le \int_{-\infty}^\infty |g(y) - g_0(y)| \frac{1}{2M} \int_{-\infty}^\infty \one_{-M,M} (w-y) \, \d w \, \d y < \ve. 
\end{align*}
Suppose that $g_0$ is supported on $[-B,B]$. Then for sufficiently large $M$ one has 
\begin{align*}
&\int_{-\infty}^\infty  \left|g_0* \frac{1}{2M} \one_{[-M,M]}\right| \, \d w \\ & \quad = \int_{-\infty}^\infty \frac{1}{2M} \left|\int_{w - M}^{w+M} g_0(y)  \, \d y \right| \, \d w \\
& \quad = \frac{1}{2M}\int_{-B- M}^{B - M}\left|  \int_{w - M}^{w + M } g_0(y) \, \d y\right| \, \d w + \frac{1}{2M} \int_{-B + M }^{B+ M} \left|\int_{w - M}^{w + M } g_0(y)\right| \, \d y \, \d w \\
& \quad \quad + \frac{1}{2M}  \int_{B - M }^{-B+ M} \left| \int_{w - M}^{w+ M } g_0(y) \, \d y \right| \, \d w. 
\end{align*}
Note that $\left|\frac{1}{2M} \int_{-B - M}^{B- M} \left|\int_{w - M}^{w+ M } g_0(y) \, \d y \right| \, \d w\right| \le  \frac{2B}{2M}\norm{g_0(y)}_1 \ra 0 $. Applying similar reasoning to the last two lines of the previous display and the fact that 
\begin{align*}
\frac{1}{2M} \int_{B - M}^{-B + M} \left| \int_{w - M}^{w + M} g_0(y) \, \d y \right| \, \d w  =  \frac{2M - 2B}{2M} \left|\int_{-\infty}^\infty g_0(y) \, \d y\right| \, \d w \ra \left|\int_{-\infty}^\infty g_0(y) \, \d y \right|.
\end{align*}
implies the desired result 
\end{proof}

Now we revert to our original notation, replacing $x$ with $k$ and $\mu$ with scaled counting measure. Let $K_1 \subset K$ denote the set of indices $k$ for which $p_k(0) > p_k(1)$ and $K_2 \subsetneq K$ the subset on which $p_k(0) < p_k(1)$ (recall that we have reduced to the case $p_k(0) \neq p_k(1)$, all $k$). Using Lemma \ref{L:lt5} we may take $M$ so high that for all $k \in K_1$ the inequality 
\[\left| \int_{-\infty}^\infty |\F^{-1}\gamma_k (w)| \, \d w - \left|\int_{-\infty}^\infty \F^{-1}\gamma_k (w) \, \d w\right|  \right| < p_k(1)
\]
so that $\int_{-\infty}^\infty |\F^{-1} \gamma_k(w)| \, \d w < p_k(0)$. Similarly for $k \in K_2$ we may arrange for the integral bound $\int_{-\infty}^\infty |\F^{-1} \gamma_k(w) | \, \d w < p_k(1)$. Now for $k \in K_1$ let the density for $y$ given $X = k$, $W = 0$ be given as
\begin{align*}
f^0_k(w) \equiv \frac{|\F^{-1} \gamma_k(w) | }{\int_{-\infty}^\infty |\F^{-1} \gamma_k(w') | \, \d w'} \ge \frac{|\F^{-1} \gamma_k(w) | }{p_k(0)}
\end{align*}
and let 
\begin{align*}
f^1_k(w) \equiv \frac{1}{p_k(1)} \left( p_k(0)f_k^0(w) - \F^{-1} \gamma_k(w)  \right).
\end{align*}
Immediately one has that $f_k^0(w)$ is a proper density function; moreover, $f_k^1(w)$ is nonnegative and 
\begin{align*}
\int_{-\infty}^\infty f_k^1(w) \, \d w = \frac{1}{p_k(1)} \left( p_k(0) - (p_k(0) - p_k(1))\right) = 1
\end{align*}
so it is a proper density. Repeat the process for $k \in K_2$; set 
\begin{align*}
&f_k^1(w) \equiv \frac{|\F^{-1}\gamma_k(w)|}{\int_{-\infty}^\infty |\F^{-1} \gamma_k(w') | \, \d w} \ge \frac{| \F^{-1} \gamma_k(w)|}{p_k(1)} \\
& f_k^0(w) \equiv \frac{1}{p_k(0)} \left( p_k(1) f_k^1(w) + \F^{-1} \gamma_k(w) \right),
\end{align*} 
where again both functions are proper densities on $\R$. Finally, notice that we have arranged these densities so that for all $k$, 
\begin{align*}
&p_k(0) f_k^0(w) - p_k(1) f_k^1(w) = \F^{-1} \gamma_k(w) \\
&p_k(0) \hf_k^0 - p_k(1) \hf_k^1 = \F\left[ p_k(0) f_k^0 - p_k(1) f_k^1 \right] = \gamma_k. 
\end{align*}
Hence, our densities satisfy the conditions in \eqref{E:ft2}, which is equivalent to \eqref{E:ft1} under the integrability conditions satisfied by our densities, and we are done. 

\subsection{Additional Proofs, Discrete Case}

\begin{proof}[Proof of Lemma \ref{L:a1}]
Let $\Gamma^0$ and $\Gamma^1$ denote arbitrary probability distribution functions for mean $0$ continuously distributed random variables, and let $U\big|(X \in J, W = w) \sim \Gamma_0$ and $U\big|(X \not\in J, W = w) \sim \Gamma_1$ for $w \in \{0,1\}$. Moreover, choose $\delta_0, \delta_1$ nonzero such that $\P{X \in J | W = 0} \delta_0 + \P{X \in J^c | W  = 0} \delta_1 = 0$ (where $J^c$ is the complement of $J$ in $[K]$) and define $\tilde{U}\big|(X \in J, W = w) \sim \tilde{\Gamma}_0$ and $\tilde{U}\big|(X \not\in J, W = w) \sim \tilde{\Gamma}_1$ for $w \in \{0,1\}$, where $\tilde{\Gamma}_\ell (u) \equiv \Gamma_\ell(u +\delta_\ell)$ for $\ell = 0 , 1$ and all $u \in \R$.  Then, one has for all $u \in \R$: 
\begin{align*}
\P{U \le u |W = 0 } = & \Gamma_0(u)\P{X \in J | W = 0} + \Gamma_1(u) \P{X \in J^c | W = 0} \\
= & \P{U \le u | W = 1} \\
\P{\tilde{U} \le u | W = 0} = & \Gamma_0(u + \delta_0) \P{X \in J | W= 0} + \Gamma_1(u + \delta_1) \P{X \in J^c | W = 0} \\
= &  \P{\tilde{U} \le u | W = 1},
\end{align*}
so that $\tilde{U}$ and $U$ are independent of $W$. Moreover, $\E{U} = 0$ and $\E{\tilde{U}} = 0$ by choice of $\delta_0, \delta_1$. Hence, letting $Y = g(X) + U$, one has $Y - g(X) \indep W$ but also $Y - g(X) - \left( \one_{X \in J} \delta_0 + \one_{X \in J^c} \delta_1 \right) = \tilde{U} \indep W$. Hence, $g$ is not point identified by the full independence restriction $Y - g(X) \indep W$, and moreover there are a continuum of $g$ in the identified set corresponding to all possible choices of $\delta_0, \delta_1$. 
\end{proof}

\begin{proof}[Proof of Proposition \ref{P:generic}]
Let $H$ denote the hyperplane of vectors which are orthogonal to the $2K$-vector $v \equiv (p_1(0), \ldots ,p_K(0), -p_1(1), \ldots, -p_K(1))'$. If the lower conditional moments of $U$ lie in $S_{K-1}$ and the vector $\{(\E{U^K| W = 0, X = k})_{k = 1}^K, (\E{U^K|W = 0, X = k })_{k = 1}^K\}$ lies in $H$ then the law of iterated expectations implies that Assumption \ref{A:2} holds. Without loss of generality let $p_K(1)> 0$. Let $U: \R^{2K-1} \ra H$ be defined by 
\begin{align*}
U(x_1, \ldots , x_{2K-1}) = \left(x_1, \ldots , x_{2K-1},  \frac{\sum_{\ell = 1}^{K }p_\ell(0) y_\ell - \sum_{\ell=1}^{K-1} p_\ell(1) y_{K+\ell}}{p_K(1)} \right).
\end{align*}
Clearly $U$ is bijective. Let $\mu \equiv U_* \lambda$ be the pushforward of Lebesgue measure on $\R^{2K-1}$ under $U$. Note that $\mu$ is translation invariant and finite on compact (bounded) sets. Moreover, by equipping $H$ with its relative topology as a subspace of $\R^{2K}$ one readily verifies that $\mu$ is both inner and outer regular on $H$ (see\ \cite{B2007}, Theorem 7.1.7). Importantly, the relative topology on $H$ agrees with the topology on $H$ generated by the metric $d(x, y )^2 = \sum_{\ell = 1}^{2K - 1} (x_\ell - y_\ell)^2$, under which $U$ is an isometry. Hence, Haar's theorem implies that $\mu$ is up to some multiplicative constant the unique Haar measure on $H$. 

Now we show that $\mu(T) = \infty$. This is a consequence of the following lemma: 
\begin{lemma} \label{L:hambuger}
Fix moments $m_0 = 1, \ldots , m_N$ corresponding to a real-valued probability distribution $\mu$ whose support contains at least $\floor{N/2}+1$ points. If $N$ is even then for every $m_{N+1} \in \R$ there is a probability distribution $\mu'$ on $\R$ such with corresponding moments $m_0, \ldots, m_{N+1}$. If $N$ is odd then there is some $L \in \R_+$ such that for all $m_{N+1} \ge L$ there is a probability distribution $\mu'$ on $\R$ corresponding with $m_0, \ldots , m_{N+1}$. 
\end{lemma}
\begin{proof}
From the Hamburger moment problem and Sylvester Criterion (see\ \cite{C1994}, \S X.7) it is well known that $\mu'$ exists if $m_0, \ldots, m_N$ may be extended into a sequence $(m_n)_{n \ge 0}$ such that the Hankel matrices 
\begin{align*}
\Delta_{n} \equiv \begin{pmatrix}
m_0 & m_1 & \cdots  & m_n \\
m_1 & m_2 & \cdots & m_{n+1} \\
\vdots & \vdots & \ddots & \vdots \\
m_n & m_{n+1} & \cdots & m_{2n} 
\end{pmatrix}
\end{align*}
all satisfy $\det(\Delta_{n}) \ge 0$. We claim that for all $n \le \floor{N/2}$ one must have $\det(\Delta_{n}) > 0$. Suppose that this is not the case; then there is a nontrivial vector $x \in \R^{n+1}$ such that
\begin{align*}
0 &= x' \Delta_{n} x = \sum_{i = 0}^{n} \sum_{j = 0}^{n} x_i x_j \int_\R y^{i + j } \, \mathrm{d}\mu(y) = \int_\R \left( \sum_{i=0}^n x_i y^i \right)^2 \, \mathrm{d}\mu(y) ,
\end{align*}
which is impossible, because $\sum_{i=0}^n x_i y^i$ is a polynomial which attains at most $n$ zeros on the $\floor{N/2} + 1 > n$ points of support of $\mu$. Now suppose that $N$ is even; fix any $m_{N+1} \in \R$. One finds that $\det( \Delta_{N/2+1}) = m_{N+2} \det ( \Delta_{N/2}) + \gamma_1 m_{N+1} + \gamma_0$, where $\gamma_1$ and $\gamma_0$ are constants determined by $m_0, \ldots , m_N$. Because $\det(\Delta_{N/2}) > 0$, a large enough choice of $m_{N+2}$ guarantees that $\det( \Delta_{N/2 + 1}) > 0$. Similarly, by choosing iteratively $m_{N+m} \in \R$ for $m$ odd and $m_{N+m}$ large enough for $m$ even, we may extend to a sequence $(m_n)_{n \ge 0}$ satisfying the positivity condition and corresponding to a probability measure on the real number line. Similarly, if $N$ is odd, $\det( \Delta_{N/2 - 1/2})> 0$ and there exists some $L \in \R_+$ such that $\det( \Delta_{N/2 + 1/2}) \ge 0$ whenever $m_{N+1} \ge L$. 

\end{proof}

Now suppose that $K-1$ is even; then for a fixed $\ell, k$ and moment vector $S_{K-1}$, Lemma \ref{L:hambuger} implies that for every $m_{K} \in \R$, there is a distribution for $U|_{W = \ell, X = k }$ with the given moments and also $\E{U^K|W = \ell, X= k } = m_K$. Hence $T = \R^{2K} \cap H = H$ and $\mu(T) = \lambda(\R^{2K-1}) = \infty$. On the other hand suppose that $K-1$ is odd, so that Lemma \ref{L:hambuger} implies that there is a vector $\mathbf{L} \in \R_+^{2K}$ so that $T \supset H \cap \left\{x\in \R^{2K} \ge \mathbf{L}\right\}$ (this fact establishes that $T$ has nonempty interior when it is viewed as a subset of $H$ with its relative topology). Then for positive real numbers $A\le B$, 
\begin{align*}
\mu(T) & = \int_H \one_{x \in T } \, \mathrm{d}U_* \lambda = \int_{\R^{2K - 1}} \one_{U(y) \in T} \, \mathrm{d} \lambda(y)  \\
& = \int_{L_1}^\infty \cdots \int_{L_{2K-1}}^\infty \one\left\{\frac{ \sum_{\ell = 1}^{K }p_\ell(0) y_\ell - \sum_{\ell=1}^{K-1} p_\ell(1) y_{K+\ell} }{p_K(1)} \ge L_{2K} \right\}\, \mathrm{d} \lambda(y_1, \ldots, y_{2K-1})\\
& \ge \lambda \left( \left[A,B\right]^K \times \left[\max\{L_k\}, A - p_k(1) L_{2K} \right]^{K-1} \right)
\end{align*}
and whenever $A,B$ are chosen suitably large the quantity on the last line is strictly positive, and bounded below by $(B-A)^K(A - p_k(1) L_{2K} - \max\{L_k\} )^{K-1}$. Taking $A,B \ra \infty$ it follows that in fact $\mu(T) = \infty$. 

Now we prove the final statement that $\mu(T_0) = 0$ and that $T_0$ is contained in a finite union of translated subspaces of strictly lower dimension than $T$. Let $H_0^\delta$ denote the hyperplane in $\R^{2K}$ which is orthogonal to the vector $(p_1(0) \delta_1, \ldots , p_K(0) \delta_K, -p_1(1) \delta_1, \ldots , -p_K(1) \delta_K)$, for some fixed $\delta \in A(S_{K-1}) \setminus \{0\}$. Recall from the proof of Theorem \ref{T:identified} that the vector is $(\delta_1, \ldots, \delta_K)$ is nonconstant; this can also be seen from the fact that if $\delta_k = \delta$ for all $k$ then $ 0 - 0 = \E{ Y - h_X } - \E{Y - g_X} = \E{\delta_X} = \delta$ for some vector $h \in \R^K$. Hence $H_0^\delta$ is not equal to $H$, and $H_0^\delta \cap H$ is a subspace of dimension $2K - 2$ defined by orthogonality to two linearly independent vectors. Because $U^{-1}$ is a linear map from $H$ to $\R^{2K - 1}$, $U^{-1}( H \cap H_0^\delta)$ is also a linear subspace of dimension at most $2K - 2$, whence it has zero Lebesgue measure. So $\mu(H_0^\delta) = \mu(H_0^\delta \cap H) = 0$. Now from \eqref{E:i2}, $T_0 = \bigcup_{\delta \in A(S_{K-1})} ( H_0^\delta + \alpha_\delta) \cap H$ for some fixed $\alpha_\delta$ which are functions of $P(S_{K-1}, \delta)$. Finally, by shift invariance of $\mu$, $\mu(T_0) \le \sum_{\delta \in A(S_{K-1})} \mu(H_0^\delta + \alpha) = 0$, as $A(S_{K-1})$ is a finite set by Theorem \ref{T:identified}. 
\end{proof}

\begin{proof}[Proof of Lemma \ref{L:unifconv}]
The proof follows straightforwardly from the fact that 
\[
\h{Q}_{m,\ell,k} \ca \E{Y^m|W = \ell, X = k } p_k(\ell)
\]
for fixed $\ell$ and $k$ by the Strong Law of Large Numbers. Let $F \subset \R^K$ be a fixed bounded set. Then we have $\sup_{h \in F}\norm{h} \le R$ for some fixed $R$. The proof is established if we show that $\hP_m$ converges uniformly almost surely to $P_m$ over $F$ for all $m = 0, \ldots , K+1$. Note that we may write
\begin{align*}
\hP_m(h) - P_m(h) = \sum_{s = 1}^K \sum_{t = 0}^k A_{s,t} h_s^t,
\end{align*}
where $A_{s,t} \ca 0$ for all $s,t$ by the strong consistency of the estimators $\h{Q}_{m,\ell,k}$. Hence, the triangle inequality implies 
\begin{align*}
\sup_{h \in F} | \hP_m(h) - P_m(h)| \le  \sum_{s = 1}^K \sum_{t = 0}^k |A_{s,t}| R^t  \ca 0. 
\end{align*}
\end{proof}

\begin{proof}[Proof of Lemma \ref{L:clt}]
By Lemma \ref{L:unifconv}, $\hGamma \ca \Gamma$ uniformly over compact sets. Because $\Gamma$ is a smooth function over $\R^K$, Assumption \ref{A:4} implies that for every $\delta > 0$ there is some $\ve(\delta) > 0$ such that $\inf_{\substack{h: \norm{h- g} \ge \delta\\ \norm{h}  \le R}} \norm{\Gamma(h)} \ge \ve(\delta)$ (if this is not the case, then compactness implies the existence of a zero for $\Gamma$ which is not equal to $g$, a contradiction). With a standard proof, uniform convergence then implies that $\limsup_{M \ra \infty} \norm{\hg - g } \overset{\mathrm{a.s.}}{\le} \delta$; because this must be true for all $\delta > 0$, $\hg \ca g$. 
\end{proof}

\begin{proof}[Proof of Theorem \ref{T:clt}]
Recall that we have defined the bijection $\iota: \{0 , \ldots , K - 1\} \times \{0,1\} \times \{1, \ldots , K\}  \ra \{1, \ldots , 2K^2 \}$ by $\iota(j,\ell,k) = 2Kj + 2k + \ell - 1$. 
In addition, we have defined the functions $\Psi_m: \R^{K} \times \R^{2K^2 + 2}, \, m = 0, \ldots , K-1$ by 
\begin{align*}
\Psi_m(v,w) \equiv\left\{ \begin{array}{ll}
 \sum_{k=1}^K \sum_{j=0}^{1}  \left( w_{\iota(j, 0,k)}/w_{2K^2 + 1} \right) (-v_k)^{m-j} & \text{ if } m = 0\\
 \sum_{k=1}^K \sum_{j=0}^{m} \binom{m}{j} \left( w_{\iota(j,0,k)}/w_{2K^2 + 1} - w_{\iota(j, 1,k )}/ w_{2K^2 + 2} \right) (-v_k)^{m-j}& \text{ otherwise }
\end{array}
\right.
\end{align*}
Letting $w^*$ (respectively $\h{w^*}$) denote the $2K^2 + 2$ vector whose $\ell^\text{th}$ coordinate is given by $C_{\iota^{-1}(\ell)}$ (respectively $\hat{C}_{\iota^{-1}(\ell)}$) for $\ell = 1, \ldots , 2K^2$ and which also satisfies $w^*_{2K^2 + 1} = \P{W = 0}, w^*_{2K^2 + 2} = \P{W = 1}$ (respectively, $\h{w}^*_{2K^2 + 1} = n^{-1} \left( \sum_{i=1}^n \one_{W_i = 0}\right)$ and $\h{w}^*_{2K^2 + 2} = n^{-1} \left( \sum_{i=1}^n \one_{W_i = 1 }\right)$), we have $P_m(h) = \Psi_m(h,w^*)$ and $\h{P}_m(h) = \Psi_m(h,\h{w}^*)$ for $h \in \R^K$ and $m \in\{ 0, \ldots , K-1\}$, by \eqref{E:Qform}. 

Now define $\Psi: \R^{K} \times \R^{2K^2 + 2} \ra \R^K$ to be the vector valued function whose $m^\text{th}$ coordinate is given by $\Psi_m$. We recall that under Assumption \ref{A:3}, $\P{W = 0}, \P{W = 1} > 0$ so that in a neighborhood of $(g,w^*)$, $\Psi(v,w)$ is a well defined rational function, whence continuously differentiable, in all of its $2K^2 + K + 2$ arguments. Moreover, $\Psi_m(g,w^*) = 0$ for all $m \in \{0, \ldots , K-1\}$, and the $K \times K$ matrix $\mathrm{D}_h \Psi(h,w^*)\Big|_{\substack{h = g }} = V$ is invertible by Assumption \ref{A:5}. We now obtain from the Implicit Function theorem (cf.\ \cite{MC1995}, Theorem M.E.1.) that:

\begin{lemma}\label{L:ift}
There exist open neighborhoods $A \subset \R^K$ and $B\subset \R^{2K^2 + 2}$ of $g$ and $w^*$, respectively, and a continuously differentiable vector valued function $\omega$ from $B$ to $A$ satisfying: 
\begin{align*}
&\omega(w^*) = g \\
&\Psi(\omega(w), w ) = 0 \text{ for all } w \in B \\
&\mathrm{D}_w \omega(w)\Big|_{w = w^*} = - \Big(\underbrace{ \mathrm{D}_h \Psi(h,w^*)\Big|_{\substack{h = g }} }_{V}\Big)^{-1} \underbrace{\left( \mathrm{D}_w \Psi(h,w^*) \Big|_{h = g} \right)}_{\Delta},
\end{align*}
and moreover $\omega$ is uniquely determined in that, for $w \in B$, $\Psi(v, w) = 0$ for $v \in A$ only if $v = \omega(w)$. 
\end{lemma}

Now we claim that $\tilde{g} \ca \omega(\h{w^*})$. Assumption \ref{A:3} implies that $\h{w^*} \ca w^*$, so that $\one_{\h{w^*} \in B} \ca 1$. Hence, Lemma \ref{L:ift} implies that $\one_{\omega(\h{w^*}) \in \mathcal{Z}_R(\hat{\Lambda})} \ca 1$, which is to say that almost surely $\omega(\h{w^*}) \in A$ eventually constitutes a zero of the function $\hat{\Lambda}$. The uniqueness part of Lemma \ref{L:ift} also implies that $\mathcal{Z}_R(\hat{\Lambda}) \cap A \ca \{\omega(\h{w^*})\}$, in the sense that the zero set eventually almost surely collapses down to a singleton supplied by the implicit function $\omega$. It remains to show that this particular zero almost surely eventually minimizes the quantity $\norm{\hGamma(x) }$. By Assumption \ref{A:4}, $g$ uniquely minimizes $\Gamma$ in the closed ball $\overline{B}(0, R) \subset \R^K$. Moreover, as in the proof of Lemma \ref{L:unifconv}, for every $\delta > 0$ there is some $\ve > 0$ such that for $x \in \overline{B}(0,R) \setminus B(g,\delta)$ one has $\norm{\Gamma(x) } > \ve$. Let $\delta$ be sufficiently small so that $B(g, \delta) \subset A$. By continuity of $\Gamma$ we may take $\delta'> 0$ sufficiently small relative to $\delta$ so that $\sup_{x \in B(g,\delta')} \norm{\Gamma(x)} < \frac{\ve}{3}$. Now, the uniform convergence result of Lemma \ref{L:unifconv} implies that 
\begin{align} \label{E:convergenceball}
\liminf_{n \ra \infty}\left( \inf_{\norm{x - g} \ge \delta } \norm{\hGamma(x)} - \sup_{ \norm{x - g } \le \delta'} \norm{\hGamma(x)} \right) \overset{\mathrm{a.s.}}{\ge} \frac{\ve}{3} > 0.  
\end{align}
We have already shown that $\mathcal{Z}_R(\hat{\Lambda}) \cap B(g,\delta) \subset \mathcal{Z}_R(\hat{\Lambda}) \cap A \ca \{\omega(\h{w^*})\}$, and in particular the continuity of $\omega(w)$ around $w^*$ and convergence $\h{w^*} \ca w^*$ imply via the Continuous Mapping Theorem that $\one_{\omega(\h{w^*}) \in B(g,\delta')} \ca 1$. Indeed, it is true that $\omega(\h{w^*}) \ca \omega(w^*) = g$. Collecting our results, we have shown that asymptotically and almost surely, $\mathcal{Z}_R(\h{\Lambda})$ contains the element $\{\omega(\h{w^*})\}$ and no other points in the set $B(g, \delta)$, that $\{\omega(\h{w^*})\} \in B(g, \delta')$, and that $\h{\Gamma}$ converges uniformly to $\Gamma$ so that $\omega(\h{w^*})$ is eventually the unique minimizer of $\h{\Gamma}$ in $\mathcal{Z}_R(\h{\Lambda})$, which is to say $\tilde{g} = \omega(\hat{w^*})$. The claim is thus established, and we have only to show that asymptotic normality holds for $\omega(\h{w^*})$. Moreover, the proof of the claim establishes $\tilde{g} \ca g$. 

The asymptotic normality now follows by applying the multivariate delta method and invoking the continuous differentiability of the function $\omega$. To do this we must show that asymptotic normality holds for $\sqrt{n} (\h{w^*} - w^*)$. We write that
\begin{align*}
\sqrt{n}(\h{w^*} - w^*) & = \frac{1}{\sqrt{n}}\sum_{i=1}^n
\begin{pmatrix}
Y_i^0 \one_{W_i = 0, X_i = 1} - \E{Y^0 \one_{W = 0,X = 1} }  \\
Y_i^0 \one_{W_i = 1, X_i = 1} - \E{Y^0 \one_{W = 0, X = 1}} \\
\vdots \\
Y_i^{K-1} \one_{W_i = 1, X_i = K} - \E{Y^{K-1} \one_{W = 1, X = K}   } \\
\one_{W_i = 0} - \P{W = 0}  \\
\one_{W_i = 1}  - \P{W = 1}\\
\end{pmatrix}\\
& \cd N(0,\Omega),
\end{align*}
where we have already defined
\footnotesize
\begin{align*}
\Omega & = \begin{pmatrix}
\var{Y^0 \one_{W = 0, X = 1} } & \cov{ Y^0 \one_{W = 0, X = 1}, Y^0 \one_{W = 1 , X = 1}} & \cdots & \cov{Y^0 \one_{W = 0,X = 1}, \one_{W = 1}} \\
\cov{ Y^0 \one_{W = 0, X = 1}, Y^0 \one_{W = 0 , X = 0}} & \var{Y^0 \one_{W = 1, X = 1} } & \cdots &  \cov{Y^0 \one_{W = 0,X = 1}, \one_{W = 1}}\\
\vdots & \vdots & \ddots & \vdots  \\
\cov{Y^0 \one_{W = 0,X = 1}, \one_{W = 1}} & \cov{Y^0 \one_{W = 0,X = 1}, \one_{W = 1}} & \cdots & \var{\one_{W = 1}} 
\end{pmatrix}.
\end{align*}

\normalsize 

Hence, the multivariate delta method implies that 
\begin{align*}
\sqrt{n} ( \omega(\h{w^*}) - \omega(w^*)) & \cd N\left(0, (\mathrm{D}_w \omega(w^*))\,  \Omega \, (\mathrm{D}_w \omega(w^*))  '\right),
\end{align*}
whence 
\begin{align*}
\sqrt{n} ( \tilde{g} - g) & \cd N\left(0, (\mathrm{D}_w \omega(w^*))\,  \Omega \, (\mathrm{D}_w \omega(w^*))  '\right). 
\end{align*}
Suppose in addition that $\E{Y^{2K-2} \one_{ W = \ell, X = k}}$ exists for all $W \in \{0,1\}, X \in [K]$. By smoothness of $\Psi$ and its derivatives in a neighborhood of $(g, w^*)$ as well as the convergence $\tilde{g} \ca g$, $\h{w^*} \ca w^*$, Lemma \ref{L:ift} implies that we may define the consistent estimator
\begin{align*}
\h{  \mathrm{D}_w \omega(w^*) } \equiv  -\Big(\mathrm{D}_h \Psi(h,\h{w^*})\big|_{\substack{h = \tg }}\Big)^{-1} \left( \mathrm{D}_w \Psi(h,\h{w^*}) \Big|_{h = \tg} \right) \ca \mathrm{D}_w \omega(w^*)
\end{align*}
Moreover, letting $\h{\Omega}$ be the plug-in estimator of $\Omega$ given by
\scriptsize 
\begin{align*}
\h{\Omega}  \equiv & \frac{1}{n}\begin{pmatrix}
\sum_{i=1}^n (Y^0_i \one_{W_i = 0, X_i= 1} )^2 & \sum_{i=1}^n Y_i^0 \one_{W_i = 0, X_i = 1} Y_i^0 \one_{W_i = 1 , X_i = 1} & \cdots & \sum_{i=1}^n Y_i^0 \one_{W_i = 0,X_i = 1}\one_{W_i = 1} \\
\sum_{i=1}^n Y_i^0 \one_{W_i = 0, X_i = 1} Y_i^0 \one_{W_i = 1 , X_i = 1} & \sum_{i=1}^n (Y_i^0 \one_{W_i = 1, X_i = 1})^2 & \cdots &  \sum_{i=1}^n Y_i^0 \one_{W_i = 0,X_i = 1}\one_{W_i = 1}\\
\vdots & \ddots & & \vdots  \\
\sum_{i=1}^n Y_i^0 \one_{W_i = 0,X_i = 1}\one_{W_i = 1} & \sum_{i=1}^n Y_i^0 \one_{W_i = 0,X_i = 1}\one_{W_i = 1} & \cdots & \sum_{i=1}^n\one_{W_i = 1}^2
\end{pmatrix} \\
& - \frac{1}{n^2} \begin{pmatrix}
\left(\sum_{i=1}^n Y^0_i \one_{W_i = 0, X_i= 1} \right)^2 & \sum_{i=1}^n Y_i^0 \one_{W_i = 0, X_i = 1} \sum_{i=1}^n Y_i^0 \one_{W_i = 1 , X_i = 1} &  &  \\
\sum_{i=1}^n Y_i^0 \one_{W_i = 0, X_i = 1} \sum_{i=1}^n Y_i^0 \one_{W_i = 1 , X_i = 1} & \left(\sum_{i=1}^n Y_i^0 \one_{W_i = 1, X_i = 1}\right)^2 &  & \\
\vdots & \ddots & &  \\
\sum_{i=1}^n Y_i^0 \one_{W_i = 0,X_i = 1}\sum_{i=1}^n \one_{W_i = 1} & \sum_{i=1}^n Y_i^0 \one_{W_i = 0,X_i = 1}\sum_{i=1}^n \one_{W_i = 1} & \cdots & \left(\sum_{i=1}^n\one_{W_i = 1}\right)^2
\end{pmatrix},
\end{align*}
\normalsize 
the SLLN implies that $\h{\Omega} \ca \Omega$. Thus, Slutsky's theorem implies the desired result that 
\begin{align*}
\sqrt{n} \left((\h{\mathrm{D}_w \omega(w^*)})\,  \h{\Omega} \, (\h{\mathrm{D}_w \omega(w^*)})  '\right)^{-1/2} (\tg - g) \ca N(0, \mathrm{I}_{K \times K}). 
\end{align*}

\end{proof}

\subsection{Proofs for Partial Identification}
\begin{proof}[Proof of Lemma \ref{L:unitinterval}]
Necessity of the condition 
\begin{align} \label{E:neceq}
\sup_{t \in [0,1]} \left| \E{ e^{it(Y - h(X))}| W = 0} - \E{ e^{it (Y - h(X)}|W = 1}\right|= 0 
\end{align}
for $(Y - h(X)) \indep W$
is obvious so we demonstrate sufficiency. 
By Assumption \ref{A:nine} and $\ell \in \{0,1\}$, $\E{U^m|W =\ell} = \E{U^m} < m!\, \rho^m$ for some fixed constant $\rho$. Hence, for $m$ odd,
\begin{align*}
\E{|U|^m} \le \E{U^{m+1}}^{\frac{m}{m+1}} \le  ((m+1)!)^{\frac{m}{m+1}} \rho^m,
\end{align*}
where by Stirling's formula, 
\begin{align*}
\frac{((m+1)!)^{\frac{m}{m+1}}}{m!} \sim  \sqrt{ \frac{(m+1)^\frac{m}{m+1}}{m} }\frac{(m+1)^m}{m^m} \sim e,
\end{align*}
so that by taking $\rho$ sufficiently large we may in fact suppose that $\E{ |U|^m| W = \ell} = \E{|U|^m} < m! \, \rho ^m$ for all $m$. 
Let $B = \max\{ \sup_{h \in \mathcal{H}} \norm{h}_\infty, 1 \}$, so that for any $h \in \mathcal{H}$ one has the bound 
\begin{align*}
\E{|Y - h(X)|^m | W = \ell } & = \E{ |(U + \delta(X))^m| | W = \ell} \le  \sum_{j = 0}^m \binom{m}{j} \norm{\delta}_\infty^{m - j}\E{ |U|^j | W = \ell}\\
& \le (4B)^m \sup_{j \le m } \E{|U|^j|W = \ell} \le m!\, (4B \rho)^m. 	
\end{align*}
where we have let $\delta = g - h$. 

Now note that for any $h \in \mathcal{H}$, $M \in \N$, and $\xi \in \C$ with $|\xi| < (4B \rho)^{-1}$, one has 
\begin{align*}
\left|\sum_{m = 0}^M \frac{i^m |\xi|^m(Y - h(X)^m }{m!} \right|  & \le  \sum_{m = 0}^\infty \frac{|\xi|^m|Y - h(X)|^m }{m!} 
\end{align*}
where by the monotone convergence series the expected value of the right hand side may be evaluated as a convergent geometric series:
\begin{align*}
\E{\sum_{m = 0}^\infty \frac{|\xi|^m|Y - h(X)|^m }{m!} \, \big| \, W = \ell } = \sum_{m = 0}^\infty  |\xi|^m\E{ \frac{|Y - h(X)|^m}{m!}  \, \big| \, W = \ell} < \infty.
\end{align*}
Hence the dominated convergence theorem implies that on the domain $|\xi| < (4B \rho)^{-1}$ the following is true:
\begin{align*}
\E{ e^{i \xi (Y - h(X))} | W = \ell} & = \sum_{m = 0}^\infty \frac{i^m \xi^m\E{(Y - h(X))^m| W = \ell}}{m!}.
\end{align*}
Theorem 2 of \cite{LS1952} implies that the function $\E{ e^{i \xi (Y - h(X))} \, | \, W = \ell}$ is holomorphic on the horizontal strip $-(4 B \rho)^{-1} < \mathfrak{Im}(\xi) < (4B \rho)^{-1}$. Suppose that \eqref{E:neceq} holds; then $\E{e^{i \xi (Y - h(X))}| W = 0 }-\E{ e^{i \xi ( Y - h(X)) } | W = 1}$ is holomorphic on the same horizontal strip and equal to $0$ on the unit interval $[0,1] \subset \R$. Conclude by power series expansion and analyticity that it vanishes on the strip and therefore on the real line, whence the characteristic functions of $(Y- h(X))|_{W = 0}$ and $(Y - h(X))|_{W = 1}$ are equal, and indeed $(Y- h(X)) \indep W$.

\end{proof}

\begin{proof}
We begin by bounding the bracketing number for $\mathcal{F}$ using Assumption \ref{A:ten}, as almost the same proof suffices for $\mathcal{E}$. Note that for $(t,h,\ell), (t', h', \ell') \in [0,1] \times \mathcal{H}\times \{0,1\}$, we have the Lipschitz bound 
\begin{align*}
&\left| \exp{it (y - h(x))} \one_{w = \ell} - \exp{ i t' (y - h'(x))} \one_{w = \ell'} \right| \\
& \qquad \le 2 | \ell - \ell'| + |t (y - h(x)) - t'(y - h'(x))| \\
& \qquad \le 2 | \ell - \ell'| + |t - t'| y + |t h(x) - t' h'(x)| \\
& \qquad \le 2 |\ell - \ell'| + |t - t'| y + |t - t'| \sup_{h \in \mathcal{H}} \norm{h}_\infty + \norm{h - h'}_\infty \\
& \qquad \le (C + y) ( |t - t'| + \norm{h - h'}_\infty + |\ell - \ell'| ),
\end{align*}
where $C = \max\{2, \sup_{h \in \mathcal{H}} \norm{h}_\infty \}$ is a constant, and we have employed the inequality $|e^{ix} - e^{i y}| \le \int_x^y |i e^{i \xi}| \, \mathrm{d} \xi \le |x - y|$. Let $F(y) \equiv y + C$ and note that $F$ is an envelope funtion for $\F$. By considering the parameter space $[0,1] \times \mathcal{H} \times \{0,1\}$ equipped with metric $d((t,h,\ell), (t',h', \ell')) = |t - t'| + \norm{ h - h'}_\infty + |\ell - \ell'| $ it follows from Theorem 2.7.11 of \cite{VW1996} that for any norm $\norm{\cdot}$, 
\begin{align*}
N_{[\,]}(2\ve \norm{F}, \mathcal{F}, \norm{\cdot}) & \le N(\ve, [0,1] \times \mathcal{H} \times \{0,1\}, d) \\
& \le 2 \ve^{-1} N(\ve/2, \mathcal{H}, \norm{\cdot}_\infty) < \infty. 
\end{align*}
Immediately Theorem 2.4.1 of \cite{VW1996} implies that $\F$ is Glivenko-Cantelli. 
Recall that $Y = g(X) + U$ where $g$ is bounded (as $g \in \mathcal{H}$ was assumed) and all moments of $U$ exist so $\norm{F^2}_{P,2} = \E{(c + Y)^2}^{1/2} < \infty$. Hence, $F$ has a second moment. Let $D = 2\norm{F}_{P,2}$, so that for $\ve > D$ one has $N_{[\, ] }(\ve, \F, L_2(P)) = 1$ (namely, take the bracket $[-F, F]$). This allows us to write: 
\begin{align*}
&\int_0^\infty \sqrt{ \log N_{[\,]}(\ve , \F, L_2(P) )} \, \mathrm{d} \ve + \int_0^\infty \sqrt{\log N(\ve, \F, L_2(P))   } \, \mathrm{d} \ve \\
& \qquad \le 2\int_0^D \sqrt{ \log N_{[\,]}(\ve , \F, L_2(P) )} \, \mathrm{d} \ve \\
& \qquad  \le 2 \int_0^D \sqrt{ \log 4 \norm{F}_{P,2} \ve^{-1} N\big( \frac{\ve}{4 \norm{F}_{P,2}}, \mathcal{H}, \norm{\cdot}_\infty \big)} \, \mathrm{d} \ve  \\
& \qquad \lesssim 1 + \int_0^\infty \sqrt{ \log N(\ve , \mathcal{H}, \norm{\cdot}_\infty)}  \, \mathrm{d} \ve  < \infty ,
\end{align*}
where we have used the inequalities $N(\ve, \F, \norm{\cdot}) \le N_{[\, ]}(2 \ve, \F, \norm{\cdot})$ and $\sqrt{a + b} \le \sqrt{a} + \sqrt{b}$, for $a,b \ge 0$. Theorem 2.5.6 of \cite{VW1996} thus concludes for $\F$. 
\end{proof}

\begin{proof}
By Lemma \ref{L:covering}, the convergence 
\begin{align}
\sqrt{n} ( X_n  - \theta) \equiv \sqrt{n} \left( n^{-1} \sum_{i= 1}^n e^{it(Y_i - h(X_i))} \one_{W_i = \ell} - \E{e^{it(Y - h(x))} \one_{W = \ell}  } \right) \rs \G
\end{align}
 holds over all $t,h, \ell$, where the element $\theta \in \ell^\infty(\F)$ is defined as $\phi(t, h, \ell) = \E{ e^{it (Y - h(X))} \one_{W = \ell}}$, and $X_n \equiv n^{-1} \sum_{i= 1}^n e^{it(Y_i - h(X_i))} \one_{W_i = \ell}$ is a random variable taking values in $\ell^\infty(\F)$. $\G$ is a tight Borel measurable zero-mean Gaussian element in $\ell^\infty(\F)$ and $\F = [0,1] \times \H \times \{0,1\}$. Let $\F'$ be a copy of $\F$ and define a function $\phi: L^\infty(\F) \supset D_\phi \ra L^\infty(\F \cup \F' )$ by the piecewise relation:
\begin{align*}
\phi(f)(\gamma) = \left\{ \begin{array}{l l}
f(t, h , \ell)&  \text{ if }\gamma = (t, h , \ell) \in \F \\
\frac{\E{ e^{it' (Y - h'(X))} \one_{W = \ell'} }}{ f(0,g,\ell') }& \text{ if } \gamma =  ( t', h', \ell') \in \F'
\end{array} \right. 
\end{align*}
where the domain of $\phi$ is defined as $D_\phi = \{ f \in L^\infty(\F): \min_{\ell \in \{0,1\}} | f(0,g,\ell) | > 0 \}$. Set 
\begin{align*}
\phi'_\theta(f)(\gamma) & = \left\{ \begin{array}{ l l}
f(t, h, \ell) & \text{ if } \gamma = (t, h , \ell) \in \F \\
- \frac{\E{e^{i t' (Y - h'(X))} \one_{W = \ell'}}}{\P{W = \ell'}^2}f(0,g,\ell') & \text{ if } \gamma = (t', h', \ell') \in \F'
\end{array}\right. ,
\end{align*}
which is clearly a bounded linear map (in its first argument). 
We claim that $\phi$ is Hadamard differentiable at $\theta$ (see \cite{VW1996} \S 3.9) with derivative $\phi_\theta'$. Indeed, for any bounded set $K \subset \ell^\infty(\F)$ and $\alpha \in \R^+$ with $\alpha \cdot  \sup_{f \in K} \norm{f}_{\ell^\infty} < \min \{ \P{W = 0}, \P{W = 1}\}$, we have 
\small
\begin{align*}
&\sup_{f \in K} \norm{ \frac{\phi(\theta + \alpha f ) - \phi(\theta)}{\alpha} - \phi'_\theta(f)   }_{\ell^\infty(\F \cup \F')} \\
&\quad \le  \sup_{f \in K} \sup_{\gamma \in \F} \left|\frac{\phi(\theta + \alpha f )(\gamma) - \phi(\theta)(\gamma)}{\alpha} - \phi'_\theta(f)  \right|  + \sup_{f \in K} \sup_{ \gamma \in \F'} \left| \frac{\phi( \theta + \alpha f) (\gamma) - \phi(\theta)(\gamma)}{\alpha} - \phi_\theta'(f) \right|  \\
& \quad \le \sup_{f \in K} \sup_{\gamma \in \F'} \Big| \E{e^{it' (Y - h'(X))} \one_{W = \ell'} }\Big| \left| \alpha^{-1} \left( \frac{1}{ \theta( 0 , g, \ell') + \alpha f(0, g, \ell')} - \frac{1}{ \theta( 0 , g, \ell')} \right) + \frac{f(0,g,\ell')}{\P{W = \ell'}^2} \right| \\
&\quad \le \sup_{\gamma \in \F'} \left| \alpha^{-1} \left( \frac{1}{\P{W = \ell'} + \alpha f(0, g, \ell')} - \frac{1}{\P{W = \ell'}} \right) + \frac{f(0, g, \ell')}{\P{W = \ell'}^2} \right| = O( \alpha),
\end{align*} 
\normalsize
where in the last line we have employed Taylor expansion of the function $x \mapsto \frac{1}{\P{W = \ell'} + x}$ for $\ell' \in \{0,1\}$. Taking $\alpha \ra 0$, equation (3.9.1) of \cite{VW1996} implies that $\phi$ is Hadamard differentiable at $\theta$, and hence Theorem 3.9.4 yields the following convergence: 
\begin{align*}
\sqrt{n} (\phi(X_n) - \phi(\theta)) \rs \phi'_\theta ( \G),
\end{align*}
where $\phi_\theta'(\G)$ is a tight Borel-measurable and zero-mean Gaussian process in $\ell^\infty(\F \cup \F')$, owing to continuity and linearity of $\phi_\theta'$. 

Now letting $Y_n$ be a $\ell^\infty(\F \cup \F')$-valued random variable with values given by 
\begin{align*}
Y_n (\gamma) = \left\{ \begin{array}{ l l}
\left( n^{-1} \sum_{i=1}^n \one_{W_i = \ell} \right)^{-1} & \text{ if } \gamma = (t, h, \ell) \in \F \\
1 & \text{ if } \gamma \in \F'
\end{array}\right.
\end{align*}
It is clear that for $\ell \in \{0,1\}$ we have the convergence
\begin{align*}
Y_n \rs \psi,
\end{align*}
where $\psi(\gamma) \equiv \left\{\begin{array}{l l}
\P{W = \ell}^{-1} & \text{ if } \gamma \in \F \\
1 & \text{ if } \gamma \in \F'
\end{array}\right.$ and $\psi$ is a constant element of $\ell^\infty(\F \cup \F')$. Hence Slutsky's Theorem (\cite{VW1996}, pp.\ 32) implies that, under pointwise multiplication denoted by $\cdot$,
\begin{align} \label{E:y1}
\sqrt{n} \, Y_n \cdot ( \phi(X_n) - \phi(\theta) ) \rs \psi \cdot \phi'_\theta (\G) 
\end{align}
where the right side remains a tight zero-mean Gaussian process (pointwise multiplication is a bounded and continuous operator on $\ell^\infty(\F)$). Finally, define the continuous (by the triangle inequality) linear map $\rho: \ell^\infty(\F\cup \F') \ra \ell^\infty(\F_0)$ by
\begin{align*}
\rho(f)(t, h) \equiv f(t,h,0)|_{\F} - f(t,h,1)|_{\F}  + f(t,h,0)|_{\F'} - f(t,h,1)|_{\F'},
\end{align*}
where $f|_{\F}$ indicates that $f$ is to be evaluated as a function over $\F$, and similarly $f|_{\F'}$ is the restriction of $f$ to $\F'$. One at last has the convergence 
\begin{align*}
\rho \left(\sqrt{n} \, Y_n \cdot ( \phi(X_n) - \phi(\theta))\right) \rs \rho ( \psi \cdot \phi_\theta'(\G))\equiv \mathbb{D},
\end{align*}
where the right side is tight mean-zero Gaussian process in $\ell^\infty(\F_0)$ as desired. Unwinding our notation, find that 
\small
\begin{align*}
&\rho \left(\sqrt{n} \, Y_n \cdot ( \phi(X_n) - \phi(\theta))\right) (h,t) \\
&\quad = \sqrt{n} \Big( \frac{\phi(X_n)(t,h,0)|_\F - \phi(\theta)(t,h,0) }{n^{-1} \sum_{i=1}^n \one_{W_i = 1}} -  \frac{\phi(X_n)(t,h,1)|_\F - \phi(\theta)(t,h,1)}{n^{-1} \sum_{i=1}^n \one_{W_i = 0}}   \\
& \qquad  \quad + \phi(X_n)(t,h,0)|_{\F'} - \phi(\theta)(t, h, 0)|_{\F'} - \phi(X_n) (t, h , 1)|_{\F'}  + \phi(\theta)(t,h,1)|_{\F'}\Big)  \\
&\quad = \sqrt{n} \Bigg( \frac{ \EE{n}{ e^{it (Y - h(X))} \one_{W = 0}} - \E{e^{it (Y- h(X))}\one_{W = 0}}}{ n^{-1} \sum_{i=1}^n \one_{W_i = 0} } - \frac{ \EE{n}{e^{it (Y - h(X))} \one_{W = 1}} - \E{e^{it (Y- h(X))}\one_{W = 1}}}{ n^{-1} \sum_{i=1}^n \one_{W_i = 1} }  \\
&\qquad \quad + \frac{\E{ e^{it (Y - h(X))} \one_{W = 0}}}{n^{-1} \sum_{i=1}^n \one_{W_i = 0}} - \frac{\E{e^{it(Y - h(X))} \one_{W = 0}}}{\P{W = 0}} - \frac{\E{ e^{it (Y - h(X))} \one_{W = 1}}}{n^{-1} \sum_{i=1}^n \one_{W_i = 1}} + \frac{\E{e^{it(Y - h(X))} \one_{W = 1}}}{\P{W = 1}} \Bigg).
\end{align*}
\normalsize
Canceling like terms and applying Bayes' rule results in the left side of \eqref{E:y0}, so the proof is complete. 

\end{proof}

\begin{proof}
Let $\tilde{\H}$ denote a countable $\norm{\cdot}_\infty$-dense set in $\H$ (which exists by Assumption \ref{A:ten}) and note that the dominated convergence theorem implies that 
\small
\begin{align*}
&\left\{ \H_0 \subset \hat{\H}_n \right\}\\
&\quad = \left\{ \sup_{h \in \H_0} |\EE{n}{Y - h(X)}| \le \eta_n \text{ and } \sup_{\substack{t \in [0,1]\\ h \in \H_0}} \left| \EE{n}{e^{it (Y - h(X))}|W = 0} - \EE{n}{e^{it (Y - h(X))}|W = 1} \right| \le \eta_n \right\} \\
& \quad = \bigcap_{h \in \H_0 \cap \tilde{\H}} \left\{ |\EE{n}{Y -h(X)} | \le \eta_n \right\} \cap \bigcap_{\substack{t \in [0,1]\cap \Q \\ h \in \H_0 \cap \tilde{\H}}}  \left\{ \left| \EE{n}{e^{it (Y - h(X))}|W = 0} - \EE{n}{e^{it (Y - h(X))}|W = 1} \right| \le \eta_n \right\},
\end{align*}
\normalsize
which is a measurable set by Assumption \ref{A:ten}. Similar arguments show that we may use probability notation $\mathrm{P}$ in the place of outer integration notation $\mathrm{P}^*$, utilizing measurability owing to separability. 

To prove the first claim, apply Lemma \ref{L:covering} to infer that $\sqrt{n} \left( \EE{n}{Y - h(X)} - \E{Y - h(X)} \right) \rs \mathbb{G}_{\mathcal{E}}$ for some tight Gaussian process $\mathcal{E}$ supported in $\ell^\infty(\mathcal{E})$. By the continuous mapping theorem and Portmanteau theorem (\cite{VW1996}, Theorem 1.3.4),
\begin{align*}
&\liminf_{n \ra\infty}\P{ \sup_{h \in \H} \left|\EE{n}{Y - h(X)}  - \E{Y - h(X)}\right| < \eta_n } \\
& \qquad \ge \sup_{\delta > 0}\liminf_{n \ra \infty} \P{  \sup_{h \in \H} \sqrt{n}\left| \EE{n}{Y - h(X)} - \E{Y - h(X)} \right| < \delta} \\
& \qquad  \ge \sup_{\delta > 0} \P{ \norm{\mathbb{G}_\mathcal{E}}_{\ell^\infty(\mathcal{E})}< \delta } = 1,
\end{align*}
with the latter equality owing to tightness of $\mathbb{G}_{\mathcal{E}}$. Similar application of Proposition \ref{P:donsker} shows that 
\begin{align*}
&\lim_{n \ra \infty} \mathrm{P}\Bigg( \sup_{\substack{t \in [0,1] \\ h \in \H}}  \Bigg|\EE{n}{e^{it(Y - h(X))}|W = 0}  - \EE{n}{e^{it(Y - h(X))}|W = 1} \\
& \qquad \qquad - \E{e^{it(Y - h(X))}|W = 0}  + \E{e^{it(Y - h(X))}|W = 1}  \Bigg| < \eta_n  \Bigg) = 1. 
\end{align*}
By definition of $\H_0$ this implies $\P{ \H_0 \subset \hat{\H}_n} \ra 1$, and considering $n$ large enough such that $\alpha_n > 2 \eta_n$ establishes that $\P{ \H_{\alpha_n} \cap \hat{\H}_n = \emptyset} \ra 1$. 

Obtaining convergence $\mathrm{P}$-almost surely involves some tail estimates from \cite{VW1996} and the Borel-Cantelli lemma. Assuming Assumption \ref{A:nine} and applying the argument of Lemma \ref{L:unitinterval} (for example, assuming $0 \in \H$) it is straightforward to see that the characteristic function $\E{e^{it Y}}$ exists in a neighborhood of the origin in $\C$, so that for some $\delta > 0$, 
\begin{align*}
\E{e^{\delta|Y|}} \le \E{e^{ -\delta Y}} + \E{ e^{\delta Y}} < \infty . 
\end{align*}
Hence, Markov's inequality implies that $\P{| Y| \ge \rho} \le \E{e^{\delta|Y|}} e^{-\delta \rho}$ for any $\rho \ge 0$. 

Now let $B = \sup_{ h \in \H} \norm{h}_\infty$ and set $\beta_n = n^{1/4 - \gamma/2}$. Define
\[
\F_n = \left\{ ( \beta_n + B)^{-1} e^{it (y - h(x)) \one_{|y| \le \beta_n}} \one_{w = 0}: h \in \H, t \in [0,1]\right\}
\]
be a class of functions. We claim that for every $\ve > 0$, $N\left( \ve, \F_n, \norm{\cdot}_\infty \right) \le 3\ve^{-1} N \left( \ve/2, \H, \norm{\cdot}_\infty \right) + 1$. Indeed, for $t,t' \in [0,1]$ and $h,h' \in \H$, 
\small
\begin{align*}
\left| t(y - h(x)) - t'( y - h'(x)) \right| & \le |t (h(x) - h'(x)| + |(t - t') (y - h'(x))| \le \norm{h - h'}_\infty + |t - t'|\left( |y| + B\right), 
\end{align*}
\normalsize
so that, taking $\H_{\ve/2}$ to be an $\ve/2$-cover of $\H$ and $\mathcal{T}_{\ve/2} $ to be a $\ve/2$-cover of $[0,1]$ of size at most $3/\ve$, the set 
\begin{align*}
\F_{n,\ve} \equiv \left\{ (\beta_n +B)^{-1} e^{it (y - h(x)) \one_{|y| \le \beta_n}} : h \in \H_{\ve/2}, t \in \mathcal{T}_{\ve/2} \right\} \cup \{0\}
\end{align*}
can readily be seen using Lipschitz-ness of the complex exponential function to constitute an $\ve$-cover of $\F_n$ of size at most $3\ve^{-1} |\H_{\ve/2}| + 1$ as desired. By assumption $\log N(\ve, \H, \norm{\cdot}_\infty ) \lesssim \ve^{-\omega}$, so by taking $K$ large enough we may write 
\begin{align*}
\log N(\ve, \F_n, \norm{\cdot}_\infty) \le K e^{-\omega},
\end{align*}
where $\omega \in (0,1/2)$ and the constant $K$ is uniform in $n$. Let $c> 0$ be an arbitrary constant. Conclude by Theorem 2.14.10 of \cite{VW1996} that for a constant $C$ depending only on $\omega$ and $K$, 
\small
\begin{align*}
&\P{ \sup_{\substack{t \in [0,1]\\ h \in \H}} \left| \EE{n}{e^{it(Y - h(X))\one_{|Y| \le \beta_n}} \one_{W = 0 } } - \E{ e^{it (Y - h(X)) \one_{|Y| \le \beta_n}} \one_{W = 0}} \right| > c\eta_n/2} \\
& \qquad  = \P{ \sup_{\substack{t \in [0,1]\\ h \in \H}} \sqrt{n}\left| \EE{n}{(\beta_n + B)^{-1} e^{it(Y - h(X)) \one_{|Y| \le \beta_n}} \one_{W = 0 } } - \E{ (\beta_n + B)^{-1} e^{it (Y - h(X))\one_{|Y| \le \beta_n}} \one_{W = 0}} \right| > \frac{ c \sqrt{n} \eta_n}{2(\beta_n + B)}} \\
&\qquad \le C \exp{- \left( \frac{c \sqrt{n} \eta_n}{2(\beta_n + B)} \right)^2} = C \exp{- \left( \frac{c \sqrt{n} \eta_n}{2(\beta_n + B)} \right)^2} = O \left( \exp{ - c n^{1/2 - \gamma}/4  } \right). 
\end{align*} 
\normalsize 
Moreover, 
\small
\begin{align*}
\left|\E{e^{it (Y - h(X)) \one_{|Y| \le \beta_n}} \one_{W = 0}  } - \E{ e^{it (Y - h(X)) } \one_{W = 0}  } \right| & \le \E{ \left(e^{it (Y - h(X)) \one_{|Y| \le \beta_n}} - e^{it (Y - h(X)) }\right) \one_{W = 0}  } \\
& \le 2 \P{ |Y| > \beta_n} \le 2 \E{ e^{\delta |Y|}} e^{ - \delta \beta_n}, 
\end{align*}
\normalsize
which decays faster than $c \eta_n/2$. Conclude that 
\small
\begin{align*}
\P{ \sup_{\substack{t \in [0,1]\\ h \in \H}} \left| \EE{n}{e^{it(Y - h(X))\one_{|Y| \le \beta_n}} \one_{W = 0 } } - \E{ e^{it (Y - h(X)) } \one_{W = 0}} \right| > c\eta_n} = O \left( \exp{ - c' n^{1/2 - \gamma}/4  } \right),
\end{align*}
\normalsize
where the constant $c'$ depends on $c$ and the ratio sequence $\eta_n n^{\gamma}$. 
Hence, the union bound implies
\small
\begin{align*}
&\P{ \sup_{\substack{t \in [0,1]\\ h \in \H}} \left| \EE{n}{e^{it(Y - h(X))} \one_{W = 0 } } - \E{ e^{it (Y - h(X)) } \one_{W = 0}} \right| > c\eta_n}  \\
& \qquad\le  \P{ \max_{1 \le i \le n } |Y_i| \ge \beta_n}  +\P{ \sup_{\substack{t \in [0,1]\\ h \in \H}} \left| \EE{n}{e^{it(Y - h(X))\one_{|Y| \le \beta_n}} \one_{W = 0 } } - \E{ e^{it (Y - h(X)) } \one_{W = 0}} \right| > c\eta_n} \\
& \qquad \le n\E{ e^{\delta |Y|} } e^{ - \delta \beta_n}  +  \exp{ - c' n^{1/2 - \gamma}/4 }.
\end{align*}
\normalsize
Summing over $n$ and applying the Borel-Cantelli lemma implies that $\mathrm{P}$-almost surely, 
\small
\[
 \sup_{\substack{t \in [0,1]\\ h \in \H}} \left| \EE{n}{e^{it(Y - h(X))} \one_{W = 0 } } - \E{ e^{it (Y - h(X)) } \one_{W = 0}} \right| > c\eta_n \text{ only finitely often.} 
 \] 
\normalsize
Similarly, the same holds when one replaces $\one_{W = 0 }$ with $\one_{W = 1}$. Because $\P{ W  = 0}, \P{W = 1} > 0$, straightforward application of Hoeffding's inequality and the Borel-Cantelli Lemma implies that for $\ell \in \{0,1\}$,
\small
\begin{align*}
\left|\frac{1}{\EE{n}{\one_{W = \ell}}} -  \frac{1}{\E{ \one_{W = \ell}}} \right| > c \eta_n \text{ only finitely often}.
\end{align*}
\normalsize
By the triangle inequality, 
\small
\begin{align*}
&\sup_{\substack{t \in [0,1] \\ h \in \H}}  \Bigg|\EE{n}{e^{it(Y - h(X))}|W = 0}  - \EE{n}{e^{it(Y - h(X))}|W = 1}  - \E{e^{it(Y - h(X))}|W = 0}  + \E{e^{it(Y - h(X))}|W = 1}  \Bigg|  \\
& \qquad \le \sum_{\ell =0}^1 \Bigg|\EE{n}{e^{it(Y - h(X))}|W = \ell}   - \E{e^{it(Y - h(X))}|W = \ell}    \Bigg| \\
& \qquad  \le \sum_{\ell = 0}^1 \left|\EE{n}{e^{it(Y - h(X))} \one_{W = \ell}}\right|\left|\frac{1}{\EE{n}{\one_{W = \ell}}} -  \frac{1}{\E{ \one_{W = \ell}}} \right| + \left|\frac{\EE{n}{e^{it(Y - h(X))} \one_{W = \ell}} - \E{e^{it(Y - h(X))} \one_{W = \ell}}}{\E{\one_{W = \ell}}} \right|,
\end{align*}
\normalsize
and algebra shows that the last line exceeds $\frac{4 c\eta_n}{\min\{\P{W = 0}, \P{W = 1} \}}$ only finitely often, $\mathrm{P}$-almost surely. As the constant $c$ was arbitrary, the first line of the previous display exceeds $\eta_n$ only finitely often, almost surely. Similar application of the same Theorem 2.14.10 and Borel-Cantelli Lemma implies that 
\begin{align*}
\sup_{h \in \H} | \EE{n}{Y - h(X)} - \E{Y - h(X)}|  > \eta_n \text{ only finitely often},
\end{align*}
almost surely. The second claim then follows from arguments used to establish the first claim in the presence of convergence only in probability.

\end{proof}

\subsection{Results when $X$ is continuous}

\begin{proof}[Proof of Lemma \ref{L:operator}]

Fix $\ve > 0$. Boundedness of the operator is clear. For the first claim, note that for some $K$ sufficiently large, $h \in L^\infty(\mathcal{X} \times \R )$ implies
\begin{align*}
&\limsup_{t' \ra t} \int_\mathcal{X} \int_0^{2B} \left|h(x,u)( f_0(x, t'+u) - f_1(x, t'+u)) \right|\, \mathrm{d}u \, \mathrm{d}x \\
&\quad  \le \limsup_{t' \ra t}\Big[\int_\mathcal{X} \int_0^K \one_{u \in [0,2B]} |h(x,u)| |f_0(x,t' + u) - f_1(x,t'+u)| \, \mathrm{d} u \, \mathrm{d} x \\
& \quad \qquad +\norm{h}_{L^\infty(\X \times \R )} \int_\mathcal{X} \int_K^\infty |f_0(x,t' + u) - f_1(x,t' + u)| \, \mathrm{d} u\Big] \, \mathrm{d}x \\
&\quad \le \limsup_{t' \ra t}\int_\mathcal{X} \int_0^K \one_{u \in [0,2B]} |h(x,u)| |f_0(x,t' + u) - f_1(x,t'+u)| \, \mathrm{d} u \, \mathrm{d} x  + \ve \norm{h}_{L^\infty(\X \times \R )}  \\
&\quad  = \int_\X \int_0^K \one_{u \in [0,2B]} |h(x,u)| |f_0(x,t + u) - f_1(x,t+u)| \, \mathrm{d} u\,  \mathrm{d}x  + \ve \norm{h}_{L^\infty(\X \times \R )}\\
& \quad \le \int_\X \int_0^B  |h(x,u)| |f_0(x,t + u) - f_1(x,t+u)| \, \mathrm{d} u\,  \mathrm{d}x  + \ve \norm{h}_{L^\infty(\X \times \R )},
\end{align*}
where we have used the dominated convergence theorem and compactness of $\X \times [0 , K]$. 
Repeating the argument with $\limsup$ replaced by $\liminf$ and inequalities reversed, and taking $\ve $ arbitrarily small, one quickly observes that 
\begin{align*}
&\lim_{t' \ra t} \int_\X \int_0^{2B} |h(x,u) (f_0, x,t' + u) - f_1(x,t' + u) | \, \mathrm{d} u \, \mathrm{d} x \\
&\quad =\int_\X \int_0^{2B} |h(x,u) (f_0, x,t + u) - f_1(x,t + u) | \, \mathrm{d} u \, \mathrm{d} x.
\end{align*}
Scheff\'{e}'s lemma then concludes.

For the second, note that compactness of $\mathcal{X}$ ensures that $h \in L^1(\mathcal{X} \times [0,2B])$ (using $\lambda(\mathcal{X})< \infty$). Also $f_0^t - f_1^t$ is bounded in $\mathcal{X} \times [-\ve, 2B + \ve]$ for any $\ve > 0$; then, apply continuity of $f_0, f_1$, and the dominated convergence theorem. 

The third claim follows straightforwardly by the Cauchy-Schwarz inequality and Fubini's theorem: 
\begin{align*}
\int_\R |(Th)(t)|^2 \, \mathrm{d}t & = \int_\R \left|\int_\mathcal{X}  \int_0^{2B} h(x,u) (f_0(x, t + u) - f_1(x, t + u) ) \, \mathrm{d} u \, \mathrm{d} x \right|^2 \, \mathrm{d}t \\
& \le \int_\R \left( \int_\mathcal{X} \int_0^{2B} h(x,u)^2 \, \mathrm{d} u \, \mathrm{d} x \right) \left( \int_\mathcal{X} \int_0^{2B} (f_0(x,t + u ) - f_1(x, t + u ) )^2 \, \mathrm{d} u \, \mathrm{d} x \right) \, \mathrm{d} t \\
& \le \norm{h}^2_{L^2(\mathcal{X} \times \R )}  \int_0^{2B} \int_\mathcal{X} \int_\R (f_0(x,t + u) - f_1(x, t + u))^2 \, \mathrm{d} t\, \mathrm{d} x  \, \mathrm{d} u \\
& =2B \norm{h}^2_{L^2(\mathcal{X} \times \R )} \norm{f_0 - f_1}_{L^2(\mathcal{X} \times \R )} < \infty 
\end{align*}
\end{proof}

\begin{proof}[Proof of Lemma \ref{L:indepequiv}]
The first direction of implication has already been established. Conversely, suppose that $\mathcal{V} \subset \ker{T}$. For any fixed $t \in \R$, this implies that 
\begin{align*}
\int_\mathcal{X} \int_t^{t + 2B} h(u) (f_0(x,  u ) - f_1(x,  u )) \, \mathrm{d} u \, \mathrm{d} x  = 0 
\end{align*}
whenever $h \in L^2(\R)$ and is supported on the same set as is the random variable $U$. Suppose that for some $t$, $t+B \in \supp{U}$. By continuity of $f_U(u)$, suppose without loss of generality that $[t+B, t + B + \ve] \subset \supp{U}$ for $\ve$ small enough (otherwise, the inclusion is satisfied for $[t+B - \ve, t+B]$)  . Letting $h_\ve \equiv \frac{1}{\ve} \one_{u \in [t + B , t + B + \ve ]}$, note that for any particular $x$ one has $\lim_{\ve \ra 0 } \int_t^{t + 2B} h_\ve(u) (f_0(x,u) - f_1(x,u)) = f_0(x, t + B) - f_1(x, t + B)$ by continuity of the density functions $f_w$. Noting that both $f_0$ and $f_1$ are bounded in the compact set $\mathcal{X} \times [t, t + 2B]$ and applying the dominated convergence theorem, one thus has 
\begin{align*}
f_{0,U}(t + B) - f_{1,U}(t + B) & = \int_\mathcal{X} (f_0(x, t + B) - f_1(x, t+B)) \, \mathrm{d} x \\
& = \lim_{\ve \ra 0 } \int_{\mathcal{X}} \int_t^{t + 2B} h_\ve(u) (f_0(x,u) - f_1(x,u)) \, \mathrm{d} u \, \mathrm{d} x  = 0,
\end{align*}
where we have used $f_{w,U}(u)$ to denote the marginal density function of $U$ given $W = w$. Letting $t$ vary over $\R$, one has $f_{0,U} = f_{1,U}$ and hence $U \indep W$. 
\end{proof}

\begin{proof}[Proof of Theorem \ref{T:contident}]
Suppose that Assumption \ref{A:C3}(i) holds with the other stated assumptions, and for the sake of contradiction suppose that $g$ is not point identified, so that there is some function $h \neq g$ such that: 
\begin{align*}
Y - g(X) \indep W \\
Y - h(X) \indep W
\end{align*}
and $g \neq h$. As $\E{U } = 0$ is stipulated, $h - g $ is nonconstant. Equivalently, for all $t \in \R$, 
\begin{align*}
\P{ Y - g(X) \le t | W = 0 } - \P{ Y - g(X) \le t | W = 1}  = 0 \\
\P{Y - h(X) \le t | W = 0} - \P{Y - h(X) \le t | W = 1 } = 0.
\end{align*}
Denoting $\delta\equiv h - g$, we have the relations 
\begin{align} \label{E:C1}
&\P{ U \le t |W = 0} - \P{U \le t | W = 1} = 0 \\
&\P{U + \delta(X) \le s | W = 0 } - \P{U + \delta(X) \le s| W = 1} = 0 \nonumber
\end{align}
for all pairs $(t,s) \in \R^2$.
Subtracting the first line of \eqref{E:C1} from the second and applying the law of iterated expectations, one therefore has 
\begin{align} \label{E:C2}
&\int_\mathcal{X} \Big(\big[ \P{U + \delta(x) \le t|W = 0, X = x} - \P{U \le s | W = 0, X = x}\big] f_0(x)\\ 
& \qquad - \big[ \P{U + \delta(x) \le t  | W = 1 , X = x}  - \P{U \le s| W = 1, X = x} \big] f_1(x)\Big) \, \mathrm{d}x = 0, \nonumber
\end{align}
where we have let $f_w(x)$ denote the marginal density of $X$ given $W = w$. Note that by translating the function $\delta$ by $m_\delta \equiv \max_{x \in \mathcal{X}} \delta(x)$, we have 
\begin{align*}
\P{U + \delta(x) \le t | W = w, X= x} = \P{U \le t - m_\delta - (\delta(x) -m_\delta) | W = w, X = x} 
\end{align*}
for all $w, x$, so that by replacing $s = t - m_\delta$ in \eqref{E:C2} and letting $t$ vary over the real numbers, the preceding display implies that for all $t \in \R$, 
\begin{align}
\label{E:C3}
&\int_\mathcal{X} \Big(\big[ \P{U  \le t + \delta_0(x) |W = 0, X = x} - \P{U \le  t| W = 0, X = x}\big] f_0(x)\\ 
& \qquad - \big[ \P{U \le t + \delta_0(x) | W = 1 , X = x}  - \P{U \le t| W = 1, X = x} \big] f_1(x)\Big) \, \mathrm{d}x = 0, \nonumber
\end{align}
where we have defined $\delta_0 \equiv m_\delta - \delta(x) \ge 0$. As $\norm{\delta}_\infty \le B$, $|\delta_0| \le 2B$ follows from the triangle inequality.

Now, one may write 
\begin{align*}
&\P{U \le t + \delta_0(x)| W = w, X = x } - \P{U \le t| W = w, X = x}\\
 & \qquad = \frac{1}{f_w(x)}\int_0^{\delta_0(x)} f_w(x, t + u) \, \mathrm{d} u 
\end{align*}
so that \eqref{E:C3} becomes 
\begin{align*}
T(\one_{u \in [0,\delta_0(x)]})(t) &= \int_\mathcal{X} \int_0^{2B} \one_{u \in [0,\delta_0(x)]} (f_0(x,t + u)   - f_1(x, t + u))\, \mathrm{d} u \,  \mathrm{d}x\\
& = \int_\mathcal{X} \int_0^{\delta_0(x)} (f_0(x,t + u)   - f_1(x, t + u))\, \mathrm{d} u \,  \mathrm{d}x \\
&  = 0,
\end{align*}
for all $t \in \R$. As $\one_{u \in [0,\delta_0(x)]} \in \mathcal{W}$, we have produced the desired contradiction to Assumption \ref{A:C3}(i). It follows that our stated conditions are sufficient for point identification of $g$ in $\mathcal{G}$. 

Conversely, suppose that Assumption \ref{A:C3}(i) does not hold so that there is some nonconstant $\delta(x) \in C(\mathcal{X})_+$ such that $\one_{u \in [0,\delta(x)}] \in \ker(T)$. Then by replicating our previous arguments in reverse, one deduces that \eqref{E:C3} holds with $\delta_0$ replaced by $\delta$. Then by independence of $U$ and $W$, in fact $\P{U + \delta(X) \le t | W = 0} = \P{ U + \delta(X) \le t | W = 0} = 0$ for all $t \in \R$, and one has $U + \delta(X) - \E{\delta(X)} \indep W$. Hence, letting $h(X) = g(X) - \delta(X) + \E{\delta(X)}$, one has 
\begin{align*}
&Y - h(X) \indep  W \\
&\E{Y - h(X) } = 0,
\end{align*}
which implies that $g$ is not point identified. So Assumption \ref{A:C3} is both necessary and sufficient for point identification of $g$ under our other stated assumptions and regularity conditions. 
\end{proof}

\begin{proof}[Proof of Lemma \ref{L:compeleteinstrument}]
The conditional distribution of $U, U+V$ has density $(2B)^{-1}f_U(u) \one_{s - u \in [0,2 B]}$, so that the conditional distribution of $U$ conditioning on the event $U + V = s$ is $\frac{f_U(u) \one_{s - u \in [0,2B]} }{ \int_{s}^{s-2B} f_U(u') \, \mathrm{d} u' }$. Hence, under assumptions \ref{A:C1} and \ref{A:C2}, for any rectangle $R = I \times J \subset \mathcal{X}\times \R$, one has 
\begin{align} \label{E:C5}
\P{(X,U) \in R | U + V = s} & = \E{ \one_{X \in I} \one_{U \in J} |  U + V  = s} \\
& = \int_{ s-2B}^{s} \frac{\E{ \one_{X \in I } \one_{U \in J} | U = u, U + V = s }}{ \int_{ s-2B}^{s} f_U(u') \, \mathrm{d}u'  }f_U(u) \, \mathrm{d} u \nonumber \\
& =\frac{1}{\P{U \in [s-2B,s]}} \int_{ s -2B}^{s} \one_{u \in J } \E{ \one_{X \in I } | U = u} f_U(u) \, \mathrm{d} u  \nonumber \\
& =\frac{1}{\P{U \in [s-2B,s]}} \int_{s-2B}^{s} \int_{\mathcal{X}} \one_{u \in J } \one_{x \in I} \frac{f(x,u)}{f_U(u)} f_U(u) \, \mathrm{d}x\, \mathrm{d} u  \nonumber \\
& = \frac{1}{\P{U \in [s-2B,s]}} \int_{s - 2B}^{s} \int_{\mathcal{X}} \one_{(u,x) \in R} f(x,u) \, \mathrm{d} x \, \mathrm{d} u;  \nonumber
\end{align}
furthermore, application of the Lebesgue differentiation theorem implies that the relation above holds for any measurable $R \subset \mathcal{X} \times \R$, so that conditional upon $U + V  = s$, $(X,U)$ has a distribution with density $\frac{f(x,u) \one_{u \in [s-2B, s]}}{\gamma(s)}$, where $\gamma(s)$ is an appropriate constant which is defined if $\P{U \in [s-2B, s]} > 0$. Replacing $s = t + 2B$ in \eqref{E:C5} and passing to expectations, one has 
\begin{align}
\label{E:C6} 
\E{ h(X,U -t ) | U + V  = t + 2B } & = \frac{1}{\P{U \in [t, t + 2B]} } \int_t^{t + 2B} \int_{\mathcal{X}} h(x,u -t) f(x,u) \, \mathrm{d}x \, \mathrm{d} u  \\
& = \frac{1}{\P{U \in [t, t + 2B]}} \E{ h(X, U - t) \one_{U \in [t, 2B + t] }  } \nonumber
\end{align}
whenever $\P{U \in [t,t+ 2B]} > 0$. 
Now, substituting $f(x,u)$ with $f_w(x,u)$ in \eqref{E:C6} and noting that $\P{U \in [t, t+2B]| W = 0} = \P{U \in [t,t+2B] | W = 1} = \P{U \in [t, t+2B]}$ for all $t$ by independence of $U$ and $W$, we have: 
\begin{align*}
Th(t) = & \E{ h(X, U - t) \one_{U \in [t, t + 2B]} | W = 0} -  \E{ h (X, U - t) \one_{U \in [t, t + 2B]}|W = 1} \\
=& \P{U \in [t, t + 2B]} \\
 &  \cdot \big(\E{ h(X,2B - V)| U + V = t + 2B, W = 0} - \E{ h(X,2B - V)| U + V = t + 2B, W = 1}  \big),
\end{align*}
where we have used the convention that $\frac{0}{0}= 0$. Letting $t$ vary over $\R$, Assumption \ref{A:C3}(ii) is thus the condition that for any $h \not\in \mathcal{V}$, 
\[
\E{ h(X,2B - V)| U + V = t , W = 0} - \E{ h(X,2B - V)| U + V = t , W = 1} \neq 0
\]
for some $t$ such that $\P{U \in [t, t + 2B]} > 0$ (i.e.\ such that the density of $U +V$, which is easily seen to be a continuous function, is positive at $t$). Conclude by rewriting substituting $\tilde{V} = 2B - V$ and rewriting the expectation. 
\end{proof}

\begin{proof}[Proof of Proposition \ref{P:densityapprox}]
For simplicity, we consider first the case where $B < \infty$. 
Fix a function $\gamma \in \Gamma$ and $\ve \in (0,1)$. We will approximate $\gamma$ with a function $\gamma_\ve \in \Gamma_0$ satisfying $\norm{ \gamma - \gamma_\ve}_{L^1(\X \times \R )}< 5\ve$. Recall that as $\gamma$ is the difference of continuous probability density functions, one has $\norm{\gamma}_{L^1(\X \times \R )} \le 2$. Furthermore, there exists some $K \in \R$ such that $\norm{\gamma}_{L^1(\X \times \R )}- \ve/4 < \norm{\gamma \one_{|u| \le K} }_{L^1(\X \times \R  )} $. For $j \in \N$ let $e_{2j -1}(x,u)$ be a sequence of continuous orthonormal basis functions for $L^2(\X \times [0,2B] )$ whose closed linear span is $L^2(\X \times [0,2B] )$ (e.g.\ polynomials). Furthermore, set $\tilde{\gamma}_0$ to be a continuous function on $\X \times [-K,K]$ such that $\norm{\gamma - \tilde{\gamma}_0}_{L^1(\X \times \R )} < \ve/4$ and $\int_\X \tilde{\gamma_0}(x,u) \, \mathrm{d}x = 0$ for $u \in [-K,K]$ (using the fact that $\int_\X \gamma(x,u) \, \mathrm{d}x = 0$ for all $u$). Then let $e_{2j - 2}, j \in \N$ be a sequence in $L^2(\X \times [0,2B] )$ chosen such that the function 
\begin{align*}
\tilde{\gamma}(x,u) \equiv \tilde{\gamma}_0(x,u) \one_{|u| \le K} + \sum_{j=0}^\infty e_j (x,u - K - j(2B)) \one_{u \in [K + j(2B), K + (j+1)(2B)]}
\end{align*} 
is continuous (i.e.\ a continuous interpolation between $e_{2j-3}$ and $e_{2j-1}$); it is no challenge to ensure that for every such $j$ one has $\norm{e_{2j-2}}_{L^2(\X \times [0,2B] )} \le 1$, so we make this assumption. Now let $\psi: \R \ra (0,1]$ be a continuous function such that $\psi(u) = 1$ whenever $|u| \le K$ and 
\[
\norm{ \psi(u) \one_{u \in [K+j(2B), K + (j+1)(2B)]  }}_{L^2(\X \times \R  )} \le \frac{\ve}{2^{j+2}}
\]
for all $j \ge 0$ (making use of the fact that $\X$ is compact). Finally, noting that $\mathcal{V}$ is a closed linear subspace of $L^2(\X \times [0,2B] )$, let $P_\mathcal{V}: L^2(\X \times [0,2B] ) \ra L^2(\X \times [0,2B] )$ denote projection onto $\mathcal{V}$. Then $P_\mathcal{V}$ can be written explicitly as 
\begin{align*}
P_\mathcal{V}[h](x,u) = \lambda(\X)^{-1} \int_{\X} h(x', u) \, \mathrm{d}x'. 
\end{align*}
Finally, let $\rho$ be a continuous probability density function supported on $\X \times [-K,K]$, and $\alpha \equiv \int_\X \int_\R (I - P_\mathcal{V})\tilde{\gamma}(x,u) \psi(u) \, \mathrm{d}u \, \mathrm{d}x$ (the integral exists, as we show below). Furthermore, let $\kappa = \min\left\{ 1, 2/\norm{(I - P_\mathcal{V})\tilde{\gamma}(x,u) \psi(u) - \alpha \rho   }_{L^1(\X \times \R )} \right\}$. 
We let $\gamma_\ve(x,u) \equiv \kappa[(I - P_\mathcal{V})\tilde{\gamma}(x,u) \psi(u) - \alpha \rho] = \kappa[ \psi(u) (I - P_\mathcal{V}) \tilde{\gamma}(x,u) - \alpha \rho]$. Note that by the independence assumption $U \indep W$ in the definition of $\Gamma$, it is true that $P_\mathcal{V}(\gamma) = 0$. By the Cauchy-Schwarz and triangle inequalities,
\small
\begin{align*}
\norm{(I - P_\mathcal{V})\tilde{\gamma} \psi - \gamma}_{L^1(\X \times \R )} = & \int_\X \int_\R |\gamma_\ve(x,u) - \gamma(x,u)| \, \mathrm{d} x \, \mathrm{d} u  \\
= & \int_\X \int_{-K}^K |\gamma_\ve (x,u) - \gamma(x,u)| \, \mathrm{d}x \, \mathrm{d} u \\
& + \sum_{j = 0}^\infty \int_\X \int_{K + j(2B)}^{K + (j+1)(2B)}  |\psi(u)(I-P_\mathcal{V}) e_j(x,u - K - j(2B)) - \gamma(x,u)| \, \mathrm{d}u \, \mathrm{d}x \\
& + \int_\X \int_{-\infty}^{-K} |\gamma(x,u)| \, \mathrm{d}u \, \mathrm{d} x \\
\le &  \ve/2 + \sum_{j=0}^\infty \norm{ \psi(u)(I- P_\mathcal{V}) e_j(x, u - K - j(2B)) \one_{u \in [K + j(2B), K + (j+1) (2B)]} }_{L^1(\X \times \R )} \\
\le & \ve/2 +  \sum_{j=0}^\infty \norm{\psi(u)}_{L^2(\X \times \R )} \norm{ (I-P_\mathcal{V}) e_j(x,u) }_{L^2(\X \times [0,2B]} \\
\le & \ve/2 + \sum_{j=0}^\infty \norm{\psi(u)}_{L^2(\X \times \R )} < \ve. 
\end{align*}
\normalsize
Noting that $\int_\X \int_\R \gamma \, \mathrm{d}u \, \mathrm{d}x =0$, one has $|\alpha| < \ve$ and hence the triangle inequality implies 
\[
\norm{((I - P_\mathcal{V}) \tilde{\gamma} \psi - \alpha \rho) - \gamma}_{L^1(\X \times \R )} < 2 \ve,
\]
as desired. Hence $\norm{((I - P_\mathcal{V}) \tilde{\gamma} \psi}_{L^1(\X \times \R )} < 2 + 2\ve$ and $1 - \ve < \kappa \le 1$, so a series of straightforward calculations shows $\norm{\gamma_\ve - \gamma}_{L^1(\X \times \R )} <5\ve$. Moreover, $\int_\X \int_\R \gamma_\ve \, \mathrm{d}u \, \mathrm{d}x = 0$ and by arrangement $\norm{\gamma_\ve}_{L^1(\X \times \R )} \le 2$. 

Now note that for any $h \in L^2(\X \times [0,2B] )$ if one has $T_{\psi \tilde{\gamma} }h = 0$ then letting $t = K + j(2B)$, $j = 0, 1,2,\ldots$ in the definition of $T$ implies that 
\begin{align*}
\inner{ \psi e_j, h}_{L^2(\X \times [0,2B] )} = \int_\X \int_0^{2B}\psi(u) e_j (x, u) h(x,u) \, \mathrm{d}u \, \mathrm{d}x  = 0,
\end{align*}
for all $j$. Because $\{e_j(x,u)\}_{j \ge 0}$ contains an orthonormal basis, $\psi(u) h(x,u) = 0$ (in the $L^2$ norm), and so $h(x,u) = 0$ (as $\psi(u) \neq 0$). We claim now that when $T_{\gamma_\ve}$ is a viewed as an operator from $L^2(\X \times[0,2B] )$ to $C(\R)$ one has $\ker{T_{\gamma_\ve}} = \mathcal{V}$. For if $T_{\gamma_\ve} h = 0$ then for all $j$, 
\begin{align} \label{E:innerprod}
\inner{\psi e_j, (I - P_\mathcal{V}) h}_{L^2(\X \times [0,2B] )} & = \inner{ (I - P_\mathcal{V}) \psi e_j, h}_{L^2(\X \times [0,2B] )} = T_{\gamma_\ve}[ h] (K + j (2B)) = 0. 
\end{align}
This implies $(I - P_\mathcal{V}) h = 0$, which is true if and only if $h \in \mathcal{V}$. This establishes the claim. 

Finally we must show that $\gamma_\ve \in \Gamma$, which is to say that $\gamma_\ve = f_0^\ve - f_1^\ve$ for continuous density functions $f_0^\ve$ and $f_1^\ve$. It is easy to verify that the following specifications suffice:
\begin{align*}
&f_0^\ve \equiv \gamma_\ve^+ + \left(1 - \int_\X \int_\R \gamma_\ve^+ \, \mathrm{d} u \, \mathrm{d}x \right) \rho\\
&f_1^\ve \equiv \gamma_\ve^- + \left(1 - \int_\X \int_\R \gamma_\ve^+ \, \mathrm{d} u \, \mathrm{d}x  \right) \rho.
\end{align*}

For the case where $B = \infty$ (which is to say that $\delta$ is not necessarily bounded by any fixed constant in $\R$), the idea of proof is simply to take $\{e_{2j-1}(x,u)\}_{j = 0}^\infty$ to be a dense collection of continuous functions with bounded support in $L^2(\X \times \R )$, and $\{e_{2j-2}\}_{j=0}^\infty$ a collection continuous interpolations between them. Then one quickly verifies that \eqref{E:innerprod} still holds when $L^2(\X \times [0,2B] )$ is replaced with $L^2(\X \times \R )$, and the remainder of the proof is entirely similar to the case $B < \infty$. 
\end{proof}

\begin{proof}[Proof of Corollary \ref{C:densityapprox}]
Assume $\ve < 1$. 
By supposition, $\gamma = f_0 - f_1 \in \Gamma$. By Proposition \ref{P:densityapprox}, there exists $\gamma_\ve \in \Gamma_0$ such that $\norm{\gamma_\ve - \gamma}_{L^1(\X \times \R )} < \ve/4$. Let $\zeta \equiv \gamma_\ve - \gamma$ and then 
\begin{align*}
&f_0^\ve \equiv (1 + \norm{\rho}_{L^1(\X \times \R)} + \norm{\zeta^+}_{L^1(\X \times \R)})^{-1}\left(f_0 + \zeta^+  + \rho \right) \\
&f_1^\ve \equiv (1 + \norm{\rho}_{L^1(\X \times \R)} + \norm{\zeta^+}_{L^1(\X \times \R)})^{-1} \left(f_1 + \zeta^-  + \rho\right),
\end{align*}
where $\rho$ is a smooth function on $\X \times \R$ chosen to satisfy $\norm{\rho}_{L^1(\X \times \R)} < \ve/4$ and $\int_\R \int_\X u f_0^\ve(x,u) \, \mathrm{d} x \, \mathrm{d} u = 0$. 
Then for $\ve$ sufficiently small (which may be assumed), 
\begin{align*}
\norm{f_0 - f_0^\ve}_{L^1(\X \times \R )}  \le 4/3\left( \norm{\zeta^+}_{L^1(\X \times \R )} + \norm{\zeta^+}_{L^1(\X \times \R )} \norm{f_0}_{L^1(\X \times \R )} + \norm{\rho}_{L^1(\X \times \R)}\right) \le \ve,
\end{align*}
as $\norm{\zeta^+}_{L^1(\X \times \R )} \le \norm{\zeta}_{L^1(\X \times \R )}$. 
Similar argumentation for $f_1^\ve$ suffices, with the observation that $f_0^\ve - f_1^\ve$ is a scalar multiple of $\gamma_\ve$.

\end{proof}

\begin{proof}[Proof of Proposition \ref{P:resid}]
It suffices to show that $\Gamma_1^c \cap \Gamma$ is meagre in the induced topology, i.e.\ contained in the countable union of closed nowhere dense sets. Note that $\Gamma_1^c$ consists of those elements in $\Gamma$ for which $\ker{T_\gamma} \cap \mathcal{W}_\mathrm{Lip} \neq \emptyset$. One can write $\mathcal{W}_{\mathrm{Lip}} = \bigcup_{a,b \in \N} \mathcal{W}_{a,b}$, where 
\begin{align*}
\mathcal{W}_{a,b} \equiv \left\{ \delta \in C(\X)_+: \, \norm{\delta}_{\mathrm{Lip}} \le a, \, \norm{\delta}_\infty \le b, \, \inf_{c \in \R}\norm{\delta - c }_\infty \ge b^{-1} \right\},
\end{align*}
and we have used the notations 
\begin{align*}
&\norm{h}_{\mathrm{Lip}} \equiv \sup_{x,y \in \X} \frac{|h(x) - h(y)|}{|x - y|} \\
&\norm{h}_\infty \equiv \sup_{x \in X} |h(x)|.
\end{align*}
Let $\Gamma_{a,b} \equiv \left\{ \gamma \in \Gamma: \ker{T_\gamma} \cap \mathcal{W}_{a,b}\neq \emptyset \right\}$, so that 
\[
\Gamma_1^c = \bigcup_{a,b \in \N} \Gamma_{a,b}.
\]
We show that $\Gamma_{a,b}$ is a closed set for all $a,b$. So suppose that $(\gamma_n)_{n \in \N}$ is sequence occurring in $\Gamma_{a,b}$ and converging to $\gamma$ in the $L^1$ norm. There is a matching sequence $(\delta_n)_{n \in \N}$ occurring in $\mathcal{W}_{a,b}$ such that $T_{\gamma_n} \one_{u \in [0,\delta_n(x)]}  = 0$. By the Arzel\`{a}-Ascoli Theorem, there is a function $\delta$ such that $\norm{\delta_n - \delta}_\infty \ra 0$; it is straightforward to verify then that $\delta \in \mathcal{W}_{a,b}$. Thus it suffices to show that $T_\gamma \one_{u \in [0, \delta(x)]} = 0$. But, for all $t \in \R$,  
\begin{align*}
\left| T_\gamma \one_{u \in [0,\delta]}(t) - T_{\gamma_n} \one_{ u \in [0,\delta_n]}(t) \right|  \le& \int_\X \int_0^{2B} \left| \one_{u \in [0,\delta(x)]} \gamma(x,u+t) - \one_{u \in [0,\delta_n(x)]} \gamma_n(x,u+t) \right| \, \mathrm{d}u \, \mathrm{d} x \\
\le & \int_\X \int_0^{2B}\left| \one_{u \in [0,\delta(x)]} \gamma(x,u+t) - \one_{u \in [0,\delta_n(x)]} \gamma(x,u+t) \right| \, \mathrm{d}u \, \mathrm{d} x \\
&+ \int_\X \int_0^{2B}\left| \one_{u \in [0,\delta_n(x)]} \gamma(x,u+t) - \one_{u \in [0,\delta_n(x)]} \gamma_n(x,u+t) \right| \, \mathrm{d}u \, \mathrm{d} x \\
\le&  \int_\X \int_0^{2B} \one_{u \in [\delta(x)- \norm{\delta - \delta_n}_\infty, \delta(x) +\norm{\delta- \delta_n}_\infty]} |\gamma(x,u+t)| \, \mathrm{d}u \, \mathrm{d}x\\
&+ \norm{\gamma - \gamma_n}_{L^1(\X \times \R )}.
\end{align*}
As $n \ra \infty$ the Dominated Convergence Theorem implies that the second to last line converges to $0$ and by assumption $\norm{\gamma - \gamma_n}_{L^1(\X \times \R )} \ra 0$. Hence, $T_\gamma \one_{u \in [0,\delta](t)} = 0$ and $\Gamma_{a,b}$ is closed. But note that $\Gamma_{a,b}$ is also nowhere dense, because Proposition \ref{P:densityapprox} implies that $\Gamma_{a,b}^c$ is dense in $\Gamma$ (take the case $B = \infty$). 

\end{proof}

\begin{proof}[Proof of Corollary \ref{C:densityapprox2}]
Let $\tau: \mathfrak{X} \times \mathfrak{X} \ra L^1(\X \times \R )$ be defined by $\tau: (f_0, f_1) \mapsto f_0 - f_1$. The triangle inequality implies that $\tau$ is continuous and one has $\tau^{-1}(\Gamma) = \mathfrak{F}$, so that $\mathfrak{F}$ is closed in $\mathfrak{X} \times \mathfrak{X}$ and completely metrizable. Let $\tau_0$ denote the restriction of $\tau$ to $\mathfrak{F}$, with the induced topology (so that $\tau_0$ is still continuous). Then Proposition \ref{P:resid} implies that, for some collection of dense and open subsets $\{G_a\}_{a \in \N}$ of $\Gamma$,
\begin{align*}
\mathfrak{F}_1 &= \tau^{-1} \left(  \Gamma_1 \right)
 \supset \tau^{-1} \left( \bigcap_{a \in \N} G_a\right) = \bigcap_{a \in \N} \tau^{-1}(G_a),
\end{align*}
so that $\mathfrak{F}_1$ contains a $G_\delta$ set. In addition, density of each $G_a$ in $\Gamma$ and arguments similar to those employed in the proof of Corollary \ref{C:densityapprox} implies that each $G_a$ is dense, which shows that $\mathfrak{F}_1$ is residual in the induced topology. 
\end{proof}

\begin{proof}[Proof of Corollary \ref{C:boundedsppt}]
Apply Corollary \ref{C:densityapprox} to the densities $f_0, f_1$ to conclude that there are continuous approximations $f_0', f_1'$ with unbounded support such that $\norm{f_\ell' - f_\ell}_{L^1(\X \times \R)} < \ve/16$ for $\ell = 0, 1$ and $f_0' - f_1' \in \Gamma_0$, $\int_\R \int_\X u f_0'(x,u) \, \mathrm{d} x \, \mathrm{d}u = 0$. Multiply both $f_0'$ and $f_1'$ by a suitable factor $\chi(u)$ to assume without loss of generality that both are elements of $L^2(\X \times \R)$ (see the proof of Proposition \ref{P:densityapprox}). Now let $\kappa$ be a smooth, square integrable probability density function over $\R$ which has the property that $\kappa$ is the restriction of a complex analytic function to the real line, and this holomorphic function has bounded complex derivative in some horizontal strip containing the real line, $\{z: -c < \mathfrak{Im}(z) < c\}$: the Gaussian kernel $\frac{1}{\sqrt{2\pi}} \exp{ - z^2/2}$ suffices. For real $\delta$ let $\kappa_\delta(z) \equiv \delta^{-1}\kappa(z \delta^{-1})$. Then for $\ell \in \{0,1\}$ and $x \in \X$ the mollification: 
\begin{align*}
f_0^\delta(x,u) \equiv f_0'(x,\cdot) * \kappa_\delta (u)= \int_\R f_0'(x,u') \kappa_\delta(u - u') \, \mathrm{d} u'
\end{align*}
is by virtue of the dominated convergence theorem holomorphic on a horizontal strip containing the real line when viewed as a function of $u$, and standard results imply that for $\delta$ small enough one has $\norm{f_\ell^\delta - f_\ell'}_{L^1(\X \times \R)} < \ve/16$. Now note that by compactness of $\X$ and application once again of the dominated convergence theorem, 
\begin{align} \label{E:convholo}
T_{f_0^\delta - f_1^\delta}h(t) = \int_\X \int_0^{2B} h(x,u) \left(f_0^\delta(x, t + u) - f_1^\delta(x, t + u )\right) \, \mathrm{d} u\, \mathrm{d} x 
\end{align}
extends to a holomorphic function in $t$ in the horizontal strip for $h \in L^2(\X \times [0,2B]) \subset L^1(\X \times [0,2B])$. Hence, $T_{f_0^\delta - f_1^\delta}h = 0$ if and only if \eqref{E:convholo} vanishes on a subset of $\R$ which contains one of its accumulation points. But for $t \in [-C_1, C_2  - 2B]$, we have
\small
\begin{align*}
&\int_\X \int_0^{2B} h(x,u) \left(f_0^\delta(x, t + u) - f_1^\delta(x, t + u )\right)\one_{t + u \in [-C_1, C_2]} \, \mathrm{d} u\, \mathrm{d} x \\
& \qquad =  \int_\X \int_0^{2B} h(x,u) \left(f_0^\delta(x, t + u) - f_1^\delta(x, t + u )\right) \, \mathrm{d} u\, \mathrm{d} x. 
\end{align*}
\normalsize

Accordingly, define the density functions $f_\ell''(x,u) \equiv f_\ell^\delta(x,u) \one_{u \in [-C_1, C_2]}$ and note that we have the bound $\norm{f_\ell'' - f_\ell^\delta}_{L^1(\X \times \R)} < \ve/ 8$ (the total mass of $f_\ell^\delta$ outside $\X \times [ -C_1, C_2]$ is in fact bounded by $\ve/ 4$). We have shown that $T_{f_0'' - f_1''} h = 0$ implies $T_{f_0^\delta - f_1^\delta} h (t)$ for $t \in [-C_1, C_2 - 2B]$, which implies $T_{f_0^\delta - f_1^\delta} h = 0$. The latter implies that for all $t\in \R$:
\small
\begin{align*}
&\kappa_\delta * \int_\X \int_\R h(x,u)( f_0'(x, \cdot + u ) - f_1'(x,\cdot + u)) \, \mathrm{d}u \, \mathrm{d} x\, (t) \\
&\qquad = \int_\R \kappa_\delta(t - z) \int_\X \int_\R h(x,u)(f_0'(x, z + u) - f_1'(x,z + u))  \, \mathrm{d} u \, \mathrm{d} x \, \mathrm{d} z  \\
&\qquad = \int_\X \int_\R f_0'(x,\cdot) * \kappa_\delta(u + t) h(x,u) \, \mathrm{d}u \, \mathrm{d} x = 0. 
\end{align*}
\normalsize
Lemma \ref{L:operator} guarantees that $\int_\X \int_\R h(x,u) f_0'(x, \cdot + u ) \, \mathrm{d}u \, \mathrm{d} x$ is continuous and in $L^2(\R)$ by taking Fourier transforms of both sides and applying Plancherel's theorem we find that $\int_\X \int_\R h(x,u) f_0'(x, \cdot + u ) \, \mathrm{d}u \, \mathrm{d} x  = 0$. Hence, by construction of $f_0'$ and $f_1'$, $h(x,u) \in \mathcal{V}$ and we have shown that the inclusion $\ker{T_{(f_0'' - f_1'')\one_{u \in [-C_1, C_2]}}} \subset \mathcal{V}$ holds, as desired. Note also that for all $u\in [-C_1, C_2],$
\begin{align*}
\int_\X (f_0''(x,u) - f_1''(x,u)) \, \mathrm{d} x  = \kappa_\delta*\int_\X (f_0'(x,\cdot) - f_1'(x,\cdot))\, \mathrm{d} x (u) = 0,
\end{align*}
so $f_0'' - f_1'' \in \Gamma$, as desired (note: $\norm{f_0'' - f_1''}_{L^1(\X \times \R)} < 2$ because of multiplication of the indicator $\one_{u \in [-C_1, C_2]}$. The remainder of the proof consists in setting 
\begin{align*}
f_0^\ve & = \norm{f_0''+ \rho}_{L^1(\X \times [-C_1, C_2])} (f_0'' + \rho )\\
f_1^ \ve & = \norm{f_1''+ \rho}_{L^1(\X \times [-C_1, C_2])} (f_0'' + \rho )
\end{align*}
where $\rho$ is a probability density chosen to satisfy the last condition $\int_\R \int_\X u f_0^\ve (x,u) \, \mathrm{d} x \, \mathrm{d} u = 0$ and has mass less than $\frac{\ve}{4}$ for $\ve$ small enough, similarly to the corresponding function in the proof of Corollary \ref{C:densityapprox}. 
\end{proof}

\begin{proof}[Proof of Proposition \ref{P:quantile}]
Suppose first that $W$ is boundedly complete for $X,U$. If $g': \supp{U} \ra \R$ is in the identified set then we must have the relation 
\begin{align*}
\P{U + g(X) - g'(X) \le 0|W } = \P{ Y - g'(X) \le 0 | W } \eqas \gamma(W) = \P{U \le 0 |W }.
\end{align*}
Letting $\delta = g - g'$ this implies 
\begin{align*}
\E{ \one_{U + \delta(X) \le 0 }  - \one_{U \le 0} |W } = 0 . 
\end{align*}
By bounded completeness it follows that $\delta(X) \eqas 0$, which suffices. Alternately, if $\supp{U} = \R$ then we claim $\delta$ is nonconstant; if for the sake of contradiction $g' = g -\delta$, $\delta$ some nonzero constant, was in the identified set, then
\begin{align*}
\P{U + \delta \le 0|W} = \P{Y - g'(X) \le 0 |W } \eqas \gamma(W) \eqas \P{U \le 0|W}, 
\end{align*}
so that $\P{U  \le -\delta } = \P{U\le 0}$ and the probability of $U$ lying between $-\delta$ and $0$ vanishes. Hence we may repeat the proof used above under the presumption that $\delta$ is nonconstant in $X$, and thus so is $\one_{U + \delta(X) \le 0} - \one_{U \le 0 }$, which concludes. 
\end{proof}

\bibliographystyle{plain}
\bibliography{MKbib}

\end{document}